\documentclass[12pt]{article}

\newcommand{\papertitle}{Semiparametric Inference for Dynamic Causal Effects from Observational Time Series}
\newcommand{\authorlist}{Shibo Yu, Yan Chen, Jin-Hong Du, Guodong Li}

\usepackage{iftex}
\ifPDFTeX
  \usepackage[T1]{fontenc}
\fi
\usepackage{natbib}
\usepackage[margin=1.0in]{geometry}
\usepackage{titling}
\usepackage{authblk} 
\usepackage{mathtools}
\usepackage{amsfonts}
\usepackage{amsmath,amssymb,amsthm}
\usepackage[mathscr]{euscript}
\allowdisplaybreaks

\usepackage{multirow,multicol}
\usepackage{graphicx,color}
\usepackage{comment}
\usepackage{pifont} 
\usepackage{textcomp}
\usepackage{makecell, booktabs}
\usepackage{caption}
\usepackage{tabularx}
\usepackage{array}
\usepackage{titlesec}
\usepackage{algorithm}
\usepackage{algorithmic}
\usepackage{float}
\usepackage{subfig}
\usepackage[english]{babel}
\usepackage[toc,page]{appendix}

\newcommand{\romannum}[1]{\romannumeral #1}
\DeclareRobustCommand{\Romannum}[1]{\MakeUppercase{\romannum{#1}}}

\newcommand{\bm}{\boldsymbol}

\newcommand{\Fr}{{\mathrm{F}}}
\newcommand{\op}{{\mathrm{op}}}
\DeclareMathOperator{\cov}{cov}
\newcommand{\var}{{\mathrm{var}}}
\newcommand{\E}{\mathbb{E}}

\newcommand{\R}{\mathbb{R}}
\newcommand{\I}{\mathbf{1}} 

\definecolor{turquoise}{rgb}{0.03, 0.7, 0.87}

\newcommand{\sgnorm}{{\psi_2}}
\newcommand{\senorm}{{\psi_1}}
\newcommand{\dto}{\xrightarrow{d}} 
\newcommand{\pto}{\xrightarrow{p}} 
\DeclareMathOperator*{\vectorize}{vec}
\DeclareMathOperator*{\rank}{rank}
\DeclareMathOperator*{\trace}{tr}
\DeclareMathOperator*{\argmin}{arg\,min}

\newtheorem{assumption}{Assumption}
\newtheorem{condition}{Condition}
\newtheorem{example}{Example}

\newtheorem{lemma}{Lemma}
\newtheorem{proposition}{Proposition}
\newtheorem{theorem}{Theorem}
\newtheorem{corollary}{Corollary}

\newtheorem{remark}{Remark}

\newcommand{\setupcleverefnames}{
	\crefname{assumption}{Assumption}{Assumptions}
	\Crefname{assumption}{Assumption}{Assumptions}
	\crefname{condition}{Condition}{Conditions}
	\Crefname{condition}{Condition}{Conditions}
	\crefname{example}{Example}{Examples}
	\Crefname{example}{Example}{Examples}
	\crefname{definition}{Definition}{Definitions}
	\Crefname{definition}{Definition}{Definitions}
	\crefname{figure}{Figure}{Figures}
	\Crefname{figure}{Figure}{Figures}
	\crefname{lemma}{Lemma}{Lemmas}
	\Crefname{lemma}{Lemma}{Lemmas}
	\crefname{proposition}{Proposition}{Propositions}
	\Crefname{proposition}{Proposition}{Propositions}
	\crefname{theorem}{Theorem}{Theorems}
	\Crefname{theorem}{Theorem}{Theorems}
	\crefname{corollary}{Corollary}{Corollaries}
	\Crefname{corollary}{Corollary}{Corollaries}
	\crefname{remark}{Remark}{Remarks}
	\Crefname{remark}{Remark}{Remarks}
	\crefname{algorithm}{Algorithm}{Algorithms}
	\Crefname{algorithm}{Algorithm}{Algorithms}
}

\usepackage[colorlinks,linkcolor=blue,citecolor=blue]{hyperref}
\usepackage[noabbrev]{cleveref}
\setupcleverefnames
\usepackage{bibunits}
\defaultbibliographystyle{apalike}
\defaultbibliography{dynamicIV}

\makeatletter
\let\arxiv@mainfloat\@xfloat
\let\arxiv@mainfootnote\@footnotetext
\let\arxiv@mainmpfootnote\@mpfootnotetext
\newtoks\arxiv@maindisplay
\arxiv@maindisplay=\expandafter{\the\everydisplay}
\makeatother
\usepackage{setspace}
\makeatletter
\let\arxiv@suppfloat\@xfloat
\let\arxiv@suppfootnote\@footnotetext
\let\arxiv@suppmpfootnote\@mpfootnotetext
\newtoks\arxiv@suppdisplay
\arxiv@suppdisplay=\expandafter{\the\everydisplay}
\newcommand{\mainspacing}{%
  \let\@xfloat\arxiv@mainfloat
  \let\@footnotetext\arxiv@mainfootnote
  \let\@mpfootnotetext\arxiv@mainmpfootnote
  \everydisplay=\expandafter{\the\arxiv@maindisplay}}
\newcommand{\supplementspacing}{%
  \let\caption@ORI@xfloat\arxiv@suppfloat
  \let\@footnotetext\arxiv@suppfootnote
  \let\@mpfootnotetext\arxiv@suppmpfootnote
  \everydisplay=\expandafter{\the\arxiv@suppdisplay}}
\mainspacing

\newcommand{\bibliographypart}[1]{%
  \def\@extra@b@citeb{.#1}%
  \def\@extra@binfo{.#1}}

\let\arxiv@pageanchor\Hy@EveryPageAnchor
\renewcommand{\Hy@EveryPageAnchor}{%
  \begingroup\Hy@hypertexnamesfalse\arxiv@pageanchor\endgroup}
\makeatother

\numberwithin{equation}{section}
\title{\papertitle}
\author{
	\centering
	\authorlist
	\thanks{Address for correspondence: Department of Statistics and Actuarial Science, The University of Hong Kong, Hong Kong, China.}
}
\affil{Department of Statistics and Actuarial Science, The University of Hong Kong}
\date{\today}

\begin{document}
\begin{bibunit}
\bibliographypart{main}
\captionsetup{font={stretch=1.667}}
\renewcommand{\paragraph}{\textbf}
\setlength{\droptitle}{-4em}
\maketitle
\vspace{-3em}

\begin{abstract}
    In observational time series, statistical inference for dynamic causal effects of a one-time intervention across horizons is complicated by high-dimensional observed pre-treatment information, unmeasured confounding, and serial dependence.
    To address these challenges, we develop a semiparametric framework for inference from a single serially dependent time series, integrating debiased machine learning with instrumental variables through buffered block cross-fitting.
    Under geometric beta-mixing, we derive non-asymptotic bounds on estimation error, asymptotic normality at each fixed horizon, and feasible inference that accommodates serial dependence.
    We further show how learner-specific prediction guarantees under temporal dependence can be used to verify the nuisance-rate conditions required for orthogonal inference.
    In a monetary-policy application with 468 months and 1464 lagged FRED-MD controls, we show an instrumented policy tightening lowers housing starts at medium horizons, with sensitivity analyses that support the finding.
\end{abstract}

\section{Introduction} \label{sec:intro}
Many scientists and policymakers ask how an outcome trajectory would have changed if, at a specific time, a treatment or policy variable had been shifted by an external intervention while the pre-treatment history was held fixed \citep{adamek2024local}. For example, a central bank cares about how inflation and real activity would respond over the following months to a change in its policy rate \citep{romer2004new,stock2018identification}, a fiscal authority considers how a tax change propagates through the economy \citep{mertens2013dynamic}, and a security analyst investigates whether airstrikes at one date raise or lower insurgent violence in the weeks that follow \citep{papadogeorgou2022causal}. Similar examples can be found in time series experiments \citep{bojinov2019time} and adaptive treatment of a single patient \citep{vanderlaan2018robust}. In observational time series data, such questions are complicated, as such counterfactual shifts are not directly observed. Treatment assignment often responds to the evolving state, so its ordinary association with future outcomes combines the causal effect of interest with endogenous selection. Instrumental-variable (IV) methods address this problem by isolating treatment variation induced by an external source and are therefore widely used across fields of causal inference \citep{wright1928tariff,stock2018identification,khalaf2026monetary}. However, credible use of such variation may require adjustment for a rich pretreatment information set, both to support the identifying assumptions and to sharpen estimation by accounting for predictable variation in treatment and outcome \citep{tan2006regression,chernozhukov2015postselection}. This often results in high dimensional state covariates recording such information, thwarting the inference of the causal target.

Our empirical application of FRED-MD data offers a motivation. We ask how U.S.\ inflation, unemployment, industrial production, and housing starts respond to a change in the monetary-policy stance. The policy rate is set in response to hundreds of macroeconomic indicators and to private forecasts that never enter the public record \citep{romer2004new,bernanke2005favar}, so there can always be confoundedness caused by latent conditions, no matter how many covariates are chosen. Luckily, the Romer--Romer shock \citep{romer2004new,wieland2021updated} serves as a widely-used instrument for policy shocks. However, to adjust for known confounders and purify causal effects, a high-dimensional state for adjustment is also needed, which, in our application,
is twelve lags of the FRED-MD panel containing 1464 variables \citep{mccracken2016fred}. Given the sample size being only 468 months, a sophisticated analysis framework is hence necessitated.

Formally speaking, the causal target is the horizon-specific response matrix \(\{\mathbf B_h:h\in\mathcal H\}\), whose entries record the expected change in the outcome vector \(\mathbf Y_{t+h}\) that an externally induced unit shift in the period-$t$ treatment \(\mathbf W_t\) would produce when the pre-treatment history is held fixed and the system evolves naturally afterward. This can be viewed as a semiparametric local-projection IV (LP--IV) target in the sense of \citet{jorda2005estimation} and \citet{stock2018identification}. Here, \emph{dynamic} refers to the propagation of one intervention across horizons, as in the impulse-response tradition, and not to a sustained treatment regime. We will show that under certain conditions, \(\mathbf B_h\) coincides with the normalized structural impulse response of proxy-SVAR analysis \citep{mertens2013dynamic,montielolea2021inference}.

Estimating this target involves three difficulties at once, namely an endogenous treatment identified only through an external instrument, an adjustment state whose dimension may exceed the length of the series, and serial dependence along the single trajectory. Related methods address these difficulties under different identifying and modeling assumptions. LP--IV and proxy-SVAR methods accommodate an instrument but adjust for a low-dimensional prespecified control set, and their inference theory is built around that fixed-dimensional adjustment \citep{stock2018identification,montielolea2021local,montielolea2021inference,plagborg2021local}. High-dimensional time series methods, from sparse vector autoregressions and Lasso-based local projections to factor-augmented regressions, permit many predictors \citep{basu2015regularized,kock2015oracle,chernozhukov2021lasso,krampe2023structural,adamek2024local,fan2023bridging}, but focus generally on prediction or on structural impulse responses under prespecified identification restrictions \citep{adamek2024local,krampe2023structural}, and their guarantees are tied to a particular learner. 
A common approach to identifying effects of time-varying treatments assumes sequential ignorability, meaning that the observed history suffices to control confounding at each treatment decision \citep{robins1986new,murphy2003optimal,lewis2021double}. That assumption does not fit our settings.
Debiased machine learning (DML) \citep{chernozhukov2018double} separates a low-dimensional causal parameter from an arbitrary nuisance learner through a Neyman-orthogonal score and cross-fitting, and its IV form handles endogenous treatments. In the independent-data setting, sample splitting makes the fitted nuisance independent of the held-out data. We will show that adopting DML to a single dependent trajectory is nontrivial, because the training and validation blocks remain dependent, and the independent-data guarantees do not apply directly.
Recent work brings orthogonal estimation to dependent data. \citet{semenova2023debiased} leave out temporal neighbors in dynamic panels with many units. \citet{ciganovic2026double}, which propose reverse cross-fitting for partially linear regressions, require unconfoundedness and very strong assumptions on temporal dependence. \citet{ballinari2024semiparametric} construct orthogonal impulse-response estimators for binary treatments under conditional unconfoundedness, whereas \citet{goncalves2026semiparametric} consider finite additive shifts of an i.i.d.\ structural shock. To our knowledge, no existing result covers an endogenous treatment instrumented by an external variable, a learner-agnostic adjustment for a high-dimensional state, and a single temporally-dependent trajectory, together with a theory that converts a generic prediction guarantee into the conditions the orthogonal estimator needs in a block cross-fitting scheme.

In light of these, we make the following contributions. First, we formulate the total response to a single counterfactual intervention, derive its semiparametric local-projection representation, and identify it through residualized IV moments; see \Cref{sec:model}. Then, we develop an orthogonal estimator with buffered block cross-fitting for a single time series with high-dimensional adjustment in \Cref{sec:estimation}. Under geometric \(\beta\)-mixing and regularity conditions, we derive a non-asymptotic estimation-error bound and, under additional stability conditions, joint asymptotic normality and feasible HAC inference at each fixed horizon; see \Cref{sec:inference}. Furthermore, we establish a transfer principle that converts prediction guarantees of a temporal learner into nuisance-error rates under buffered block cross-fitting in \Cref{sec:nuisance-learners}. This makes a framework for time-series prediction theory to be usable in orthogonal inference across different choices of nuisance learners, and can be applied to other similar settings. We exemplify how to verify these conditions for a factor model and sparse spline in the supplement, so that concrete results can be readily obtained. 

In the FRED-MD application in \Cref{sec:application}, we find a decline in housing starts at medium horizons after an instrumented tightening, which persists under multiple ways of residualization. To check that the procedure is reliable, we also run a plasmode experiment to establish our method's robustness. Simulations in \Cref{sec:simulations} also evaluate our method.

\paragraph{Notations and preliminaries.}
For a probability law $P$ and a fixed vector-valued function $f$, we write
\(
    \|f\|_{\psi_2,P}
    =\sup_{\|v\|_2=1}\inf\left\{c>0:
      \E_{U\sim P}\exp\left[\{v^\top f(U)\}^2/c^2\right]\le2
    \right\}
\)
for the distributional sub-Gaussian norm of $f$ under $P$. We use $\sigma(\cdot)$ and $\mathcal L(\cdot)$ to denote the sigma-field generated by, and the law of, a random variable or a random vector, respectively. We write \(a\vee b=\max\{a,b\}\), and use \(O_p(\cdot)\) to denote stochastic boundedness.
\section{Model setup and causal target}\label{sec:model}
Let $\{(\mathbf Y_t,\mathbf W_t,\mathbf Z_t,\mathbf U_t)\}_{t=1}^T$ be the observed time series, where \(\mathbf Y_t\in\mathbb R^p\) is the outcome vector, \(\mathbf W_t\in\mathbb R^s\) is the observed treatment random vector, \(\mathbf Z_t\in\mathbb R^r\) is an instrumental variable, and \(\mathbf U_t\in\mathbb R^{q_T}\) is an observed pre-treatment adjustment state. Let \(\mathcal H\subset\{0,1,\ldots,H\}\) denote the finite set of horizons of interest. The dimension $q_T$ may grow with $T$, while $p$, $s$, $r$, and $H$ are treated as fixed or small. We ask how the expected outcome \(\mathbf Y_{t+h}\) changes when only \(\mathbf W_t\) is externally shifted, holding the pre-treatment history fixed and allowing subsequent treatments and covariates to respond through their natural mechanisms.

The state \(\mathbf U_t\) may include lagged outcomes, lagged treatments, lagged instruments, and a large panel of pre-treatment covariates in applications. However, it should not contain post-treatment variables or mediators, as they may block the causal path from treatment to outcome. The instrument \(\mathbf Z_t\) is a source of exogenous variation in the treatment, which may be a policy shock or a randomized encouragement, which will be detailed in \Cref{sec:id}.

\subsection{Structural equation for the causal target}
The causal target is counterfactual by nature. Let \(\mathcal W\subseteq\mathbb R^s\) denote the set of all possible deterministic intervention or treatment values. For a fixed intervention \(w\in\mathcal W\), let \(\mathbf Y_{t+h}(w)\in\mathbb R^p\) denote the potential outcome at horizon $h$ that would be observed if the period-$t$ treatment were externally set to $w$.
We introduce \(\mathcal C_t\) to denote the full pre-assignment state information, represented as a sigma-field. It contains the observed pre-treatment history and may also include latent macroeconomic fundamentals that affect both the endogenous part of the treatment and future outcomes. The observed state \(\mathbf U_t\) is then a coarsening of this full information set, in the sense that \(\sigma(\mathbf U_t) \subsetneq \mathcal C_t\), serving as observed confounders. On the other hand, the treatment random variable \(\mathbf W_t\) is often unmeasurable with respect to \(\mathcal C_t\), since it may contain exogenous variation, including that induced by the instrument \(\mathbf Z_t\).

To formally derive the causal target, we impose a semiparametric linear response model for each \(h\in\mathcal H\) conditional on the full pre-assignment information. For any fixed intervention \(w\in\mathcal W\), we assume the structural equation
\begin{equation} \label{eq:po-linear-response}
    \E[\mathbf Y_{t+h}(w)\mid \mathcal C_t] =
    \mathbf B_h w+\mathbf m_h(\mathcal C_t),
\end{equation}
where \(\mathbf B_h\in\mathbb R^{p\times s}\) is the horizon-specific causal effect matrix and \(\mathbf m_h(\mathcal C_t)\in\mathbb R^p\) is an unknown function of the full pre-treatment state. Equivalently, for any increment \(\delta\in\mathbb R^s\) such that \(w,w+\delta\in\mathcal W\), \eqref{eq:po-linear-response} implies the constant conditional causal contrast
\begin{equation} \label{eq:causal-contrast}
    \E[\mathbf Y_{t+h}(w+\delta)-\mathbf Y_{t+h}(w)\mid \mathcal C_t]
    = \mathbf B_h\delta .
\end{equation}
Thus \(\mathbf B_h\) is defined as the causal response parameter. The horizon $h$ records how the effect of this intervention evolves over time.

To bridge the counterfactual target to the observed outcomes, we first decompose the potential outcome for each fixed $w$ into its conditional mean \eqref{eq:po-linear-response} and a residual as
\begin{equation}\label{eq:po-decomposition}
    \mathbf Y_{t+h}(w) = \mathbf B_h w + \mathbf m_h(\mathcal C_t) + \boldsymbol\eta_{t+h}(w),
\end{equation}
where, by construction, \(\E[\boldsymbol\eta_{t+h}(w)\mid \mathcal C_t]=\mathbf 0\).
To connect the counterfactual response to observed outcomes, we impose conditional mean exchangeability given \(\mathcal C_t\), so that \(\E[\mathbf Y_{t+h}(w)\mid\mathcal C_t,\mathbf W_t]=\E[\mathbf Y_{t+h}(w)\mid\mathcal C_t]\) for every \(w\in\mathcal W\),
which is weaker than conditional independence of the potential outcome and treatment \citep{imbens2004nonparametric}. Under \eqref{eq:po-decomposition}, it is equivalent to
\begin{equation}
    \E[\boldsymbol\eta_{t+h}(w)\mid \mathcal C_t,\mathbf W_t]=\mathbf 0 \quad \text{ for every } w\in\mathcal W.
    \label{eq:exchangeability}
\end{equation}
It then follows that $\E[\boldsymbol\eta_{t+h}(\mathbf W_t)\mid \mathcal C_t,\mathbf W_t] = \mathbf 0$ as the treatment random variable $\mathbf W_t$ is also supported on \(\mathcal W\). By the tower property, \(\E[\boldsymbol\eta_{t+h}(\mathbf W_t)\mid \mathcal C_t]=\mathbf 0\).

For the observed outcomes, the consistency condition dictates \(\mathbf Y_{t+h}=\mathbf Y_{t+h}(\mathbf W_t)\). Evaluating \eqref{eq:po-decomposition} at the observed treatment \(\mathbf W_t\) yields the full-state structural equation
\[
    \mathbf Y_{t+h}=\mathbf B_h\mathbf W_t+\mathbf m_h(\mathcal C_t)+\boldsymbol\eta_{t+h},
    \qquad \E[\boldsymbol\eta_{t+h}\mid \mathcal C_t]=\mathbf 0,
\]
where \(\boldsymbol\eta_{t+h} := \boldsymbol\eta_{t+h}(\mathbf W_t)\) is the future structural shock evaluated at \(\mathbf W_t\).

However, the complete information set \(\mathcal C_t\) is never observed in full. Therefore, we project the unobserved baseline \(\mathbf m_h(\mathcal C_t)\) onto the available state $\mathbf U_t$ by defining the observed-state nuisance function \(\mathbf g_{h,t}(\mathbf U_t)=\E[\mathbf m_h(\mathcal C_t)\mid \mathbf U_t]\) as a surrogate to be estimated. Substituting this into the full-state equation yields the observed-data local projection (LP) representation
\begin{equation}
    \mathbf Y_{t+h} = \mathbf B_h\mathbf W_t + \mathbf g_{h,t}(\mathbf U_t) + \boldsymbol\varepsilon_{t+h}, \qquad \E[\boldsymbol\varepsilon_{t+h}\mid \mathbf U_t]=\mathbf 0,
    \label{eq:model}
\end{equation}
where \(\boldsymbol\varepsilon_{t+h} = \mathbf m_h(\mathcal C_t)-\mathbf g_{h,t}(\mathbf U_t) + \boldsymbol\eta_{t+h}\) is the composite error term, whose zero conditional mean follows by the tower property.
The difference \(\mathbf m_h(\mathcal C_t)-\mathbf g_{h,t}(\mathbf U_t)\) is the unobserved confounding effect from the latent state, and the term \(\boldsymbol\eta_{t+h}\) is the pure structural future shock.
However, because the observed treatment \(\mathbf W_t\) may depend on latent components of \(\mathcal C_t\), e.g., private information not captured in \(\mathbf U_t\), we generally have \(\E[\boldsymbol\varepsilon_{t+h}\mid \mathbf W_t,\mathbf U_t]\neq\mathbf 0\). We therefore introduce an instrumental variable (IV) $\mathbf Z_t$ to overcome this endogeneity and identify $\mathbf B_h$, by isolating externally generated variation in \(\mathbf W_t\); see \Cref{sec:id} below. 

\begin{remark}
    Without the linear response restriction in \eqref{eq:po-linear-response}, an IV estimand would generally have a local or weighted interpretation analogous to LATE-type estimands. Such heterogeneous-effect interpretations go beyond the focus of this paper, as the target here is the global, horizon-specific causal response matrix \(\mathbf B_h\) in \eqref{eq:causal-contrast}. At the sample level, \eqref{eq:model} cleanly aligns our causal target with the structural equations commonly seen in the LP literature.
\end{remark}

\subsection{Identification of the causal estimand} \label{sec:id}
This section gives the population argument identifying the causal response matrix in \eqref{eq:causal-contrast} from the observed-data LP representation in \eqref{eq:model}. 
Fix \(h\in\mathcal H\), and let \(\mathcal T_h=\{1,\ldots,T-h\}\). All variables appearing in the following second moments are assumed square-integrable.

For each \(t\in\mathcal T_h\), define the observed-state conditional means
\begin{equation} \label{eq:conditional-mean}
    \boldsymbol\mu_{Y,h,t}(\mathbf U_t)=\E[\mathbf Y_{t+h}\mid \mathbf U_t],
    \qquad
    \boldsymbol\mu_{W,t}(\mathbf U_t)=\E[\mathbf W_t\mid \mathbf U_t],
    \qquad
    \boldsymbol\mu_{Z,t}(\mathbf U_t)=\E[\mathbf Z_t\mid \mathbf U_t],
\end{equation}
which are to be learnt by nuisance learners,
and the corresponding residualized variables
\begin{equation} \label{eq:residuals}
    \widetilde{\mathbf Y}_{t+h}=\mathbf Y_{t+h}-\boldsymbol\mu_{Y,h,t}(\mathbf U_t),
        \qquad
    \widetilde{\mathbf W}_{t}=\mathbf W_t-\boldsymbol\mu_{W,t}(\mathbf U_t),
        \qquad
    \widetilde{\mathbf Z}_{t}=\mathbf Z_t-\boldsymbol\mu_{Z,t}(\mathbf U_t).
\end{equation}
The residualization is with respect to the observed adjustment state $\mathbf U_t$ only, and again, \(\E[\boldsymbol\varepsilon_{t+h}\mid \mathbf W_t,\mathbf U_t]\) may be nonzero, indicating the presence of unmeasured confounding.
Note that we allow the conditional laws to vary with $t$ as we are currently dealing with population-level parameters, so each term is indexed with $t$.

\begin{assumption}[Conditional IV validity]\label{ass:iv}
For each \(h\in\mathcal H\), the instrument \(\mathbf Z_t\) satisfies:
\begin{enumerate}
    \item[(i)] (Conditional exclusion)
    For every \(t\in\mathcal T_h\), \(\E[\boldsymbol\varepsilon_{t+h}\mid \mathbf Z_t,\mathbf U_t]=\mathbf 0\).
    \item[(ii)] (Aggregate residualized relevance) Define the time-specific first-stage matrices
    \begin{equation} \label{eq:first-stage-matrix}
        \mathbf A_t = \E[\widetilde{\mathbf W}_t\widetilde{\mathbf Z}_t^\top]
        \in\mathbb R^{s\times r}.
    \end{equation}
    Their horizontal concatenation $\boldsymbol{\mathcal A}_{h,T}:=[\,\mathbf A_t:t\in\mathcal T_h\,]\in\mathbb R^{s\times r(T-h)}$ has full row rank $s$.
\end{enumerate}
\end{assumption}

Part (i) is the familiar exclusion restriction from the IV literature, here imposed on the observed-data error after adjusting for the state \(\mathbf U_t\). It requires the instrument to be unrelated to the composite error in \eqref{eq:model}, so that the influence of \(\mathbf Z_t\) on the future outcome \(\mathbf Y_{t+h}\) must travel only through the treatment \(\mathbf W_t\) and never along a direct path.

Part (ii) is the relevance condition stated in an aggregate form. Instead of asking each period to identify every treatment direction, it only requires that the periods jointly do so, and that equivalently, \(\sum_{t\in\mathcal T_h}\mathbf A_t\mathbf A_t^\top\) be nonsingular. As will be seen in \Cref{prop:identification}, any individual \(\mathbf A_t\) may be rank deficient, as long as other periods supply relevance along the directions it misses.  
Otherwise, some nonzero treatment direction is never shifted by the residualized instrument, and the corresponding part of \(\mathbf B_h\) cannot be recovered.
Under part (ii), the pseudo-inverse
$\boldsymbol{\mathcal A}_{h,T}^{\dagger} :=\boldsymbol{\mathcal A}_{h,T}^{\top}(\boldsymbol{\mathcal A}_{h,T}\boldsymbol{\mathcal A}_{h,T}^{\top})^{-1} \in \mathbb{R}^{r(T-h) \times s}$
satisfies \(\boldsymbol{\mathcal A}_{h,T}\boldsymbol{\mathcal A}_{h,T}^{\dagger}=\mathbf I_s\).

\begin{proposition}[Population identification]\label{prop:identification}
Suppose the observed-data representation \eqref{eq:model} holds for a fixed \(h\in\mathcal H\) and every \(t\in\mathcal T_h\). If \Cref{ass:iv} holds, then for every \(t\in\mathcal T_h\), writing $\mathbf M_{h,t}=\E[\widetilde{\mathbf Y}_{t+h}\widetilde{\mathbf Z}_t^\top] \in\mathbb R^{p\times r}$, we have
\begin{equation}\label{eq:identification-moment}
    \mathbf M_{h,t}=\mathbf B_h\mathbf A_t.
\end{equation}
Consequently, stacking these moment equations identifies \(\mathbf B_h\) as
\begin{equation}\label{eq:identification-formula}
    \mathbf B_h=\boldsymbol{\mathcal M}_{h,T}\boldsymbol{\mathcal A}_{h,T}^{\dagger},
\end{equation}
where \(\boldsymbol{\mathcal{M}}_{h,T}:=\big[\mathbf M_{h,t}:t\in\mathcal T_h\big] \in \mathbb{R}^{p \times r(T-h)}\).
\end{proposition}
This proposition recovers the causal response matrix of \eqref{eq:causal-contrast} from observable second moments, using only the two IV conditions. Strictly speaking, the instrument \(\mathbf Z_t\) is not required to be residualized by $\mathbf U_t$ for population identification, but in practice, if the nuisance learners fail to remove all the influence of \(\mathbf U_t\) from \(\mathbf Y_{t+h}\) and \(\mathbf W_t\), the effective error may contain part of \(\mathbf g_{h,t}(\mathbf U_t)\), which could be correlated with the original instrument through the observed state. In this way, spurious correlation between the instrument and the error may arise. Another reason is that the population identity \eqref{eq:identification-moment} pins down the target of estimation, and its estimation procedure in \Cref{sec:estimation} will show that residualization on \(\mathbf Z_t\) ensures robustness to nuisance errors via Neyman-orthogonality. 

\subsection{Alternative interpretations of the model}
\label{sec:alternative-interpretations}
We relate \(\mathbf B_h\) to structural impulse responses and to causal contrasts used in sequential-treatment analysis. These comparisons depend on how the intervention is imposed and on how subsequent treatments are determined.

\paragraph{Structural impulse responses.}
From the perspective of the structural time series, the dynamics of treatment and outcome is jointly generated from a set of unobserved structural shocks that propagate through the system, and an intervention is imagined as exogenous shift in some parts of a shock $\bm\zeta_{t}$. 
In the stable structural vector MA (SVMA) representation
\begin{equation}\label{eq:svma-wold}
    [\mathbf W_t^\top \quad \mathbf Y_t^\top]^\top
    =\sum_{j\ge 0} \bm\Theta_j \bm\zeta_{t-j},
\end{equation}
the observables are driven by unobserved i.i.d.\ mean-zero structural shocks $\bm\zeta_t$, and the impact matrix $\bm\Theta_j$ records how a shock at date $t-j$ moves the observables $j$ periods later. 
The form \eqref{eq:svma-wold} also encompasses the stable finite-order SVAR by the Wold representation theorem. In this language, $\mathbf U_t$ plays the role of the information available before the period-$t$ shock, such as the lagged state in a VAR. The nuisance $\mathbf g_{h,t}$ aims to forecast the part of the outcome after accounting for the predictable component of $\mathbf W_t$. The instrument $\mathbf Z_t$ is called a proxy for the target shock in $\mathbf W_t$ in the sense that it is only non-orthogonal to target shocks.
After removing the components predictable from \(\mathbf U_t\), taking moments with the residualized proxy isolates the target-shock contributions to \(\mathbf W_t\) and \(\mathbf Y_{t+h}\). The resulting normalized structural impulse response equals \(\mathbf B_h\) under the proxy-SVMA assumptions of \Cref{ex:svma} \citep{stock2018identification,mertens2013dynamic,montielolea2021inference}. This gives \(\mathbf B_h\) the economic interpretation of being a response to structural shocks normalized by their treatment impact. We can therefore estimate these structural responses through horizon-specific IV moments without estimating a transition model for the full system.

\paragraph{Sequential-treatment effects.}
In sequential-treatment analysis, causal contrasts depend on the continuation regime \citep{robins1986new,robins1994correcting,murphy2003optimal}. For example, a structural nested mean model (SNM) can describe the effect of changing the current treatment while fixing subsequent treatments at reference values. In contrast, our parameter \(\mathbf B_h\) includes the effects transmitted through changes in later treatments induced by the initial intervention. This generally differs from the effect obtained by holding later treatments fixed. We impose \eqref{eq:po-linear-response} directly on the total response and identify it through the period-$t$ instrument, without specifying effects under other continuation regimes. Our framework thus permits IV inference on this total response from a single dependent trajectory, even when flexible adjustment for the high-dimensional $\mathbf U_t$ cannot remove treatment confounding.

\section{Orthogonal estimation via block cross-fitting} \label{sec:estimation}

This section proposes an estimation pipeline for \(\mathbf B_h\) via debiased machine learning (DML) route, in which flexible learners first fit the conditional means and the causal target is then estimated from a single moment equation, and then states the requirements for the pipeline.

Before proceeding, note that identification in \Cref{prop:identification} placed no restriction on how the conditional laws evolve across time, since the moment equation \eqref{eq:identification-moment} holds for each $t$ even when the state distribution drifts. However, estimation cannot proceed at that level of generality. A single realized trajectory supplies only one draw of \((\mathbf Y_{t+h},\mathbf W_t,\mathbf Z_t)\) in each period, whereas the conditional means \(\boldsymbol\mu_{Y,h}\), \(\boldsymbol\mu_W\), and \(\boldsymbol\mu_Z\) can be recovered only by pooling observations over $t$ so that the sample size can increase. Basically speaking, such pooling is informative when these functions do not themselves vary with $t$. We therefore work in a temporally homogeneous regime where \(\boldsymbol\mu_{Y,h,t}=\boldsymbol\mu_{Y,h}\), \(\boldsymbol\mu_{W,t}=\boldsymbol\mu_W\), and \(\boldsymbol\mu_{Z,t}=\boldsymbol\mu_Z\) are constant in $t$; see \Cref{ass:estimation-primitive}(i) for details. This constrains only three regression functions, so it is much weaker than stationarity of the entire process, even though stationarity is commonly seen in time series analysis. Nuisance functions that drift smoothly with $t$ could in principle be tracked by localizing in time, but that extension may be left for future work.

Still, time homogeneity of the conditional means does not require the residualized first-stage moment to be constant, because \(\mathbf A_t\) also depends on the time-specific distribution of the residuals. Therefore, we define the averaged moments
\begin{equation}\label{eq:average-first-stage}
    \overline{\mathbf M}_{h,T}
    =\frac{1}{n_h}\sum_{t=1}^{n_h}\mathbf M_{h,t} \in \mathbb R^{p\times r},
    \qquad
    \overline{\mathbf A}_{h,T}
    =\frac{1}{n_h}\sum_{t=1}^{n_h}\mathbf A_t \in \mathbb R^{s\times r},
\end{equation}
Averaging \eqref{eq:identification-moment} over the estimation period gives
\(
    \overline{\mathbf M}_{h,T}=\mathbf B_h\overline{\mathbf A}_{h,T}.
\)
Thus, we require $r\ge s$ and that \(\overline{\mathbf A}_{h,T}\) has full row rank, and then, in population level,
\[
    \mathbf B_h
    =\overline{\mathbf M}_{h,T}\overline{\mathbf A}_{h,T}^{\dagger},
    \qquad
    \overline{\mathbf A}_{h,T}^{\dagger}
    =\overline{\mathbf A}_{h,T}^{\top}
    (\overline{\mathbf A}_{h,T}\overline{\mathbf A}_{h,T}^{\top})^{-1}.
\]
The remainder of this section first introduces a Neyman-orthogonal score which prevents first-order errors in learned nuisances from perturbing the estimating equation, then gives the cross-fitted estimator, and finally records the learner-facing conditions used in the theory.

\subsection{Orthogonal IV score}

Let \(\eta=(\mu_Y,\mu_W,\mu_Z)\) collect generic nuisance functions of \(\mathbf U_t\). For fixed $h$, define the residualized IV moment function, which serves as the Neyman-orthogonal score,
\begin{equation}\label{eq:orthogonal-score}
    \psi_{t,h}(\mathbf B,\eta)
    = \{\mathbf Y_{t+h}-\mu_Y(\mathbf U_t) -\mathbf B(\mathbf W_t-\mu_W(\mathbf U_t))\}
    \{\mathbf Z_t-\mu_Z(\mathbf U_t)\}^\top .
\end{equation}
Recall the definition of the population residuals in \eqref{eq:residuals}. At the true nuisance value \(\eta_{0,h}=(\boldsymbol\mu_{Y,h},\boldsymbol\mu_W,\boldsymbol\mu_Z)\) and the causal parameter, define the oracle score
\begin{equation}\label{eq:oracle-score}
    \psi^0_{t,h}
    :=\psi_{t,h}(\mathbf B_h,\eta_{0,h})
    =\boldsymbol\varepsilon_{t+h}\widetilde{\mathbf Z}_t^\top.
\end{equation}
The oracle score is unbiased, i.e., $\E(\psi^0_{t,h})=\mathbf 0_{p\times r}$. Moreover, the score function \eqref{eq:orthogonal-score} is Neyman orthogonal, in the sense that the Gateaux derivative of \(\E[\psi_{t,h}(\mathbf B_h,\eta_{0,h}+u a)]\) at $u=0$ is zero for any square-integrable $\mathbf U_t$-measurable perturbation $a$ and the population moment is thus unaffected by first-order nuisance errors; see \Cref{prop:supp-orthogonality}.

This score is the residualized form of the IV moment in \Cref{prop:identification}. If the nuisance functions were known, averaging the population moments would give \(\overline{\mathbf M}_{h,T}=\mathbf B_h\overline{\mathbf A}_{h,T}\). The feasible estimator below replaces the population term with sample estimates.

\subsection{Cross-fitted residualized IV estimator}

Fix a horizon \(h\in\mathcal H\) and require $n_h:=T-h>2$.
For any time-indexed array $f_t$, define
\[
    P_{n,h}f=\frac{1}{n_h}\sum_{t=1}^{n_h} f_t .
\]
The estimator first constructs out-of-sample estimates of the nuisance functions \(\boldsymbol\mu_{Y,h}\), \(\boldsymbol\mu_W\), and \(\boldsymbol\mu_Z\), i.e., the nuisance functions are trained in the training sample and evaluated in the validation sample.
In time-series settings, such sample splitting should be block-wise, and the blocks take turns to serve as the validation set. In each turn, the training set should be separated from the validation set by a growing buffer $b_T$ to reduce the effect of first-order bias of nuisance estimation.
Specifically, a fixed number $L$ of contiguous, balanced validation blocks \(\mathcal I_{\ell,h}\) partition $\{1,\ldots,n_h\}$. For fold $\ell$, form the admissible training-index set
\begin{equation} \label{eq:training-set}
    \mathcal J_{\ell,h}
    :=\left\{s\in\{1,\ldots,n_h\}:
    \min_{t\in\mathcal I_{\ell,h}}|s-t|\ge b_T\right\},
\end{equation}
and the nuisance learners for this fold  \(\mathcal I_{\ell,h}\) are trained only on
\(
    \{(\mathbf Y_{s+h},\mathbf W_s,\mathbf Z_s,\mathbf U_s):
    s\in\mathcal J_{\ell,h}\}
\).

In this way, all time points are used for validation exactly once, but may be used for training many different validation folds. Let \(\widehat{\boldsymbol\mu}_{Y,h}(\mathbf U_t)\), \(\widehat{\boldsymbol\mu}_W(\mathbf U_t)\), and \(\widehat{\boldsymbol\mu}_Z(\mathbf U_t)\) denote the cross-fitted nuisance estimates evaluated at time $t$. Define the estimated residuals
\begin{equation}\label{eq:residuals-estimated}
    \widehat{\mathbf e}_{Y,t,h}
    = \mathbf Y_{t+h}-\widehat{\boldsymbol\mu}_{Y,h}(\mathbf U_t),
    \qquad
    \widehat{\mathbf e}_{W,t}
    = \mathbf W_t-\widehat{\boldsymbol\mu}_W(\mathbf U_t),
    \qquad
    \widehat{\mathbf e}_{Z,t}
    = \mathbf Z_t-\widehat{\boldsymbol\mu}_Z(\mathbf U_t).
\end{equation}
Note that the cross-fitted nuisance functions $\mathbf U_t \mapsto \widehat{\boldsymbol\mu}_\cdot (\mathbf U_t)$ depend on which validation fold the $\mathbf U_t$ belongs to, and different learners may be used for different folds. Set
\begin{equation}
    \widehat{\mathbf M}_h
    = P_{n,h}\!\left(
    \widehat{\mathbf e}_{Y,t,h}\widehat{\mathbf e}_{Z,t}^\top
    \right)\in\mathbb R^{p\times r},
    \qquad
    \widehat{\mathbf A}_h
    = P_{n,h}\!\left(
    \widehat{\mathbf e}_{W,t}\widehat{\mathbf e}_{Z,t}^\top
    \right)\in\mathbb R^{s\times r}.
    \label{eq:sample-moments}
\end{equation}
When \(\widehat{\mathbf A}_h\) has full row rank, which will be shown to hold with high probability, we recall that $\widehat{\mathbf A}_h^\dagger = \widehat{\mathbf A}_h^\top (\widehat{\mathbf A}_h\widehat{\mathbf A}_h^\top)^{-1} \in\mathbb R^{r\times s}$,
and the causal response matrix can be estimated by
\begin{equation}\label{eq:orthogonal-estimator}
    \widehat{\mathbf B}_h=\widehat{\mathbf M}_h\widehat{\mathbf A}_h^\dagger \in\mathbb R^{p\times s}.
\end{equation}
Now we describe the whole cross-fitting scheme in \Cref{alg:orthogonal-lpiv} that defines the estimator. It fits the out-of-sample nuisance means block by block, residualizes on each block, aggregates the residuals into \(\widehat{\mathbf M}_h\) and \(\widehat{\mathbf A}_h\), and solves the low-dimensional IV step for \(\widehat{\mathbf B}_h\).

\begin{algorithm}[H]
\caption{Orthogonal residualized IV estimator via block cross-fitting regime}
\label{alg:orthogonal-lpiv}
\begin{algorithmic}[1]
\REQUIRE Horizon $h$, observations \(\{(\mathbf Y_{t+h},\mathbf W_t,\mathbf Z_t,\mathbf U_t):1\le t\le n_h\}\), $L$ balanced contiguous blocked folds, buffer length $b_T$, and nuisance learners.
\FOR{\(\ell=1,2,\ldots,L\)}
\STATE Fit \(\boldsymbol\mu_{Y,h}\), \(\boldsymbol\mu_W\), and \(\boldsymbol\mu_Z\) using only observations with indices in \(\mathcal J_{\ell,h}\) formed via \eqref{eq:training-set}.
\label{line:fit-nuisance}
\STATE Calculate the residuals \eqref{eq:residuals-estimated} for all $t\in\mathcal I_{\ell,h}$ via the fitted nuisance functions in line \ref{line:fit-nuisance}.
\ENDFOR
\STATE Aggregate the residuals over all folds and compute \(\widehat{\mathbf M}_h\) and \(\widehat{\mathbf A}_h\) via \eqref{eq:sample-moments}.
\ENSURE \(\widehat{\mathbf B}_h=\widehat{\mathbf M}_h\widehat{\mathbf A}_h^\dagger\) if \(\widehat{\mathbf A}_h\) has full row rank; otherwise report first-stage failure.
\end{algorithmic}
\end{algorithm}

In the exactly identified case $r=s$, \eqref{eq:orthogonal-estimator} reduces to the usual residualized IV ratio \(\widehat{\mathbf B}_h=\widehat{\mathbf M}_h\widehat{\mathbf A}_h^{-1}\), which is the sample analogue of the averaged population identity following \eqref{eq:average-first-stage}, estimated by general methods-of-moments based on the sample score \(P_{n,h}\psi_{t,h}(\mathbf B,\widehat\eta)\).

\subsection{Conditions on the learned nuisances}

To use the orthogonality property, the fitted nuisance functions need to be accurate and well behaved to some extent. To characterize this, we first define the cross-fitted nuisance errors
\[
    \boldsymbol\Delta_{Y,t,h}
    = \widehat{\boldsymbol\mu}_{Y,h}(\mathbf U_t)-\boldsymbol\mu_{Y,h}(\mathbf U_t),
    \qquad
    \boldsymbol\Delta_{W,t}
    = \widehat{\boldsymbol\mu}_{W}(\mathbf U_t)-\boldsymbol\mu_{W}(\mathbf U_t),
    \qquad
    \boldsymbol\Delta_{Z,t}
    = \widehat{\boldsymbol\mu}_{Z}(\mathbf U_t)-\boldsymbol\mu_{Z}(\mathbf U_t).
\]
Write $\boldsymbol\mu_Y=\boldsymbol\mu_{Y,h}$ for convenience, and, for
$a\in\{Y,W,Z\}$, define fold-specific error functions
\[
    \boldsymbol\delta_{a,\ell}(u)
    =\widehat{\boldsymbol\mu}^{(-\ell)}_a(u)-\boldsymbol\mu_a(u).
\]
Thus, for $t\in\mathcal I_{\ell,h}$,
$\boldsymbol\Delta_{a,t}=\boldsymbol\delta_{a,\ell}(\mathbf U_t)$, with the
horizon index $h$ on the outcome error suppressed when no confusion can arise.
For $P_t=\mathcal L(\mathbf U_t)$, define
\[
    \|\boldsymbol\Delta_a\|_{n,h,2}
    =\left(P_{n,h}\|\boldsymbol\Delta_{a,t}\|_2^2\right)^{1/2},
    \qquad
    \|\boldsymbol\Delta_a\|_{L_2}
    =\max_{1\le\ell\le L}\sup_{t\in\mathcal I_{\ell,h}}
      \left(\E_{U\sim P_t}
      \|\boldsymbol\delta_{a,\ell}(U)\|_2^2\right)^{1/2}.
\]
In the population norm, the trained function is held fixed and only the fresh
state $U\sim P_t$ is integrated out. The following condition characterizes the size of these fitted error functions.
\begin{condition}[Cross-fitted nuisance prediction rates]\label{cond:nuisance-rate}
For the fixed horizon $h$, there is an event \(\mathcal E_{T,h}\) with \(\Pr(\mathcal E_{T,h})\to1\) on which the cross-fitted nuisance estimates satisfy the two controls below. Throughout, \(\|\cdot\|\) denotes either the population \(L_2\) norm, evaluated uniformly over the time-specific laws of \(\mathbf U_t\) in each validation block with the trained nuisance held fixed, or its empirical counterpart \(\|\cdot\|_{n,h,2}\), and both versions are required.
\begin{enumerate}
    \item[(i)] \textup{(Individual logarithmic-rate consistency)} The three nuisance errors satisfy the logarithmically strengthened consistency requirement
    \begin{equation} \label{eq:def-rbar_T}
        \max\{\|\boldsymbol\Delta_{Y,h}\|,\|\boldsymbol\Delta_W\|,\|\boldsymbol\Delta_Z\|\}
        \le \bar r_T ,
        \qquad \bar r_T\log n_h\to0 .
    \end{equation}
    \item[(ii)] \textup{(Product rate)} The instrument error and the larger of the remaining two satisfy
    \begin{equation} \label{eq:def-r_T}
        \|\boldsymbol\Delta_Z\|\cdot\max\{\|\boldsymbol\Delta_{Y,h}\|,\|\boldsymbol\Delta_W\|\}
        \le r_T^2 ,
        \qquad r_T= o(T^{-1/4}) .
    \end{equation}
\end{enumerate}
\end{condition}

Part (ii) is the binding requirement. It governs the second-order bias of the orthogonal score, which enters the estimation error only through the product form, without requiring a fast rate on each nuisance separately. Part (i) is much weaker and requires each nuisance error to vanish faster than $\log^{-1} n_h$. This rate controls the first-order fluctuation of the empirical score.
The product structure is the dependent-data counterpart of the refined condition for the Robinson-style partially linear IV score in \citet{chernozhukov2018double}. In particular, when the instrument is close to exogenous with respect to the state, its conditional mean \(\boldsymbol\mu_Z\) is nearly constant and \(\|\boldsymbol\Delta_Z\|\) is small, so Part (ii) can hold even when the outcome and treatment nuisances are estimated at a considerably slower rate. Many econometric applications can exploit this asymmetry, where externally identified instruments make \(\boldsymbol\mu_Z\) simple to learn, so that the policy rule \(\boldsymbol\mu_W\) and the outcome projection can carry more of the high-dimensional burden.

Such a condition may also be required in settings from other topics where first-stage errors need to be controlled for a trajectory of dependent observations. There is no existing result tailored to the block cross-fitting scheme in \Cref{alg:orthogonal-lpiv} on how to verify such rates for an off-the-shelf temporal learner, and \Cref{thm:regression-rate-transfer} provides a general recipe; see \Cref{sec:nuisance-learners}.

Apart from the size, the tails of the nuisance errors also need to be controlled for empirical process analysis, and we require that errors for the fitted functions have uniformly light tails on the corresponding validation block.
\begin{condition}[Tail envelope for fitted nuisance errors]
\label{cond:sg-nuisance-error}
Let \(\mathcal G_\ell\) be the $\sigma$-field generated by the training observations indexed by \(\mathcal J_{\ell,h}\) and any independent auxiliary learner randomization for fold \(\ell\), so that the functions
$\boldsymbol\delta_{a,\ell}$ are fixed conditional on $\mathcal G_\ell$.
On the event \(\mathcal E_{T,h}\) in \Cref{cond:nuisance-rate}, and for $P_t=\mathcal L(\mathbf U_t)$, there is a model-dependent constant $\bar K>0$, such that
\[
    \sup_{1\le\ell\le L}\sup_{t\in\mathcal I_{\ell,h}}
    \max_{a\in\{Y,W,Z\}}
    \|\boldsymbol\delta_{a,\ell}\|_{\psi_2,P_t}
    \le \bar K.
\]
\end{condition}
Both conditions constrain the fitted nuisances rather than the data-generating process, and both can be checked directly for concrete learners. The tail envelope holds, for instance, whenever the learners are uniformly bounded, or clipped to a fixed range with bounded truth and the fitted error then inherits a fixed sub-Gaussian envelope. For sparse linear models, sub-Gaussian covariates together with a suitable coefficient-error bound imply the condition.

\section{Statistical properties}\label{sec:inference}
With the estimator and the learner-facing conditions in place, we now establish its statistical behavior. We first give a non-asymptotic estimation error bound and an oracle expansion for \(\widehat{\mathbf B}_h\). We then use that expansion for joint fixed-horizon inference and, as a scalar specialization, feasible inference for causal contrasts.

\subsection{Estimation theory}
\label{subsec:est-err}
Apart from the conditions on nuisance errors in \Cref{sec:estimation}, the estimation error bound needs control on the underlying process, governing its temporal stability, dependence, and tails, characterized by following assumptions.

\begin{assumption}[Temporal stability, mixing, and tails]\label{ass:estimation-primitive}
Fix \(h\in\mathcal H\). The following conditions hold with their constants being model-dependent but not depending on $T$.
\begin{enumerate}
    \item[(i)] The conditional mean functions \eqref{eq:conditional-mean} are time-invariant.
    \item[(ii)] The joint process \(\mathcal X_{t,h}=(\boldsymbol\varepsilon_{t+h},\widetilde{\mathbf W}_t,\widetilde{\mathbf Z}_t,\mathbf U_t)\) is $\beta$-mixing, i.e, there exist some $C_\beta,c_\beta>0$ such that
    \[
        \beta_h(k) :=
        \sup_{j\in\mathbb Z}\beta\left(
        \sigma(\mathcal X_{t,h}:t\le j),
        \sigma(\mathcal X_{t,h}:t\ge j+k)
        \right)
        \le C_\beta\exp(-c_\beta k)
    \]
    \item[(iii)] The residuals have uniform sub-Gaussian bounds
    \[
        \sup_t\|\widetilde{\mathbf W}_t\|_\sgnorm\le K_W,\qquad
        \sup_t\|\widetilde{\mathbf Z}_t\mid\mathbf U_t\|_\sgnorm\le K_Z,\qquad
        \sup_t\|\boldsymbol\varepsilon_{t+h}\mid\mathbf U_t\|_\sgnorm\le K_\varepsilon.
    \]
    Also, \(\|\mathbf B_h\|_{\op}\le C_B\) for some constant \(C_B>0\).
\end{enumerate}
\end{assumption}

Part (i) is the temporal-homogeneity restriction motivated at the start of \Cref{sec:estimation}, restricting only the forms of the three conditional means and not the distribution of the entire process.
Under \Cref{ass:estimation-primitive} and \Cref{ass:iv}(i), each coordinate of \(\boldsymbol\varepsilon_{t+h}\widetilde{\mathbf Z}_t^\top\) and of \(\widetilde{\mathbf W}_t\widetilde{\mathbf Z}_t^\top-\mathbf A_t\) is centered, uniformly sub-exponential, and geometrically mixing.
Such constraints are standard for deriving concentration inequalities for dependent processes, and they are used to control the leading stochastic terms in the estimation error.

\Cref{ass:estimation-primitive}(ii) quantifies serial dependence. To control its effect on cross-fitting estimation, the training and validation sets in \Cref{alg:orthogonal-lpiv} need enough probabilistic separation.
\begin{assumption}[Block cross-fitting regularization]\label{ass:cross-fitting}
The buffer in \Cref{alg:orthogonal-lpiv} satisfies $b_T\to\infty$ and $b_T/n_h\to0$ as $T\to\infty$.
If a learner uses any auxiliary randomization for training, such randomness is independent of the data process.
\end{assumption}

We further require that the average first-stage moment \(\overline{\mathbf A}_{h,T} = n_h^{-1}\sum_{t=1}^{n_h}\mathbf A_{t}\) is well-conditioned, so that the IV equation can be solved stably via the form \eqref{eq:orthogonal-estimator}.
\begin{assumption}[Average residualized relevance]\label{ass:average-relevance}
For the fixed horizon $h$, there are a fixed $T_0<\infty$ and a model-dependent constant $\underline\sigma_A>0$ such that
\[
    \sigma_{\min}(\overline{\mathbf A}_{h,T})
    \ge \underline\sigma_A
    \qquad\text{for every }T\ge T_0.
\]
\end{assumption}
This assumption is stronger than \Cref{ass:iv}(ii), prohibiting the full row rank of the horizontally stacked matrix $\boldsymbol{\mathcal A}_{h,T}$ from cancellation when $\mathbf A_t$'s are averaged over time. Still, it allows the individual \(\mathbf A_t\)'s to be rank-deficient and have no limiting distribution, so it is much weaker than the usual fixed-dimensional IV relevance condition.
\begin{theorem}[Orthogonal estimation and linear expansion]\label{thm:orthogonal-estimation}
Fix \(h\in\mathcal H\). Suppose the observed-data model \eqref{eq:model} and \Cref{ass:iv,ass:estimation-primitive,ass:average-relevance,ass:cross-fitting} hold. Under \Cref{cond:nuisance-rate,cond:sg-nuisance-error}, the \(\widehat{\mathbf B}_h\) estimated via \Cref{alg:orthogonal-lpiv} satisfies the following.
\begin{enumerate}
    \item[(i)] \textup{(Non-asymptotic consistency)} For every $c>0$ and sufficiently large $T$, with probability at least
    \begin{equation}
        \Pr(\mathcal E_{T,h})-n_h^{-c}-3L\, C_\beta\exp(-c_\beta b_T),
    \end{equation}
    the matrix \(\widehat{\mathbf A}_h\) has full row rank and simultaneously, writing \(\rho_{n,h}=\sqrt{\log n_h/n_h}\),
    \begin{equation}\label{eq:consistency-rate}
        \|\widehat{\mathbf B}_h-\mathbf B_h\|_\Fr
        \le C\big\{\sqrt{1+c}\,\rho_{n,h}+r_T^2\big\}
        \lesssim \sqrt{\frac{(1+c)\log n_h}{n_h}}
    \end{equation}
    where the constant $C$ depends only on the fixed dimensions and fold number $(p,s,r,L)$ and the model primitives $(C_\beta,c_\beta,K_W,K_Z,K_\varepsilon,C_B,\bar K,\underline\sigma_A)$, but not on $T$ and $c$.
    \item[(ii)] \textup{(Asymptotic linear expansion)} The estimator admits the expansion
    \begin{equation}\label{eq:orthogonal-linear-expansion}
        \widehat{\mathbf B}_h-\mathbf B_h
        = P_{n,h}(\psi^0_{t,h}\overline{\mathbf A}_{h,T}^\dagger)
        + \mathbf R_{T,h},
    \end{equation}
    where \(\|\mathbf R_{T,h}\|_\Fr=o_p(n_h^{-1/2})\). Consequently, \(\|\widehat{\mathbf B}_h-\mathbf B_h\|_\Fr=O_p(n_h^{-1/2})\).
\end{enumerate}
\end{theorem}

A more general and detailed formulation of \Cref{thm:orthogonal-estimation} and their proofs are given in \Cref{app:subsec:est-error-proof}. Part (i) is a high-probability consistency rate with an adjustable tail exponent. A larger $c$ lowers the failure probability $n_h^{-c}$ at the cost of the \(\sqrt{1+c}\) factor on the oracle empirical-process scale $\rho_{n,h}$, while the nuisance term $r_T^2$ is unaffected. Neyman-orthogonality ensures that the first-order nuisance errors is absorbed into the leading term, so the overall error rate is dominated by the oracle empirical score. 

Part (ii) concerns only boundedness in probability and depicts the finer behaviour of the estimation error by isolating the average oracle influence \(P_{n,h}(\psi^0_{t,h}\overline{\mathbf A}_{h,T}^\dagger)\), namely the sample average that would arise if the nuisance functions were known. The remainder is $o_p(n_h^{-1/2})$, so the estimator is first-order equivalent to this oracle average and attains the parametric $n_h^{-1/2}$ rate even if the high-dimensional nuisances rates are looser.
This approximate linear expansion serves as the premise for uncertainty quantification analyzed below.

\subsection{Normal approximation for causal targets and contrasts}
\label{subsec:clt}

The linear expansion in \Cref{thm:orthogonal-estimation}(ii) gives a fixed-dimensional first-order representation for the whole causal matrix. Throughout this subsection the horizon $h$ is fixed, and we let $n=n_h=T-h$ and $P_n=P_{n,h}$. Write
\(\mathbf q_{t,h}:=\vectorize(\boldsymbol\varepsilon_{t+h}\widetilde{\mathbf Z}_t^\top)\in\mathbb R^{pr},
\)
and vectorizing \eqref{eq:orthogonal-linear-expansion} gives
\[
    \vectorize(\widehat{\mathbf B}_h-\mathbf B_h)
    =\frac1{n}\sum_{t=1}^n(\overline{\mathbf A}_{h,T}^{\dagger\top}\otimes\mathbf I_p)\mathbf q_{t,h}
    +\vectorize(\mathbf R_{T,h}).
\]
Time variation in $\mathbf A_t$ therefore affects the uncertainty quantification only through its sample average, and we need to control its limiting behavior shown in the next assumption.
\begin{assumption}[Oracle-score and population-average first-stage stability]\label{ass:uq-stationarity}
For the fixed horizon $h$, the oracle score process $(\mathbf q_{t,h})_{t\in\mathbb Z}$ is strictly stationary, and there exists full row-rank matrix $\mathbf A_\infty\in\mathbb R^{s\times r}$ such that $\overline{\mathbf A}_{h,T}\to \mathbf A_\infty$ as \(T\to\infty\).
\end{assumption}
The centered oracle score needs stationarity and the average first-stage matrix need to converge, so that the partial sum of the influence process can have a limiting distribution. Indeed, together with \Cref{ass:estimation-primitive}(ii)--(iii), the Davydov covariance argument used in \Cref{prop:supp-variance-proxy} ensures the long-run variance of $\mathbf q_{t,h}$ is well defined as
\[
    \boldsymbol\Lambda_h
    :=\sum_{m\in\mathbb Z}\cov(\mathbf q_{t,h},\mathbf q_{t+m,h})
    \in\mathbb R^{pr\times pr},
\]
and accordingly, the limiting long-run variance of the vectorized influence process
\begin{equation}\label{eq:vector-lrv}
    \boldsymbol\Omega_h
    :=(\mathbf A_{\infty}^{\dagger\top}\otimes\mathbf I_p)\boldsymbol\Lambda_h
    (\mathbf A_{\infty}^{\dagger}\otimes\mathbf I_p)
\end{equation}
is also finite and well-defined because a full-row-rank \(\mathbf A_\infty\) ensures \(\overline{\mathbf A}_{h,T}^\dagger\to \mathbf A_\infty^\dagger\). But it can be rank deficient, in which case the Gaussian limit in \eqref{eq:vector-clt} below is degenerate on \(\operatorname{col}(\boldsymbol\Omega_h)\).

\begin{theorem}[Central limit theorem]\label{thm:vector-clt}
Suppose the assumptions and the conditions of \Cref{thm:orthogonal-estimation} and \Cref{ass:uq-stationarity} hold, and fix \(h\in\mathcal H\). Then
\begin{equation} \label{eq:vector-clt}
    \sqrt n\,\vectorize(\widehat{\mathbf B}_h-\mathbf B_h)
    \dto
    N(\mathbf 0,\boldsymbol\Omega_h).
\end{equation}
And for every scalar causal contrast $\theta_h(a,b)=a^\top\mathbf B_h b$ with its estimator \(\widehat\theta_h(a,b)=a^\top\widehat{\mathbf B}_h b\) where \((a,b)\in\mathbb R^p\times\mathbb R^s\) is fixed, whenever its limiting long-run variance \(\sigma_h^2(a,b):=(b\otimes a)^\top\boldsymbol\Omega_h(b\otimes a)>0\), we have the scalar central limit theorem
\[
    {\sqrt n\{\widehat\theta_h(a,b)-\theta_h(a,b)\}}\big/
    {\sigma_h(a,b)}
    \dto N(0,1).
\]
\end{theorem}

\subsection{Feasible inferential theory}
To turn the asymptotic normality established in \Cref{subsec:clt} into a practical inference procedure, we need to estimate the long-run variance \(\boldsymbol\Omega_h\) and its scalar counterpart \(\sigma_h^2(a,b)\), as they depend on the unobserved oracle influence process. 

Since the influence vectors are autocorrelated across time $t$, we use the HAC estimation protocol.
A feasible version of the joint CLT uses a long-run variance estimator. Define the vectorized oracle influence for sample size \(T\) and plug-in influence respectively as
\[
    \boldsymbol\psi_{t,h}^{(T)} = \operatorname{vec}\!\left(
    \boldsymbol\varepsilon_{t+h}\widetilde{\mathbf Z}_t^\top
    \overline{\mathbf A}_{h,T}^{\dagger}\right)
    =(\overline{\mathbf A}_{h,T}^{\dagger\top}\otimes\mathbf I_p)
    \mathbf q_{t,h},
    \quad
    \widehat{\boldsymbol\psi}_{t,h} = \vectorize\!\left[
    \big(\widehat{\mathbf e}_{Y,t,h}
    -\widehat{\mathbf B}_h\widehat{\mathbf e}_{W,t}\big)\,
    \widehat{\mathbf e}_{Z,t}^\top\widehat{\mathbf A}_h^\dagger
    \right]\in\mathbb R^{ps}.
\]
We use the Bartlett kernel \(K(x)=(1-|x|)\I\{|x|\le1\}\) to form a Newey--West estimator \citep{newey1987simple} for simplicity and explicitness, although more general kernels and bandwidths can be used; see, e.g., \citet{andrews1991heteroskedasticity} for their conditions. We use an unweighted HAC estimator because the oracle influence is stationary.
With this kernel $K(\cdot)$ and identity weight, a feasible long-run variance estimator is
\begin{equation}\label{eq:vector-hac}
    \widehat{\boldsymbol\Omega}_h
    = \widehat{\boldsymbol\Gamma}_0
    +\sum_{m=1}^{\ell_n} K\!\left(\frac{m}{\ell_n}\right)
    (\widehat{\boldsymbol\Gamma}_m+\widehat{\boldsymbol\Gamma}_m^\top),
    \quad
    \widehat{\boldsymbol\Gamma}_m
    = \frac1n\sum_{t=m+1}^n
    (\widehat{\boldsymbol\psi}_{t,h}-P_n\widehat{\boldsymbol\psi}_{t,h})
    (\widehat{\boldsymbol\psi}_{t-m,h}-P_n\widehat{\boldsymbol\psi}_{t,h})^\top.
\end{equation}
The truncation bandwidth $\ell_n$ in \eqref{eq:vector-hac} needs to grow with the sample size, but not too fast, so that the bias of the estimator is dominated by its variance. The following condition specifies the behaviour so that the HAC estimator is consistent.
\begin{condition}[Vector plug-in stability]\label{cond:vector-plugin}
In \eqref{eq:vector-hac}, the bandwidth satisfies \(\ell_n\to\infty\) and \(\ell_n/n\to0\), and the plug-in influence vectors satisfy $\ell_n P_n\|\widehat{\boldsymbol\psi}_{t,h}-\boldsymbol\psi_{t,h}^{(T)}\|_2^2=o_p(1)$.
\end{condition}
\begin{remark}[Sufficient conditions from estimation theory]\label{rem:vector-plugin-sufficient}
Under the conditions of \Cref{thm:orthogonal-estimation}, we prove in \Cref{app:wald-inference} that
\(
    P_n\|\widehat{\boldsymbol\psi}_{t,h}-\boldsymbol\psi_{t,h}^{(T)}\|^2
    =O_p\!\left(\bar r_T^2\log n+\log n/n\right).
\)
Consequently, \Cref{cond:vector-plugin} holds whenever $\ell_n\to\infty$,
$\ell_n\bar r_T^2\log n\to0$
and
$\ell_n\log n/n\to0$.
For example, the common choice $\ell_n\asymp n^{1/3}$ is compatible with any $\bar r_T=o(T^{-1/6}\log^{-1/2}T)$.
\end{remark}

For a fixed causal contrast, we obtain the corresponding scalar variance estimate by projecting the HAC matrix onto the contrast direction, i.e.,
$s^2(a,b) := (b \otimes a)^\top \widehat{\boldsymbol\Omega}_h (b\otimes a)$.

With these estimates all set, we offer a modular proposition on the consistency of the feasible long-run variance matrix.
\begin{proposition}[Consistency of the feasible long-run variance]\label{prop:hac-consistent}
Under Assumptions \ref{ass:iv}, \ref{ass:estimation-primitive}, \ref{ass:uq-stationarity} and \Cref{cond:vector-plugin},
$\widehat{\boldsymbol\Omega}_h\pto\boldsymbol\Omega_h$,
and for every fixed contrast $(a,b)$, \(s^2(a,b)\pto\sigma_h^2(a,b)\).
\end{proposition}
The proposition is based on the standard kernel HAC consistency \citep{dejongdavidson2000consistency}, and is the cornerstone for feasible inference; see corollaries below.

\begin{corollary}[Feasible joint inference]\label{cor:vector-wald}
Suppose the conditions of \Cref{thm:vector-clt} and \Cref{prop:hac-consistent} hold, and in addition \(\boldsymbol\Omega_h\) is invertible. Then the Wald statistic satisfies
\[
    n\,\vectorize(\widehat{\mathbf B}_h-\mathbf B_h)^\top
    \widehat{\boldsymbol\Omega}_h^{-1}
    \vectorize(\widehat{\mathbf B}_h-\mathbf B_h)
    \dto \chi^2_{ps},
\]
and an asymptotic $1-\alpha$ joint confidence region for \(\vectorize(\mathbf B_h)\) follows by inverting the corresponding Wald test.
More generally, for any fixed full-row-rank matrix \(\mathbf M\in\mathbb R^{k\times ps}\),
\[
    n\,
    (\mathbf M\vectorize(\widehat{\mathbf B}_h-\mathbf B_h))^\top
    (\mathbf M\widehat{\boldsymbol\Omega}_h\mathbf M^\top)^{-1}
    (\mathbf M\vectorize(\widehat{\mathbf B}_h-\mathbf B_h))
    \dto
    \chi^2_{k}.
\]
\end{corollary}

\begin{corollary}[Feasible studentized inference]\label{cor:studentized-inference}
Suppose the conditions of \Cref{thm:vector-clt} and \Cref{prop:hac-consistent} hold. For a fixed contrast $(a,b)$ with $\sigma_h(a,b)>0$, the statistic satisfies
\[
    \frac{\sqrt n\{\widehat\theta_h(a,b)-\theta_h(a,b)\}}
    {s(a,b)}
    \dto N(0,1).
\]
Consequently, the asymptotic Wald interval for $\theta_h(a,b)$ can be constructed accordingly.
\end{corollary}

\section{Nuisance estimation and rate transfer}\label{sec:nuisance-learners}
All inferential properties above depend on satisfaction of the two conditions on the population and block-sample behaviours of the nuisance errors. Meanwhile, existing regression theorems typically bound population prediction error or excess risk, and it is still lacking in the literature how such population guarantees transfer to the validation blocks under temporal dependence. In light of this, we establish a rate-transfer theorem for the general block cross-fitting scheme, analyze its implications, and apply it to our context. Note that \Cref{cond:sg-nuisance-error} concerns the data-generating mechanism, and thus needs ad-hoc analysis.

Consider a generic regression problem with the predictor and response being square-integrable random vectors \(\{(\mathbf U_t,\mathbf D_t)\}_{t=1}^T\). The population $L_2$ risk of a time-invariant square-integrable map \(\mathbf f\) at time $t$ is $R_t(\mathbf f)=\E\|\mathbf D_t-\mathbf f(\mathbf U_t)\|_2^2$. We define its excess risk at time $t$ as the difference between this risk and the risk of the true conditional mean \(\boldsymbol\mu(\mathbf U_t)=\E[\mathbf D_t\mid \mathbf U_t]\). By the squared-loss excess-risk identity, 
\begin{equation}\label{eq:excess-risk-L2}
    R_t(\mathbf f)-R_t(\boldsymbol\mu_a)
    =\E_{\mathbf U_t}\|\mathbf f(\mathbf U_t)-\boldsymbol\mu(\mathbf U_t)\|_2^2.
\end{equation}

Thus excess risk is also the squared population \(L_2\) error in estimating the conditional mean.
It is random through the fitted map, and we evaluate the risk by holding the map fixed and integrating a fresh observation by the time-\(t\) law, independent of its training data.

Let \(\boldsymbol\mu\) be a deterministic, square-integrable vector-valued function of \(\mathbf U_t\), and write \(\widehat{\boldsymbol\mu}_\ell\) for its estimate trained outside validation fold \(\ell\) and its buffer.
To drop the time index, we define the maximum population and validation errors over the \(L\) folds by
\begin{align*}
    E_{P,T}
    =\max_{t\in\cup_{\ell\le L}\mathcal I_{\ell,h}}
      \E_{\mathbf U_t}
        \|\widehat{\boldsymbol\mu}_\ell(\mathbf U_t)-\boldsymbol\mu(\mathbf U_t)\|_2^2, \quad
    E_{V,T}
    =\max_{\ell\le L}
      \frac{1}{|\mathcal I_{\ell,h}|}
      \sum_{t\in\mathcal I_{\ell,h}}
        \|\widehat{\boldsymbol\mu}_\ell(\mathbf U_t)-\boldsymbol\mu(\mathbf U_t)\|_2^2.
\end{align*}
Here \(E_{P,T}\) measures expected squared error at a fresh state, uniformly over validation times, with each fitted map held fixed. The quantity \(E_{V,T}\) measures empirical mean squared error averaged over the states in each observed validation block. 
In practice, $E_{P,T}$ or its related quantities are often supplied to characterize the performance of the estimators.

\begin{theorem}[Population prediction rates on buffered validation blocks]\label{thm:regression-rate-transfer}
Use the temporal block cross-fitting scheme of \Cref{alg:orthogonal-lpiv}, with \(L\) fixed. Suppose the full observed data sequence has $\beta$-mixing coefficients \(\beta(k)\to0\) as \(k\to\infty\). For each fold \(\ell\), let the learner \(\widehat{\boldsymbol\mu}_\ell(u)\) be a measurable function of its training observations, independent learner randomization, and the input \(u\). Let \(\rho_T\ge0\) be a deterministic rate sequence. Then,
\begin{equation}\label{eq:buffered-rate-transfer-prob}
    \Pr\{E_{V,T}>MK\rho_T^2\}
    \le \Pr(E_{P,T}>K\rho_T^2)+\frac{L}{M}
       +3L\beta(b_T)
\end{equation}
for any deterministic \(M\ge1\) and \(K>0\).
Consequently, under \Cref{ass:cross-fitting},
\begin{equation}\label{eq:buffered-rate-transfer-op}
    E_{P,T}=O_p(\rho_T^2)
    \quad\Longrightarrow\quad
    E_{V,T}=O_p(\rho_T^2).
\end{equation}
\end{theorem}

Under the common stationary condition, $E_{P,T}$ is just the population excess risk, and \Cref{thm:regression-rate-transfer} then allows the order of its rate to be transferred to the empirical version.

In our model \eqref{eq:model} under \Cref{ass:estimation-primitive}(i)--(ii), the full observed sequence at horizon \(h\) is
$\mathcal D_{t,h}:=(\mathbf Y_{t+h},\mathbf W_t,\mathbf Z_t,\mathbf U_t)$,
which inherits the mixing coefficients of the process \(\mathcal X_{t,h}\). Hence, the theorem applies with \(\beta(k)=C_\beta e^{-c_\beta k}\). \Cref{cor:iv-nuisance-rates} below applies the transfer separately to the three nuisance regressions, allowing each learner to have its own rate.

\begin{corollary}[Prediction guarantees for the IV nuisance regressions]\label{cor:iv-nuisance-rates}
Suppose the conditions of \Cref{thm:regression-rate-transfer} hold separately for each target \(\boldsymbol\mu_a\), \(a\in\{Y,W,Z\}\), with fold estimates \(\widehat{\boldsymbol\mu}^{(-\ell)}_a\), and the risks and conditional means above are well defined. Let \(m_{\ell,T}\asymp T\) be the number of observations used to train fold \(\ell\). Suppose that the imported regression theorem applies to these training observations and, for every \(a\in\{Y,W,Z\}\) and \(\ell\le L\), gives
\begin{equation}\label{eq:generic-excess-risk}
    \Pr\left\{
      \sup_{t\in\mathcal I_{\ell,h}}
      \left[R_{a,t}(\widehat{\boldsymbol\mu}^{(-\ell)}_a)
            -R_{a,t}(\boldsymbol\mu_a)\right]
      >C_a^2\rho_a(m_{\ell,T})^2
    \right\}\le\delta_a(m_{\ell,T}),
\end{equation}
where \(C_a>0\) is independent of \(T\) and \(\ell\); \(\rho_a(m)\ge0\) and \(\delta_a(m)\to0\) are deterministic. Let
\[
    \rho_{a,T}=\max_{\ell\le L}\rho_a(m_{\ell,T}),\qquad
    \delta_{a,T}=\sum_{\ell\le L}\delta_a(m_{\ell,T})=o(1),\qquad
    \rho_{Q,T}=\rho_{Y,T}\vee\rho_{W,T},
\]
and \(\rho_{\max,T}=\rho_{Q,T}\vee\rho_{Z,T}\). Then, for all \(a\in\{Y,W,Z\}\), 
\begin{equation}\label{eq:generic-rates}
\begin{aligned}
    \|\boldsymbol\Delta_a\|_{L_2}\vee\|\boldsymbol\Delta_a\|_{n,h,2}
      =O_p(\rho_{a,T}), \quad
    \|\boldsymbol\Delta_Z\|
      \max\{\|\boldsymbol\Delta_{Y,h}\|,\|\boldsymbol\Delta_W\|\}
    =O_p(\rho_{Z,T}\rho_{Q,T}),
\end{aligned}
\end{equation}
where the second equation holds for each of the population \(L_2\) and empirical \((n,h,2)\) norms defined and required in \Cref{cond:nuisance-rate}. Consequently, \Cref{cond:nuisance-rate} holds if
\begin{equation}\label{eq:generic-product-condition}
    \rho_{\max,T}\log T\to0
    \quad\text{and}\quad
    \rho_{Z,T}\rho_{Q,T}=o(T^{-1/2}).
\end{equation}
\end{corollary}

\begin{remark}\label{rem:temporal-rate-transfer}
In i.i.d.\ settings, sample splitting makes a fitted nuisance function independent of its validation data \citep{chernozhukov2018double}, but for a single time series, disjoint folds remain dependent. \Cref{thm:regression-rate-transfer} controls the effect of this dependence through the buffer, without requiring independence among observations within a validation block.
\end{remark}

\Cref{cor:iv-nuisance-rates} uses the actual training sizes \(m_{\ell,T}\) in the imported prediction rates, so rate functions should be evaluated at training sizes. If an imported theorem requires one contiguous training stretch, the learner may use the longer of the two stretches outside the buffer. Under the corollary's training-size requirement, this still leaves order \(T\) observations.

A bound on the expected maximum excess risk, where the expectation is over the training sample and any learner randomization, also yields the population \(O_p\) rate by Markov's inequality, so \Cref{thm:regression-rate-transfer} also applies. However, a bound on raw forecast loss includes irreducible noise and can appear larger, so its transferrability needs case-by-case verification. 

As for a parameter-estimation theorem, e.g., if
\[
    \boldsymbol\mu_a(\mathbf U_t)
    =\boldsymbol\Theta_a^\top\mathbf X_t+\mathbf A_{a,t},
    \qquad
    \widehat{\boldsymbol\mu}_{a,\ell}(\mathbf U_t)
    =\widehat{\boldsymbol\Theta}_{a,\ell}^\top\mathbf X_t,
\]
with $\mathbf A_{a,t}$ being the approximation error, then, with $\boldsymbol\Sigma_{X,t}=\E(\mathbf X_t\mathbf X_t^\top)$, the excess risk
\begin{equation}\label{eq:coefficient-to-prediction}
    \E\!\left[
      \|\widehat{\boldsymbol\mu}_{a,\ell}(\mathbf U_t)
        -\boldsymbol\mu_a(\mathbf U_t)\|_2^2\right]
    \le
    \|\boldsymbol\Sigma_{X,t}
       (\widehat{\boldsymbol\Theta}_{a,\ell}-\boldsymbol\Theta_a)\|_\Fr^2
    +\E\|\mathbf A_{a,t}\|_2^2.
\end{equation}
Thus a coefficient or prediction-error rate and a bound on $\mathbf A_{a,t}$ are sufficient. If the regressors \(\mathbf X_t\) are also estimated, as with factors in various factor models \citep{stockwatson2002forecasting,wang2019matrixfactor,fan2023bridging}, their estimation error should be included as well. 

This framework hence accommodates sparse, low-rank/factor regressions, sieves, splines, kernels, and neural networks for semiparametric inference, whenever their guarantees can control the required population norm and their assumptions hold for the observations used to train each fold. And if $\boldsymbol\mu_Z(\mathbf U_t)$ is an unknown constant, e.g., when the instrument is exogenous to $\mathbf U_t$,  estimated by averaging over $m$ observations, we have $\rho_Z(m)=m^{-1/2}$, and the requirements for other regressions are significantly relaxed; see \Cref{supp:rem:const-instrument}. 

Before referring to concrete examples in \Cref{sec:nuisance-examples}, we comment briefly on high dimensionality of nuisance functions. Generally speaking, for a parametric linear model, the prediction error is $O_p(\sqrt{d/T})$ with $d$ parameters. Obtaining the convenient rate $o_p(T^{-1/4})$ thus requires $d=o(T^{1/2})$, which guides the design of low-dimensional modeling assumptions.
For a generic nonparametric task, the minimax rate is $O_p(T^{-\beta/(2\beta+q_T)})$, where $\beta$ measures smoothness. The requirement $\beta>q_T/2$ alone can rule out high dimensionality because $\beta$ is fixed while $q_T$ grows with $T$, so materially low-dimensional structure is also needed to reduce the effective complexity of the nuisance regression, just like parametric models, besides smoothness assumptions such as Sobolev or H\"older smoothness.

\section{Simulation studies}\label{sec:simulations}
We compare estimators of a dynamic causal effect when a high-dimensional state has a factor structure. We examine point estimation, interval coverage, and sensitivity to nuisance specification and instrument quality. \Cref{app:simulation-details} gives the full designs and additional results.

\subsection{Design}
\label{subsec:simulation-design}

The main design has \(T=500\) observations, \(q=1000\) state variables, six factors, and horizon \(h=4\). The observed state is
$U_t=(\Lambda F_t+E_t)/{\sqrt7}$ with $\Lambda^{\mathsf T}\Lambda/q=I_6$,
where the independent Gaussian factor and idiosyncratic coordinates have unit stationary variance and AR(1) coefficients 0.8 and 0.2. The nuisance functions depend on normalized observed indices \(S_t=B^{\mathsf T}U_t\), with \(\operatorname{Var}(S_t)=I_6\), so their coefficients in the original \(q\) coordinates are dense. We generate univariate $Z_t = g^\top S_t + \zeta_t$, $W_t = d^\top S_t + 0.6\zeta_t + 0.8C_t + 0.7\nu_t$, and
\begin{align*}
 Y_{t+h}^{(h)}=\beta_hW_t+\gamma^{\mathsf T}S_t+0.8C_t
 +\frac{0.8}{\sqrt{h+1}}\sum_{j=0}^h e_{t+j}\in \mathbb R,
 \qquad \beta_h=0.8e^{-h/8} .
\end{align*}
Multivariate case is explored in the supplement.
Here \(C_t,\nu_t,e_t\) are mutually independent standard Gaussian white-noise sequences. The residual instrument follows
\[
 \zeta_t=0.8\zeta_{t-1}+0.6\{aC_{t-1}+\sqrt{1-a^2}\,\xi_t\},
 \qquad a=0.9,
\]
with an independent standard Gaussian innovation \(\xi_t\). The instrument responds to past confounding shocks but is independent of the current structural error, so that conditional exclusion holds. We also consider \(a=0\) and \(a=0.5\) to assess the effect of this feedback.

We use ten contiguous validation folds with temporal buffers. Within each training fold, the estimator selects the factor count by \(\operatorname{IC}_{p2}\) criterion in \citet{baing2002determining} and fits the three nuisance functions by principal-component regression (PCR). Standard errors use Bartlett HAC with OLS sensitivity adjustment that holds the selected factor spaces fixed. We first conduct a baseline comparison with the above design with shared random numbers. The shuffled iid DML--IV uses the same PCR learner and ten randomly assigned folds, treating the data as iid.  
We then vary sample size, instrument feedback, and nuisance functional form to test our model's sensitivity to these features. All settings have 1000 replications.

\subsection{Results}
\begin{table}[!t]
\centering
\caption{Estimation of $\beta_4=0.485$ at $T=500$, $q=1000$, and six factors. Coverage is for nominal 95\% intervals; length is the mean interval length.}
\label{tab:simulation-main}
\small
\setlength{\tabcolsep}{4pt}
\begin{tabular}{lrrrr}
\toprule
Estimator & Bias & RMSE & Coverage & Length \\
\midrule
Proposed PCR--IV & 0.014 & 0.135 & 0.935 & 0.527 \\
Shuffled iid PCR--IV & -0.042 & 0.145 & 0.810 & 0.363 \\
Desparsified HDLP--IV ratio & 0.016 & 0.137 & 0.879 & 0.425 \\
Conventional LP--IV, 4 controls & 0.520 & 0.568 & 0.227 & 0.748 \\
Oracle-nuisance IV & -0.011 & 0.130 & 0.916 & 0.462 \\
LP + PCR without IV & 0.423 & 0.427 & 0.000 & 0.178 \\
\bottomrule
\end{tabular}
\end{table}
\Cref{tab:simulation-main} reports bias, RMSE, and coverage of nominal 95\% intervals for the structural effect. The proposed estimator has bias 0.014 and RMSE 0.135, compared with \(-0.042\) and 0.145 for shuffled iid PCR. Coverage is 0.935 and 0.810, respectively, with mean interval lengths of 0.527 and 0.363.
The desparsified external-IV ratio, which is a combination of \citet{adamek2024local} and \citet{breunig2020illposed}, has RMSE 0.137 and coverage 0.879. Conventional LP--IV with four prespecified controls leaves substantial factor confounding and has bias 0.520. We also add an ablation without IV, with the same buffered approach and an ordinary GMM estimator, and construct intervals following \citet{ciganovic2026double}, and find zero coverage. 


\Cref{fig:simulation-main}(a) and (d) vary sample size with \(q=1000\). From \(T=250\) to 2000, the proposed estimator's bias falls from 0.020 to 0.003 and its RMSE from 0.209 to 0.066. Coverage ranges from 0.933 to 0.940 across these four sample sizes. The corresponding shuffled iid coverage worsens with sample size, indicating the importance of accounting for temporal dependence. 

Panel (b) is the experiment on instrument feedback, i.e., how the instrument responds to past confounding shocks. At \(a=0\) where the effect of temporal dependence is relatively petty, shuffled iid PCR has lower RMSE with near-identical bias. When $a \ge 0.5$, our method gains the advantage. The comparison changes little with the state dimension. 

Panel (c) adds centered quadratic functions of the factor indices to the nuisance means to show the importance of the capability to handle various nuisance functions. Population linear residualization then targets \(\beta_4+\delta_Z\delta_Y/(\sqrt2\pi)\). Quadratic PCR augments the estimated factor scores with all pairwise products, including squares, and estimates their coefficients within each training fold. At \(\delta_W=\delta_Z=\delta_Y=0.5\), it reduces bias from 0.320 to 0.037 and RMSE from 0.381 to 0.153, and raises coverage rate from 0.598 to 0.948. The instrument-quality experiments in \Cref{app:simulation-details} vary first-stage strength and introduce a direct instrument effect violating exclusion assumption to show the importance of instrument validity. 
\begin{figure}[t]
 \centering
 \includegraphics[width=\textwidth]{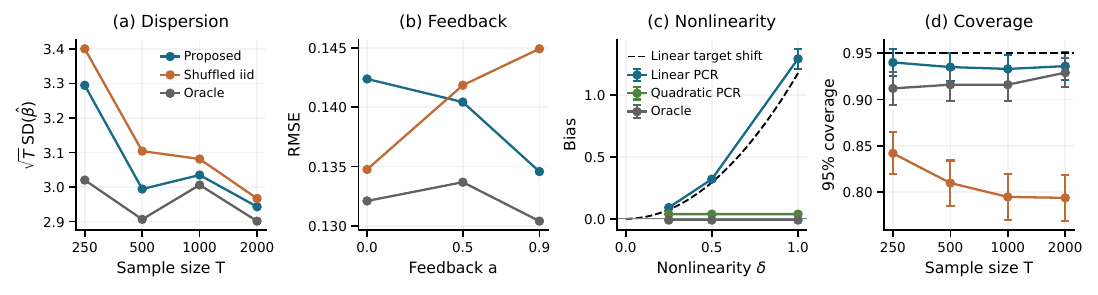}
 \caption{Simulations under factor models with instrument feedback, using 1000 replications per setting. Panels (a) and (d) vary sample size with $q=1000$; panel (b) varies instrument feedback; panel (c) varies a common nonlinear nuisance magnitude. }
 \label{fig:simulation-main}
\end{figure}

\section{Empirical application}\label{sec:application}
We estimate how U.S. macroeconomic outcomes respond to an instrumented change in monetary policy, which responds to a broad information set and to latent economic conditions, so this application calls for both high-dimensional adjustment and an external instrument. We present main results here, with additional details and diagnostics in \Cref{app:empirical-details}. 

\subsection{Data and estimation}\label{subsec:application-data}
\label{subsec:application-estimation}

We combine the April 2026 vintage of FRED-MD \citep{mccracken2016fred} with the updated Romer--Romer monetary-policy shocks \citep{romer2004new,wieland2021updated}. The instrument \(Z_t=\texttt{resid\_full}_t\) covers 468 months from January 1969 through December 2007. They are available at
\url{https://www.stlouisfed.org/research/economists/mccracken/fred-databases} and \url{https://www.icpsr.umich.edu/sites/icpsr/view/studies/135741/versions/V1.0} respectively. Treatment is the monthly change in the one-year Treasury yield, \(W_t=\Delta\texttt{GS1}_t\). Under the linear response model and conditional IV restrictions in \Cref{ass:iv}, \(\beta_h\) measures the response to a one-percentage-point instrumented increase in this yield change at horizons \(h=0,\ldots,24\). 

We study log housing starts, monthly industrial-production log growth, the monthly change in unemployment, and the second difference of log CPI, which
measures inflation acceleration. The state contains 12 lags of 122 transformed FRED-MD series, giving 1,464 pre-treatment controls. We compare Lasso nuisance regressions with regressions on eight principal components of this state. Both learners use five contiguous cross-fitting folds with buffer \(b_h=\max\{12,h+1\}\), with standardization and factor extraction using only each fold's training observations. The residuals enter the orthogonal IV moment in \Cref{alg:orthogonal-lpiv}. We report pointwise 95\% intervals using Bartlett HAC standard errors with lag \(b_h\). The Lasso residualized first-stage HAC \(F\)-statistics range from 38.8 to 53.8.

\subsection{Responses and empirical assessment}\label{subsec:application-results}
\label{subsec:application-diagnostics}

The housing response is the clearest finding in \Cref{fig:application-baseline}. With Lasso adjustment, housing starts decline by \(0.155\) log points at six months and by \(0.168\) at 12 months, both with a negative interval. Factor adjustment gives responses of \(-0.145\) and \(-0.208\), respectively, with both intervals also below zero. The factor fit produces a deeper, more persistent contraction, while ridge and a reduced control set also retain significant six- and twelve-month responses. We also find that changing the HAC bandwidth $\ell_n$ in \eqref{eq:vector-hac} does not affect the significance. 

\begin{figure}[t]
    \centering
    \includegraphics[width=0.94\textwidth]{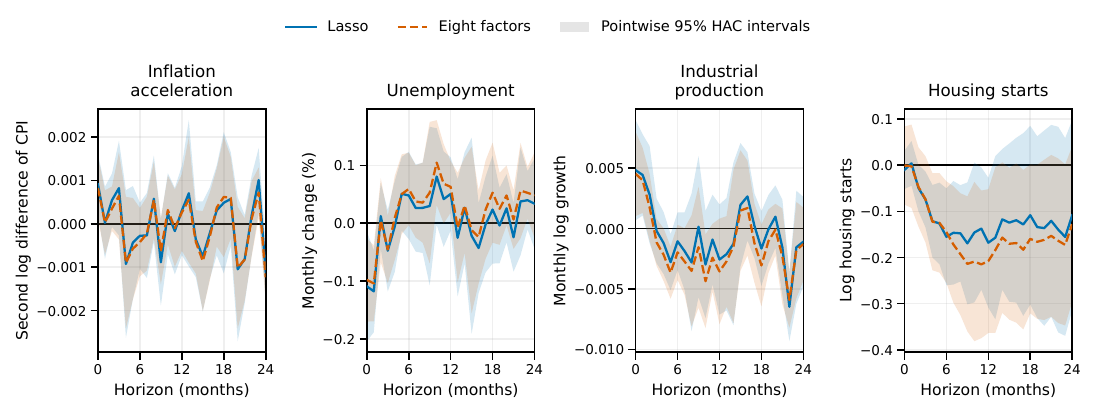}
    \caption{Monetary-policy responses with Lasso and principal-component nuisance
    regressions. Each curve gives the response to a one-percentage-point
    instrumented increase in the monthly change in the one-year Treasury yield;
    shading gives pointwise 95\% Bartlett HAC intervals with lag
    \(b_h=\max\{12,h+1\}\).}
    \label{fig:application-baseline}
\end{figure}


Most intervals for inflation acceleration and production growth include zero under both learners. Unemployment is more sensitive to nuisance specification.
At 12 months, the factor estimate is \(0.063\) percentage points with interval \([0.018,0.108]\), whereas Lasso gives \(0.051\) with interval \([-0.013,0.114]\). The housing response is therefore more stable.

An interesting pattern is the early responses of unemployment and production growth, which conflict with monetary contraction. 
A plausible explanation is the exceptional influence of the monetary-policy reversal in April 1980. The credit controls in the preceding month had unexpected adverse effects on consumer spending which prolonged to subsequent months of easing, and the downturn caused by such consumer behavior remains as the outcome disturbance; see, e.g., \citet{schreft1990credit}.
Sensitivity analyses reveal that omitting the spring 1980 data removes the pattern. This is consistent with \citet{romer2004new}. 

We also added a plasmode experiment which retains the observed treatment, instrument, and state, and generates synthetic housing outcomes with a sparse state signal and errors sampled in 12-month blocks. Coverage under a housing-like effect path stays above 95.5\%, with absolute bias below \(0.04\) log points. False rejection under a zero effect ranges from 1.0\% to 6.0\%. The plasmode results support the finite-sample reliability of our statistical inference procedures in a high-dimensional time-series setting calibrated to the empirical application. \Cref{app:empirical-details} reports the full design and comparison with other nuisance and IV.

\section{Conclusion and discussion}\label{sec:conclusion}

We develop semiparametric inference for horizon-specific causal effects of an endogenous treatment observed along a single time series via instrumental variables. The proposed procedure residualizes the outcome, treatment, and instrument against a high-dimensional pre-treatment state, combines block cross-fitting with a temporal buffer, and estimates the causal response through a Neyman-orthogonal IV moment. We develop non-asymptotic estimation error bound for the estimand and feasible HAC inference for fixed-horizon causal matrices and contrasts. The buffered rate-transfer argument shows when existing prediction guarantees for a generic temporal learner can imply the nuisance conditions an orthogonal estimator generally requires.
In the monetary-policy application, the clearest empirical pattern is a decline in housing starts at medium horizons. 

The present analysis also leaves limitations for future work. First, the identification and inference theories require the average first stage to be well conditioned (or over-identified), so weak-IV-robust tests would be a future direction. Second, the current Gaussian approximation is horizon-wise. Simultaneous inference over a growing collection of horizons that accounts for cross-horizon dependence and HAC studentization would be an important topic. Third, estimation pools observations under time-invariant conditional mean functions. Allowing the nuisance functions and score distribution to drift smoothly, or to change across structural breaks, would broaden the method to more nonstationary environments.

\putbib
\clearpage
\end{bibunit}

\begin{bibunit}
\bibliographypart{supplement}
\supplementspacing
\setcounter{section}{0}
\setcounter{algorithm}{0}
\setcounter{remark}{0}
\setcounter{table}{0}
\setcounter{condition}{0}
\setcounter{innercustom}{0}
\setcounter{footnote}{0}
\renewcommand{\theHsection}{S.\arabic{section}}
\renewcommand{\theHalgorithm}{S.\arabic{algorithm}}
\renewcommand{\theHremark}{S.\arabic{remark}}
\renewcommand{\theHtable}{S.\arabic{table}}

\renewcommand{\thesection}{\Alph{section}}
\numberwithin{equation}{section}
\numberwithin{lemma}{section}
\numberwithin{theorem}{section}
\numberwithin{assumption}{section}
\numberwithin{proposition}{section}
\numberwithin{definition}{section}
\numberwithin{corollary}{section}
\numberwithin{figure}{section}
\numberwithin{example}{section}
\renewcommand\thealgorithm{S.\arabic{algorithm}}
\renewcommand\theremark{S.\arabic{remark}}

\renewcommand{\theHequation}{\theHsection.\arabic{equation}}
\renewcommand{\theHlemma}{\theHsection.\arabic{lemma}}
\renewcommand{\theHtheorem}{\theHsection.\arabic{theorem}}
\renewcommand{\theHassumption}{\theHsection.\arabic{assumption}}
\renewcommand{\theHproposition}{\theHsection.\arabic{proposition}}
\renewcommand{\theHdefinition}{\theHsection.\arabic{definition}}
\renewcommand{\theHcorollary}{\theHsection.\arabic{corollary}}
\renewcommand{\theHfigure}{\theHsection.\arabic{figure}}
\renewcommand{\theHexample}{\theHsection.\arabic{example}}

\title{Supplementary Material for ``\papertitle''}
\pretitle{\centering\LARGE\setstretch{1.15}}
\posttitle{\par\vspace{2em}}
\preauthor{\centering\large\setstretch{1}}
\postauthor{\par\vspace{2em}}
\predate{}
\postdate{}
\renewcommand{\maketitlehooka}{\setlength{\parskip}{0pt}}
\renewcommand{\maketitlehookd}{\vspace{-1em}}

\makeatletter
\renewcommand{\@author}{\Authfont\authorlist}
\let\@thanks\@empty
\makeatother
\date{}

\setlength{\parindent}{0cm}
\setlength{\parskip}{12pt}
\titlespacing{\section}{0pt}{5pt}{-5pt}
\titlespacing{\subsection}{0pt}{5pt}{-5pt}

\maketitle
\begin{abstract}
	\vspace{-2mm}
	This supplement provides proofs and details for the main paper. In \Cref{app:model-setup}, we prove population identification and relate the causal target to structural impulse responses. \Cref{app:orthogonal-estimation-proof} establishes Neyman orthogonality and an explicit non-asymptotic oracle bound, with supporting concentration inequalities for dependent moment products, and derives consistency rates and an asymptotic linear expansion. Under the stated inference conditions, \Cref{sec:inference-proof} proves joint asymptotic normality at each fixed horizon, consistency of long-run covariance estimation, and validity of Wald inference. We prove the buffered transfer of temporal prediction guarantees to nuisance errors on validation blocks in \Cref{app:nuisance-rate-transfer}, and specialize the nuisance conditions to stationary principal-component regression and sparse spline regression in \Cref{sec:nuisance-examples}. \Cref{app:simulation-details} specifies simulation designs, estimator comparisons, and diagnostics for calibration, nonlinear misspecification, and instrument quality. \Cref{app:empirical-details} documents the monetary-policy data and estimation procedure, sensitivity analyses, and plasmode calibration.

    For easy reference, \Cref{prop:identification} is proved in \Cref{app:identification-proof}, 
    \Cref{thm:orthogonal-estimation} in \Cref{app:subsec:est-error-proof}, and \Cref{thm:vector-clt} in \Cref{app:vector-clt-proof}.
    \Cref{prop:hac-consistent,cor:vector-wald,cor:studentized-inference,rem:vector-plugin-sufficient} are proved in \Cref{app:wald-inference}. \Cref{thm:regression-rate-transfer} is proved in \Cref{app:regression-rate-transfer-proof}, and \Cref{cor:iv-nuisance-rates} in \Cref{app:iv-nuisance-rates-proof}.

\end{abstract}
\setcounter{tocdepth}{2}
\begin{spacing}{0.4}
    \section*{\contentsname}
    \makeatletter
    \input{supplement.toc}
    \makeatother
\end{spacing}
\section{Identification and related frameworks}\label{app:model-setup}
\subsection{Proof of Identification}\label{app:identification-proof}
\begin{proof}[Proof of \Cref{prop:identification}]
Fix $h\in\mathcal H$ and an arbitrary $t\in\mathcal T_h$. From \eqref{eq:model},
\[
    \mathbf Y_{t+h} = \mathbf B_h\mathbf W_t+\mathbf g_{h,t}(\mathbf U_t)+\boldsymbol\varepsilon_{t+h}, \qquad \E[\boldsymbol\varepsilon_{t+h}\mid \mathbf U_t] = \mathbf 0.
\]
Taking conditional expectations given $\mathbf U_t$ gives
\(
    \boldsymbol\mu_{Y,h,t}(\mathbf U_t) = \mathbf B_h\boldsymbol\mu_{W,t}(\mathbf U_t)+\mathbf g_{h,t}(\mathbf U_t)
\).
Recall the definition in \eqref{eq:residuals}. Subtracting this conditional mean from \eqref{eq:model} yields
\begin{equation}\label{eq:supp-residualized-model}
    \widetilde{\mathbf Y}_{t+h} = \mathbf B_h\widetilde{\mathbf W}_t+\boldsymbol\varepsilon_{t+h}.
\end{equation}
Multiplying \eqref{eq:supp-residualized-model} on the right by $\widetilde{\mathbf Z}_t^\top$ and taking expectations gives
\begin{equation}\label{eq:supp-moment-before-exclusion}
    \mathbf M_{h,t} = \mathbf B_h\mathbf A_t
    +\E[\boldsymbol\varepsilon_{t+h}\widetilde{\mathbf Z}_t^\top].
\end{equation}
Since \(\widetilde{\mathbf Z}_t = \mathbf Z_t-\boldsymbol\mu_{Z,t}(\mathbf U_t)\) is measurable with respect to \(\sigma(\mathbf Z_t,\mathbf U_t)\), \Cref{ass:iv}(i) gives
\begin{equation} \label{eq:supp-exclusion}
    \E[\boldsymbol\varepsilon_{t+h}\widetilde{\mathbf Z}_t^\top] = \E\big[\E[\boldsymbol\varepsilon_{t+h}\widetilde{\mathbf Z}_t^\top\mid \mathbf Z_t,\mathbf U_t]\big] = \E\left[\E[\boldsymbol\varepsilon_{t+h}\mid \mathbf Z_t,\mathbf U_t]\widetilde{\mathbf Z}_t^\top\right] = \mathbf 0,
\end{equation}
which leads to \eqref{eq:identification-moment} for every $t\in\mathcal T_h$.

Horizontally concatenating these equations gives
\[
    \boldsymbol{\mathcal M}_{h,T} = \mathbf B_h\boldsymbol{\mathcal A}_{h,T}.
\]
By \Cref{ass:iv}(ii), $\boldsymbol{\mathcal A}_{h,T}$ has full row rank, i.e.,
\(\boldsymbol{\mathcal A}_{h,T}\boldsymbol{\mathcal A}_{h,T}^{\dagger} = \mathbf I_s\).
Right-multiplication by \(\boldsymbol{\mathcal A}_{h,T}^{\dagger}\) therefore gives
\[
    \boldsymbol{\mathcal M}_{h,T}\boldsymbol{\mathcal A}_{h,T}^{\dagger} = \mathbf B_h\boldsymbol{\mathcal A}_{h,T}\boldsymbol{\mathcal A}_{h,T}^{\dagger} = \mathbf B_h,
\]
which proves \eqref{eq:identification-formula}.
\end{proof}

\subsection{Relationship with structural VAR/MA models} \label{app:relationship-proof}

The following example connects our causal response matrix to structural impulse responses in a structural vector moving-average (SVMA) model, including those generated by stable structural vector autoregressions (SVARs). In \eqref{eq:svma-wold}, $\mathbf X_t$ stacks the $s$ treatments $\mathbf W_t$ and the $p$ outcomes $\mathbf Y_t$, while $\bm\Theta_j$ maps structural shocks to observables $j$ periods later. The shock blocks $\bm\zeta_t^w$ and $\bm\zeta_t^o$ denote the target shocks and all remaining shocks, respectively, and $\bm\Theta_{W,0}$ and $\bm\Theta_{Y,h}$ record the target shocks' effects on the current treatment and the outcome at horizon $h$. The $r$ external proxies $\mathbf Z_t$ measure the target shocks through the loading matrix $\bm\Lambda$, with measurement noise $\bm\nu_t$, while $\mathbf U_t$ records observed information from before period $t$. 
The convention of structural impulse responses is to impose intervention on the target shocks instead of the treatment, and to report the structural response normalized by the contemporaneous treatment impact. And under the conditions below, such interventions yield exactly $\mathbf B_h$, which is the causal contrast in \eqref{eq:causal-contrast} and the coefficient of interest in \eqref{eq:model}.

\begin{example}[Proxy-SVMA and stable proxy-SVARs]\label{ex:svma}
Let $\mathbf X_t=(\mathbf W_t^\top,\mathbf Y_t^\top)^\top$ follow the stable SVMA \eqref{eq:svma-wold}, driven by i.i.d.\ mean-zero shocks $\bm\zeta_t=((\bm\zeta_t^{w})^\top,(\bm\zeta_t^{o})^\top)^\top$. Suppose $\bm\zeta_t^{w}\in\R^s$ with $\var(\bm\zeta_t^{w}) \succ 0$ is independent of $\bm\zeta_t^{o}$, and let its impacts on $\mathbf W_t$ and $\mathbf Y_{t+h}$ be $\bm\Theta_{W,0}\in\R^{s\times s}$ and $\bm\Theta_{Y,h}\in\R^{p\times s}$, respectively, with $\bm\Theta_{W,0}$ nonsingular. Suppose the instrument is a proxy for the target shocks, bearing the form
\[
    \mathbf Z_t=\bm\Lambda\bm\zeta_t^{w}+\bm\nu_t \in \mathbb R^r, \quad
    \E[\bm\nu_t] = \mathbf 0, \qquad
    \rank(\bm\Lambda)=s \le r,
\]
where $\bm\nu_t$ is proxy noise independent of the entire shock process $\{\bm\zeta_s:s\in\mathbb Z\}$. Consider any observed state of the form $\mathbf U_t=\boldsymbol\phi(\bm\zeta_{t-1},\bm\zeta_{t-2},\ldots)$ for a fixed measurable map $\boldsymbol\phi$ of the pre-$t$ shock history, such as the lagged state of a stable SVAR. When an intervention on $\mathbf W_t$ is implemented through the target shock $\bm\zeta_t^w$, the SVMA induces the potential-outcome model \eqref{eq:po-linear-response}, conditional mean exchangeability \eqref{eq:exchangeability} given \(\mathcal C_t\), and a time-invariant nuisance $\mathbf g_{h,t} \equiv \mathbf g_h$ for which \eqref{eq:model} and the conditional IV restrictions hold, with
\begin{equation}\label{eq:svma-target}
    \mathbf B_h=\bm\Theta_{Y,h}\bm\Theta_{W,0}^{-1},
\end{equation}
which can be recovered via \Cref{prop:identification}. 
\end{example}
Consistent with \citet{plagborg2021local}, linear LP can target the same population impulse responses as VAR representations, but we differ in where dimension-reducing assumptions are imposed. As shown in \Cref{sec:estimation}, our LP-IV framework localizes high-dimensional modeling to the three conditional-mean nuisance functions needed for $\mathbf B_h$ and estimates each horizon directly, rather than imposing a sparse or low-rank finite-order transition law for the entire system and iterating it. This can reduce exposure to system-wide dynamic misspecification and allow much more flexible nuisance learners, as \(\mathbf g_{h,t}(\cdot)\) is no longer constrained to be part of VAR/VMA transition functions. 
\begin{proof}[Proof of \Cref{ex:svma}]
Let \(\mathcal F_{t-1}^{\zeta}=\sigma(\bm\zeta_\ell:\ell\le t-1)\). As in the example statement, the observed state takes the form \(\mathbf U_t=\boldsymbol\phi(\bm\zeta_{t-1},\bm\zeta_{t-2},\ldots)\) for a fixed measurable map \(\boldsymbol\phi\), so \(\mathbf U_t\) is measurable with respect to this field. Write
$\bm\Sigma_w=\E[\bm\zeta_t^w\bm\zeta_t^{w\top}]$.

We first verify the potential-outcome model of \Cref{eq:po-linear-response}. Isolate the period-\(t\) target shock as
\begin{equation}\label{eq:supp-svma-decomp}
    \mathbf W_t=\bm\Theta_{W,0}\bm\zeta_t^w+\mathbf Q_t,
    \qquad
    \mathbf Y_{t+h}=\bm\Theta_{Y,h}\bm\zeta_t^w+\mathbf R_{t,h},
\end{equation}
where \(\mathbf Q_t\) and \(\mathbf R_{t,h}\) collect all remaining shock terms of the SVMA. By construction, \(\mathbf Q_t\) is a measurable function of \(\bm\zeta_t^o\) and the pre-\(t\) shocks, while \(\mathbf R_{t,h}\) is a measurable function of \(\bm\zeta_t^o\), the pre-\(t\) shocks, and the future shocks \(\bm\zeta_{t+1},\ldots,\bm\zeta_{t+h}\). The defining series converge almost surely and in \(L_2\) under the stability of the SVMA and the blanket square integrability.

Take the pre-assignment information field
\[
    \mathcal C_t=\sigma\big(\{\bm\zeta_s:s\le t-1\},\ \bm\zeta_t^o\big),
\]
which contains \(\sigma(\mathbf U_t)\), and note that \(\mathbf Q_t\) is \(\mathcal C_t\)-measurable. As is standard in proxy-SVAR analysis, we interpret an intervention on the treatment as implemented through the target shock. Externally setting \(\mathbf W_t\) to \(w\) means replacing \(\bm\zeta_t^w\) by \(\bm\Theta_{W,0}^{-1}(w-\mathbf Q_t)\) while leaving all other shocks unchanged, which is well defined because \(\bm\Theta_{W,0}\) is nonsingular. This is an interpretive assumption of the proxy-SVAR framework rather than a consequence of the SVMA, since \(\mathbf W_t\) may also load on \(\bm\zeta_t^o\). The induced potential outcome is
\[
    \mathbf Y_{t+h}(w)
    =\bm\Theta_{Y,h}\bm\Theta_{W,0}^{-1}(w-\mathbf Q_t)+\mathbf R_{t,h}
    =\mathbf B_h w+\mathbf R_{t,h}-\mathbf B_h\mathbf Q_t,
    \quad \text{where }
    \mathbf B_h:=\bm\Theta_{Y,h}\bm\Theta_{W,0}^{-1}.
\]
Evaluating at \(w=\mathbf W_t\) recovers the observed outcome by \eqref{eq:supp-svma-decomp}, so the consistency condition holds. By the linearity of the SVMA, split \(\mathbf R_{t,h}=\E[\mathbf R_{t,h}\mid\mathcal C_t]+\boldsymbol\eta_{t+h}\), where \(\boldsymbol\eta_{t+h}\) collects exactly the terms of \(\mathbf R_{t,h}\) that involve the shocks \(\bm\zeta_{t+1},\ldots,\bm\zeta_{t+h}\). Since the shocks are i.i.d., \(\boldsymbol\eta_{t+h}\) is independent of \(\sigma(\mathcal C_t,\bm\zeta_t^w)\) and has mean zero. Therefore,
\[
    \E[\mathbf Y_{t+h}(w)\mid\mathcal C_t]
    =\mathbf B_h w+\mathbf m_h(\mathcal C_t),
    \qquad
    \mathbf m_h(\mathcal C_t)
    =\E[\mathbf R_{t,h}\mid\mathcal C_t]-\mathbf B_h\mathbf Q_t,
\]
which verifies \eqref{eq:po-linear-response}, and the residual \(\boldsymbol\eta_{t+h}(w)=\mathbf Y_{t+h}(w)-\E[\mathbf Y_{t+h}(w)\mid\mathcal C_t]=\boldsymbol\eta_{t+h}\) does not depend on \(w\). Because \(\mathbf W_t\) is measurable with respect to \(\sigma(\mathcal C_t,\bm\zeta_t^w)\), the independence above gives conditional mean exchangeability, expressed as \(\E[\boldsymbol\eta_{t+h}(w)\mid\mathcal C_t,\mathbf W_t]=\mathbf 0\) for every \(w\). The causal contrast in \eqref{eq:causal-contrast} is then \(\mathbf B_h\delta\), establishing the structural interpretation in \eqref{eq:svma-target}.

We next construct the LP representation. Let
\[
    \boldsymbol\mu_{W,t}(\mathbf U_t)=\E[\mathbf W_t\mid\mathbf U_t],
    \qquad
    \boldsymbol\mu_{Y,h,t}(\mathbf U_t)=\E[\mathbf Y_{t+h}\mid\mathbf U_t],
\]
be the forecasts under the true SVMA, and define
\begin{equation}\label{eq:supp-svma-nuisance}
    \mathbf g_{h,t}(\mathbf U_t)
    =\boldsymbol\mu_{Y,h,t}(\mathbf U_t)
     -\mathbf B_h\boldsymbol\mu_{W,t}(\mathbf U_t).
\end{equation}
Then \eqref{eq:model} holds algebraically with $\boldsymbol\varepsilon_{t+h}=\widetilde{\mathbf Y}_{t+h}-\mathbf B_h\widetilde{\mathbf W}_t$ and $\E[\boldsymbol\varepsilon_{t+h}\mid\mathbf U_t]=\mathbf 0$.

This nuisance is time-invariant. The shocks are i.i.d.\ and the SVMA coefficients do not depend on \(t\), so the joint law of \((\mathbf W_t,\mathbf Y_{t+h},\bm\zeta_{t-1},\bm\zeta_{t-2},\ldots)\) does not depend on \(t\). Since \(\mathbf U_t=\boldsymbol\phi(\bm\zeta_{t-1},\bm\zeta_{t-2},\ldots)\) for the fixed map \(\boldsymbol\phi\), the joint law of \((\mathbf W_t,\mathbf Y_{t+h},\mathbf U_t)\) is also free of \(t\). The conditional-mean functions \(\boldsymbol\mu_{W,t}\) and \(\boldsymbol\mu_{Y,h,t}\) are determined by these joint laws, so they can be chosen independent of \(t\), and hence \(\mathbf g_{h,t}\equiv\mathbf g_h\) in \eqref{eq:supp-svma-nuisance}.

We now turn to the validity of the proxy $\mathbf Z_t$ in terms of \Cref{ass:iv}.
To verify conditional exclusion, recall the decomposition \eqref{eq:supp-svma-decomp}. The remainders \(\mathbf Q_t\) and \(\mathbf R_{t,h}\) are measurable functions of the other-period shocks \(\{\bm\zeta_s:s\ne t\}\) together with the contemporaneous non-target block \(\bm\zeta_t^o\), and \(\mathbf U_t\) is measurable with respect to \(\{\bm\zeta_s:s\le t-1\}\). The within-period independence of \(\bm\zeta_t^w\) and \(\bm\zeta_t^o\) and the independence of \(\bm\zeta_t=(\bm\zeta_t^{w\top},\bm\zeta_t^{o\top})^\top\) across time imply that \(\bm\zeta_t^w\) is independent of the shock collection generating \((\mathbf Q_t,\mathbf R_{t,h},\mathbf U_t)\); joining the proxy noise \(\bm\nu_t\), which is independent of the whole shock process, then makes the pair
\((\bm\zeta_t^w,\bm\nu_t)\) independent of
\((\mathbf Q_t,\mathbf R_{t,h},\mathbf U_t)\). Since \(\E[\bm\zeta_t^w\mid\mathbf U_t]=\mathbf0\),
\[
    \boldsymbol\varepsilon_{t+h}
    = \{\bm\Theta_{Y,h}-\mathbf B_h\bm\Theta_{W,0}\}\bm\zeta_t^w
    +\{\mathbf R_{t,h}-\E[\mathbf R_{t,h}\mid\mathbf U_t]\}
    -\mathbf B_h\{\mathbf Q_t-\E[\mathbf Q_t\mid\mathbf U_t]\}.
\]
The first term is zero by the definition of \(\mathbf B_h\). The remaining terms have conditional mean zero given \(\mathbf U_t\) and are conditionally independent of \(\mathbf Z_t=\bm\Lambda\bm\zeta_t^w+\bm\nu_t\). Therefore, we establish that $\E[\boldsymbol\varepsilon_{t+h}\mid\mathbf Z_t,\mathbf U_t]=\mathbf 0$.

Because the proxy is mean zero and independent of \(\mathbf U_t\), \(\widetilde{\mathbf Z}_t=\mathbf Z_t\). The first-stage and reduced-form moments are
\[
    \mathbf A
    =\E[\widetilde{\mathbf W}_t\widetilde{\mathbf Z}_t^\top]
    =\bm\Theta_{W,0}\bm\Sigma_w\bm\Lambda^\top,
    \qquad
    \mathbf M_h
    =\E[\widetilde{\mathbf Y}_{t+h}\widetilde{\mathbf Z}_t^\top]
    =\bm\Theta_{Y,h}\bm\Sigma_w\bm\Lambda^\top
    =\mathbf B_h\mathbf A.
\]
The assumptions on \(\bm\Theta_{W,0}\), \(\bm\Sigma_w\), and \(\bm\Lambda\) imply that \(\mathbf A\) has full row rank. Hence \(\mathbf A\mathbf A^\dagger=\mathbf I_s\), and \Cref{prop:identification} yields
\(
    \mathbf M_h\mathbf A^\dagger
    =\mathbf B_h\mathbf A\mathbf A^\dagger
    =\mathbf B_h
    =\bm\Theta_{Y,h}\bm\Theta_{W,0}^{-1}
\), so that \eqref{eq:svma-target} holds both as a structural impulse response and as the causal contrast in \eqref{eq:causal-contrast} identified in \eqref{eq:identification-formula}.
\end{proof}

For completeness, we consider, as an example, a stable SVMA generated by a known stable SVAR(\(d\))
\[
    \mathbf X_t=\sum_{\ell=1}^d\bm\Phi_\ell\mathbf X_{t-\ell}+\mathbf C \bm\zeta_t.
\]
Let \(\mathbf S_{t-1}=(\mathbf X_{t-1}^\top,\ldots,\mathbf X_{t-d}^\top)^\top\), let \(\mathbf F\) be the companion matrix, and let \(\mathbf J_W\) and \(\mathbf J_Y\) select the contemporaneous treatment and outcome blocks from the companion state. Concetely,
\begin{equation*}
    \mathbf F=
        \begin{pmatrix}
            \bm\Phi_1 & \bm\Phi_2 & \cdots & \bm\Phi_d \\
            \mathbf I_m & \mathbf 0 & \cdots & \mathbf 0 \\
            & \ddots & \ddots & \vdots \\
            \mathbf 0 & & \mathbf I_m & \mathbf 0
        \end{pmatrix}
        \in\mathbb R^{md\times md}, \qquad
        \begin{aligned}
            &\mathbf J_W
            =\begin{pmatrix}\mathbf I_s & \mathbf 0_{s\times p} & \mathbf 0_{s\times m(d-1)}\end{pmatrix} \in\mathbb R^{s\times md} \\
            &\mathbf J_Y =\begin{pmatrix}\mathbf 0_{p\times s} & \mathbf I_p & \mathbf 0_{p\times m(d-1)}\end{pmatrix} \in\mathbb R^{p\times md}
        \end{aligned},
\end{equation*}
where \(m=s+p\) is the total number of variables. 
Taking \(\mathbf U_t=\mathbf S_{t-1}\) gives the model-implied forecasts
\[
    \boldsymbol\mu_{W,t}(\mathbf U_t)
    =\mathbf J_W\mathbf F\mathbf S_{t-1},
    \qquad
    \boldsymbol\mu_{Y,h,t}(\mathbf U_t)
    =\mathbf J_Y\mathbf F^{h+1}\mathbf S_{t-1}.
\]
Thus the nuisance in \eqref{eq:supp-svma-nuisance} has the concrete form
\[
    \mathbf g_{h,t}(\mathbf U_t)
    =\left(\mathbf J_Y\mathbf F^{h+1}
    -\mathbf B_h\mathbf J_W\mathbf F\right)\mathbf S_{t-1}.
\]

\section{Theories for Orthogonal Estimation}\label{app:orthogonal-estimation-proof}

Fix a horizon $h\in\mathcal H$, and write $n = n_h$, $P_n = P_{n,h}$, $\overline{\mathbf A}=\overline{\mathbf A}_{h,T}=n^{-1}\sum_{t=1}^n\mathbf A_t$, and $\beta(k) = \beta_h(k)$. We suppress the horizon subscript when this causes no ambiguity. Let
\[
    \boldsymbol\Delta_{Y,t} = \widehat{\boldsymbol\mu}_{Y,h}(\mathbf U_t)-\boldsymbol\mu_{Y,h}(\mathbf U_t),
    \quad
    \boldsymbol\Delta_{W,t} = \widehat{\boldsymbol\mu}_{W}(\mathbf U_t)-\boldsymbol\mu_W(\mathbf U_t),
    \quad
    \boldsymbol\Delta_{Z,t} = \widehat{\boldsymbol\mu}_{Z}(\mathbf U_t)-\boldsymbol\mu_Z(\mathbf U_t).
\]
We denote the Frobenius norm by $\|\cdot\|_\Fr$ and the operator norm by $\|\cdot\|_{\op}$. Since $p$, $s$, and $r$ are fixed, constants may change from line to line without depending on $T$.
\subsection{Neyman-orthogonality} \label{app:subsec:neyman-orth}
The following proposition verifies Neyman orthogonality of the score in \eqref{eq:orthogonal-score}.
\begin{proposition}[Neyman orthogonality]\label{prop:supp-orthogonality}
Suppose \eqref{eq:model} and \Cref{ass:iv}(i) hold and the outcome, treatment, and instrument are square-integrable. At the true parameter and nuisance functions, the score in \eqref{eq:orthogonal-score} satisfies $\E[\psi_{t,h}(\mathbf B_h,\eta_{0,h})]=\mathbf 0_{p\times r}$. Its population mean also has zero Gateaux derivative with respect to the nuisance functions,
\[
    \left.
    \frac{\partial}{\partial u}
    \E[\psi_{t,h}(\mathbf B_h,\eta_{0,h}+u a)]
    \right|_{u = 0} = \mathbf 0_{p\times r}
\]
for every direction $a = (a_Y,a_W,a_Z)$ with conformable dimensions such that $a(\mathbf U_t)$ is square-integrable and measurable with respect to $\sigma(\mathbf U_t)$.
\end{proposition}
\begin{proof}
By \eqref{eq:supp-residualized-model}, \(\widetilde{\mathbf Y}_{t+h}-\mathbf B_h\widetilde{\mathbf W}_t = \boldsymbol\varepsilon_{t+h}\).
At the true value, by the exclusion property \eqref{eq:supp-exclusion}, 
\[\E[\psi_{t,h}(\mathbf B_h,\eta_{0,h})] =\E[\widetilde{\mathbf Y}_{t+h}\widetilde{\mathbf Z}_t^\top]-\mathbf B_h\E[\widetilde{\mathbf W}_t\widetilde{\mathbf Z}_t^\top] = \E[\boldsymbol\varepsilon_{t+h}\widetilde{\mathbf Z}_t^\top] = \mathbf 0.\]
Let \(\eta_u = (\boldsymbol\mu_{Y,h}+ua_Y,\boldsymbol\mu_W+ua_W,\boldsymbol\mu_Z+ua_Z)\).  Writing $\mathbf b_t=\mathbf B_h a_W(\mathbf U_t)-a_Y(\mathbf U_t)$ and $\mathbf c_t=a_Z(\mathbf U_t)$, the score evaluated at \(\eta_u\) can be expressed as
\[
    \psi_{t,h}(\mathbf B_h,\eta_u)
    =
    (\boldsymbol\varepsilon_{t+h}+u\mathbf b_t)
    (\widetilde{\mathbf Z}_t-u\mathbf c_t)^\top.
\]
The random vectors \(\mathbf b_t\) and \(\mathbf c_t\) are square-integrable because the direction \(a\) is square-integrable and \(\mathbf B_h\) is fixed. Hence, by Cauchy--Schwarz, all entries of $\mathbf b_t\widetilde{\mathbf Z}_t^\top$, $\boldsymbol\varepsilon_{t+h}\mathbf c_t^\top$, and $\mathbf b_t\mathbf c_t^\top$ are integrable. Expanding the difference quotient entrywise gives the exact identity
\[
    \frac{\E[\psi_{t,h}(\mathbf B_h,\eta_u)]
    -\E[\psi_{t,h}(\mathbf B_h,\eta_{0,h})]}{u}
    =
    \E[\mathbf b_t\widetilde{\mathbf Z}_t^\top]
    -
    \E[\boldsymbol\varepsilon_{t+h}\mathbf c_t^\top]
    -
    u\E[\mathbf b_t\mathbf c_t^\top],
    \qquad u\ne0.
\]
The first two terms are zero by the tower property because \(\mathbf b_t\) and \(\mathbf c_t\) are \(\mathbf U_t\)-measurable and \(\E[\widetilde{\mathbf Z}_t\mid \mathbf U_t]=\mathbf 0\) and \(\E[\boldsymbol\varepsilon_{t+h}\mid \mathbf U_t]=\mathbf 0\). Since \(\mathbf b_t\mathbf c_t^\top\) is integrable, we let \(u\to0\) and conclude that the Gateaux derivative of the expectation at \(u=0\) is zero.
\end{proof}

\subsection{Theorems on estimation error and their proofs} \label{app:subsec:est-error-proof}
We first verify that the variance proxies in the general bound are finite in \Cref{prop:supp-variance-proxy}. We then state and prove \Cref{thm:supp-general}, a non-asymptotic bound for arbitrary nuisance rates $(\bar r_T,r_T)$, where $\bar r_T$ bounds each individual nuisance error and $r_T^2$ bounds the product of the instrument error with the outcome or treatment error, as in \Cref{cond:nuisance-rate}. Specializing to $\bar r_T\log n_h\to0$ and $r_T=o(T^{-1/4})$ recovers \Cref{thm:orthogonal-estimation}. The supporting concentration and empirical-process lemmas are proved in \Cref{app:mixing-lemmas,app:estimation-lemmas}.

\begin{proposition}[Finiteness of the variance proxies]\label{prop:supp-variance-proxy}
Fix $h\in\mathcal H$ and suppose \Cref{ass:iv,ass:estimation-primitive} hold. Then each coordinate of $\boldsymbol\varepsilon_{t+h}\widetilde{\mathbf Z}_t^\top$ and of $\widetilde{\mathbf W}_t\widetilde{\mathbf Z}_t^\top-\mathbf A_t$ is centered, geometrically $\beta$-mixing, and sub-exponential uniformly in $t$, with
\begin{equation} \label{eq:supp-se-bound}
    \sup_t\big\|(\boldsymbol\varepsilon_{t+h}\widetilde{\mathbf Z}_t^\top)_{lk}\big\|_\senorm\le K_\varepsilon K_Z,
    \qquad
    \sup_t\big\|(\widetilde{\mathbf W}_t\widetilde{\mathbf Z}_t^\top-\mathbf A_t)_{jk}\big\|_\senorm\le 2K_WK_Z .
\end{equation}
Consequently, the variance proxies respectively for the oracle score and the first-stage term
\begin{align}
    V_{\varepsilon Z}^2
    &=
    \max_{\substack{1\le l\le p\\ 1\le k\le r}}
    \sup_{t\in\mathbb Z}
    \left[
    \var\big\{(\boldsymbol\varepsilon_{t+h}\widetilde{\mathbf Z}_t^\top)_{lk}\big\}
    +2\sum_{m\ge1}
    \left|\cov\big(
    (\boldsymbol\varepsilon_{t+h}\widetilde{\mathbf Z}_t^\top)_{lk},
    (\boldsymbol\varepsilon_{t+h+m}\widetilde{\mathbf Z}_{t+m}^\top)_{lk}
    \big)\right|
    \right],
    \label{eq:lrv-epsZ}\\
    V_{WZ}^2
    &=
    \max_{\substack{1\le j\le s\\ 1\le k\le r}}
    \sup_{t\in\mathbb Z}
    \left[
    \var\big\{(\widetilde{\mathbf W}_t\widetilde{\mathbf Z}_t^\top)_{jk}\big\}
    +2\sum_{m\ge1}
    \left|\cov\big(
    (\widetilde{\mathbf W}_t\widetilde{\mathbf Z}_t^\top)_{jk},
    (\widetilde{\mathbf W}_{t+m}\widetilde{\mathbf Z}_{t+m}^\top)_{jk}
    \big)\right|
    \right]
    \label{eq:lrv-WZ}
\end{align}
are finite, with $V_{\varepsilon Z}\le C_\ast K_\varepsilon K_Z$ and $V_{WZ}\le 2C_\ast K_W K_Z$ for a constant $C_\ast$ depending only on the mixing constants $(C_\beta,c_\beta)$. 
\end{proposition}
\begin{remark}
The quantities $V_{\varepsilon Z}$ and $V_{WZ}$ are proxies rather than long-run variances, since the autocovariances enter in absolute value and no stationarity is imposed. Therefore, each is an upper bound for the corresponding worst-case per-coordinate long-run variance.
\end{remark}
\begin{proof}
Fix a coordinate. The oracle-score coordinate $(\boldsymbol\varepsilon_{t+h}\widetilde{\mathbf Z}_t^\top)_{lk} = (\boldsymbol\varepsilon_{t+h})_l(\widetilde{\mathbf Z}_t)_k$ is centered by the exclusion argument \eqref{eq:supp-exclusion}. The first-stage coordinate $(\widetilde{\mathbf W}_t\widetilde{\mathbf Z}_t^\top-\mathbf A_t)_{jk}$ is centered at every $t$ by the definition $\mathbf A_t=\E[\widetilde{\mathbf W}_t\widetilde{\mathbf Z}_t^\top]$.

The tail bounds come from the conditional sub-Gaussian controls in \Cref{ass:estimation-primitive}(iii). For any unit vector $\mathbf a$, $\E[\exp\{(\mathbf a^\top\widetilde{\mathbf Z}_t)^2/K_Z^2\}\mid\mathbf U_t]\le 2$ almost surely, so taking unconditional expectations and using the tower property gives $\|\widetilde{\mathbf Z}_t\|_\sgnorm\le K_Z$.

The same argument above gives $\|\boldsymbol\varepsilon_{t+h}\|_\sgnorm\le K_\varepsilon$, while $\|\widetilde{\mathbf W}_t\|_\sgnorm\le K_W$ is already unconditional. The product property of Orlicz norms \cite[Lemma 2.8.6]{vershynin2026high} then yields $\|(\boldsymbol\varepsilon_{t+h})_l(\widetilde{\mathbf Z}_t)_k\|_\senorm\le K_\varepsilon K_Z$ and $\|(\widetilde{\mathbf W}_t)_j(\widetilde{\mathbf Z}_t)_k\|_\senorm\le K_W K_Z$ uniformly in $t$. Centering the latter variable at its time-specific expectation gives the bound $2K_W K_Z$.

We now turn to the absolute summability. The point that needs care is that the variance proxies in \eqref{eq:lrv-epsZ} and \eqref{eq:lrv-WZ} are formed from scalar coordinate processes, whereas the primitive mixing condition is imposed on the vector process
$\mathcal X_{t,h} := (\boldsymbol\varepsilon_{t+h},\widetilde{\mathbf W}_t, \widetilde{\mathbf Z}_t,\mathbf U_t)$.

For fixed coordinates $l,k,j$, define
\[
    \zeta^{\varepsilon Z}_{t,lk}
    =
    (\boldsymbol\varepsilon_{t+h})_l(\widetilde{\mathbf Z}_t)_k,
    \qquad
    \zeta^{WZ}_{t,jk}
    =
    (\widetilde{\mathbf W}_t)_j(\widetilde{\mathbf Z}_t)_k-(\mathbf A_t)_{jk}.
\]
Each of these variables is a measurable function of $\mathcal X_{t,h}$. Hence their $\beta$-mixing coefficient is no larger than that of $(\mathcal X_{t,h})$. Consequently, both scalar coordinate sequences are geometrically $\beta$-mixing with the same envelope $C_\beta e^{-c_\beta k}$.
Thus, with \eqref{eq:supp-se-bound}, \Cref{lem:supp-truncated-lrv} applies to each fixed oracle-score coordinate with $K_\zeta=K_\varepsilon K_Z$ and to each fixed first-stage coordinate with $K_\zeta=2K_WK_Z$. For the oracle coordinate, the bracketed expression in \eqref{eq:lrv-epsZ} is exactly the $V_0^2$ quantity of that lemma for $(\zeta^{\varepsilon Z}_{t,lk})$. For the first-stage coordinate, subtracting the time-specific deterministic expectation $(\mathbf A_t)_{jk}$ does not change covariances, so the bracketed expression in \eqref{eq:lrv-WZ} is the same $V_0^2$ quantity for $(\zeta^{WZ}_{t,jk})$.

By \eqref{eq:supp-variance-bound}, these coordinatewise proxies satisfy
$V_{0,\varepsilon Z,lk}^2 \le C_\ast^2K_\varepsilon^2K_Z^2$ and 
$V_{0,WZ,jk}^2 \le 4C_\ast^2K_W^2K_Z^2$,
where $C_\ast$ depends only on $(C_\beta,c_\beta)$. The maxima in \eqref{eq:lrv-epsZ} and \eqref{eq:lrv-WZ} range over finitely many fixed coordinates, so they preserve these bounds and yield $V_{\varepsilon Z}\le C_\ast K_\varepsilon K_Z$ and $V_{WZ}\le 2C_\ast K_WK_Z$.
\end{proof}

\begin{theorem}[General non-asymptotic oracle bound]\label{thm:supp-general}
Fix $h\in\mathcal H$ and any $c>0$. Suppose \Cref{ass:iv,ass:estimation-primitive,ass:average-relevance,ass:cross-fitting}, \Cref{cond:sg-nuisance-error}, and the event $\mathcal E_{T,h}$ of \Cref{cond:nuisance-rate} hold except that $\bar r_T$ and $r_T$ are arbitrary sequences with $r_T\le\bar r_T$. Let $\widehat{\mathbf B}_h$ be computed by \Cref{alg:orthogonal-lpiv} using blocked cross-fitting construction. Write $\rho_{n,h} = \{\log n_h/n_h\}^{1/2}$, let $c_3 = c_3(C_\beta,c_\beta)>0$ be the constant in the geometrically mixing Bernstein inequality of \citet{merlevede2009bernstein}, and set
\[
    \kappa_c = \Big\{\tfrac{32(c+1)}{c_3}\Big\}^{1/2},
    \qquad
    K_\circ = K_\varepsilon+(1+C_B)K_Z,
    \qquad
    \bar\ell_T =
    \begin{cases}
        1\vee\log(\bar K/\bar r_T),&\bar r_T>0,\\
        1,&\bar r_T=0.
    \end{cases}
\]
Let $C_c=C_c(c,C_\beta,c_\beta)<\infty$ be a sufficiently large constant
as specified in the proof of \Cref{lem:supp-crossfit-firstorder}.
Let $L$ be the number of cross-fitting folds, and set $N_T = \max\{N_0,\ N_1,\ 3(pr+sr+s+r+3Lpr)\}$,
where $N_0$ and $N_1$ are the dimension-free concentration thresholds of \Cref{lem:supp-oracle-empirical,lem:supp-crossfit-firstorder}. Suppose $T\ge T_0$ and $n_h\ge N_T$, and suppose the first-stage deviation
\begin{equation}\label{eq:supp-first-stage-smallness}
    D_A:=\kappa_c\sqrt{sr}\,V_{WZ}\,\rho_{n,h}
    +2\,(\sqrt{s}\,K_W+\sqrt{r}\,K_Z)\,\bar r_T+r_T^2
\end{equation}
satisfies $D_A\le\underline\sigma_A/2$. Then, with probability at least
\[
    \Pr(\mathcal E_{T,h})-n_h^{-c}-3L\,C_\beta\exp(-c_\beta b_T),
\]
the matrix $\widehat{\mathbf A}_h$ has full row rank and
\begin{equation}\label{eq:supp-general-bound}
    \|\widehat{\mathbf B}_h-\mathbf B_h\|_\Fr
    \le \frac{2}{\underline\sigma_A}\big\{
    \kappa_c\sqrt{pr}\,V_{\varepsilon Z}\,\rho_{n,h}+(1+C_B)\,r_T^2
    +LC_c\sqrt{pr}\,K_\circ
    \{\bar r_T\sqrt{\bar\ell_T}\,\rho_{n,h}+\bar K\{\log n_h\}^4/n_h\}
    \big\}.
\end{equation}
Moreover, \(\widehat{\mathbf B}_h-\mathbf B_h = P_{n,h}(\boldsymbol\varepsilon_{t+h}\widetilde{\mathbf Z}_t^\top)\overline{\mathbf A}_{h,T}^\dagger+\mathbf R_{T,h}\), where on the same event
\begin{equation}\label{eq:supp-general-remainder}
    \|\mathbf R_{T,h}\|_\Fr
    \le \frac{2}{\underline\sigma_A}\big\{(1+C_B)\,r_T^2
    +LC_c\sqrt{pr}\,K_\circ
    \{\bar r_T\sqrt{\bar\ell_T}\,\rho_{n,h}+\bar K\{\log n_h\}^4/n_h\}
    \big\}
    +\frac{3\kappa_c\sqrt{pr}\,V_{\varepsilon Z}}{\underline\sigma_A^2}\,\rho_{n,h}\,D_A .
\end{equation}
\end{theorem}

\begin{proof}[Proof of \Cref{thm:supp-general}]
Let $\mathcal G_n$ be the event on which $\mathcal E_{T,h}$ and all four bounds \eqref{eq:supp-ep-oracle}--\eqref{eq:supp-ep-firstorder} hold. By \Cref{lem:supp-oracle-empirical,lem:supp-crossfit-firstorder} and a union bound,
\[
    \Pr(\mathcal G_n)
    \ge \Pr(\mathcal E_{T,h})
    -3\{pr+sr+s+r+3Lpr\}n^{-(c+1)}
    -3L\,C_\beta e^{-c_\beta b_T}.
\]
Since $n\ge 3(pr+sr+s+r+3Lpr)$, the middle term is at most $n^{-c}$. On $\mathcal G_n$, \Cref{lem:supp-first-stage} and the smallness condition $D_A\le\underline\sigma_A/2$ imply that $\widehat{\mathbf A}_h$ has full row rank and \(\|\widehat{\mathbf A}_h^\dagger\|_{\op}\le 2/\underline\sigma_A\). Suppress the subsript $h$ when no confusion arises. 
Since \(\widehat{\mathbf A}_h\widehat{\mathbf A}_h^\dagger = \mathbf I_s\),
\begin{align}
    \widehat{\mathbf B}_h-\mathbf B_h
    = \widehat{\mathbf M}_h\widehat{\mathbf A}_h^\dagger
      -\mathbf B_h\widehat{\mathbf A}_h\widehat{\mathbf A}_h^\dagger
    = P_n\{(\widehat{\mathbf e}_{Y,t,h}-\mathbf B_h\widehat{\mathbf e}_{W,t})
      \widehat{\mathbf e}_{Z,t}^\top\}\widehat{\mathbf A}_h^\dagger .
    \label{eq:supp-estimator-identity}
\end{align}
Using \(\widetilde{\mathbf Y}_{t+h}-\mathbf B_h\widetilde{\mathbf W}_t = \boldsymbol\varepsilon_{t+h}\), the numerator in \eqref{eq:supp-estimator-identity} expands as
\[\begin{aligned}
    P_n\{(\widehat{\mathbf e}_{Y,t,h}-\mathbf B_h\widehat{\mathbf e}_{W,t})
      \widehat{\mathbf e}_{Z,t}^\top\} 
    = &\underbrace{P_n(\boldsymbol\varepsilon_{t+h}\widetilde{\mathbf Z}_t^\top)}_{\mathrm{\Romannum{1}}}
      +\underbrace{P_n\left[-\boldsymbol\Delta_{Y,t}\widetilde{\mathbf Z}_t^\top
      +\mathbf B_h\boldsymbol\Delta_{W,t}\widetilde{\mathbf Z}_t^\top
      -\boldsymbol\varepsilon_{t+h}\boldsymbol\Delta_{Z,t}^\top\right]}_{\mathrm{\Romannum{2}}} \\
      &+\underbrace{P_n\{(\boldsymbol\Delta_{Y,t}-\mathbf B_h\boldsymbol\Delta_{W,t})
      \boldsymbol\Delta_{Z,t}^\top\}}_{\mathrm{\Romannum{3}}}.
\end{aligned}\]
By \Cref{lem:supp-oracle-empirical,lem:supp-crossfit-firstorder}, the first two terms \Romannum{1} and \Romannum{2} are bounded by
\[
    \kappa_c\sqrt{pr}\,V_{\varepsilon Z}\rho_n
    \quad\text{and}\quad
    LC_c\sqrt{pr}\,K_\circ
    \{\bar r_T\sqrt{\bar\ell_T}\,\rho_n+\bar K\{\log n_h\}^4/n_h\},
\]
respectively. Here term \Romannum{2} is a first-order empirical process whose three summands pair a single nuisance error with a mean-zero residual factor, so its bound scales with the individual error size $\bar r_T$ of \Cref{cond:nuisance-rate}(i). The final term is quadratic and pairs two nuisance errors, and Cauchy--Schwarz gives
\[
    \left\|P_n\{(\boldsymbol\Delta_{Y,t}-\mathbf B_h\boldsymbol\Delta_{W,t})
      \boldsymbol\Delta_{Z,t}^\top\}\right\|_\Fr
    \le \|\boldsymbol\Delta_{Y}\|_{n,h,2}\|\boldsymbol\Delta_{Z}\|_{n,h,2}
    +C_B\|\boldsymbol\Delta_{W}\|_{n,h,2}\|\boldsymbol\Delta_{Z}\|_{n,h,2}
    \le (1+C_B)\, r_T^2 ,
\]
where the last step uses the product control \(\|\boldsymbol\Delta_{Z}\|_{n,h,2}\max\{\|\boldsymbol\Delta_{Y}\|_{n,h,2},\|\boldsymbol\Delta_{W}\|_{n,h,2}\}\le r_T^2\) of \Cref{cond:nuisance-rate}(ii) on $\mathcal E_{T,h}$. This is the only place the product error rate enters the error bound. Multiplying the three numerator bounds by \(\|\widehat{\mathbf A}_h^\dagger\|_{\op}\le 2/\underline\sigma_A\) gives \eqref{eq:supp-general-bound}.

For the expansion, subtract the oracle term $P_n(\boldsymbol\varepsilon_{t+h}\widetilde{\mathbf Z}_t^\top)\overline{\mathbf A}^\dagger$ from \eqref{eq:supp-estimator-identity}. Write the numerator as
$P_n(\boldsymbol\varepsilon_{t+h}\widetilde{\mathbf Z}_t^\top)+\mathbf S_n$,
where $\mathbf S_n=\text{\Romannum{2}} + \text{\Romannum{3}}$ collects the first-order and quadratic nuisance terms. The remainder is
\[
    \mathbf R_{T,h} = P_n(\boldsymbol\varepsilon_{t+h}\widetilde{\mathbf Z}_t^\top)(\widehat{\mathbf A}_h^\dagger-\overline{\mathbf A}^\dagger)
    +\mathbf S_n\widehat{\mathbf A}_h^\dagger .
\]
By \eqref{eq:supp-ep-oracle} and \Cref{lem:supp-first-stage}, the first term is bounded by
\[
    \kappa_c\sqrt{pr}\,V_{\varepsilon Z}\rho_n\cdot\frac{3}{\underline\sigma_A^2}\,D_A = \frac{3\kappa_c\sqrt{pr}\,V_{\varepsilon Z}}{\underline\sigma_A^2}\,\rho_n D_A ,
\]
which is the last term in \eqref{eq:supp-general-remainder}. For the remaining term, the nuisance part satisfies
\[
    \|\mathbf S_n\|_\Fr
    \le LC_c\sqrt{pr}\,K_\circ
    \{\bar r_T\sqrt{\bar\ell_T}\,\rho_n+\bar K\{\log n_h\}^4/n_h\}
    +(1+C_B)r_T^2,
\]
and \(\|\widehat{\mathbf A}_h^\dagger\|_{\op}\le 2/\underline\sigma_A\). Hence \(\mathbf S_n\widehat{\mathbf A}_h^\dagger\) contributes
\[
    \frac{2}{\underline\sigma_A}
    \{(1+C_B)r_T^2+
    LC_c\sqrt{pr}\,K_\circ
    \{\bar r_T\sqrt{\bar\ell_T}\,\rho_n+\bar K\{\log n_h\}^4/n_h\}\}.
\]
Combining the two displays gives the expansion and the bound \eqref{eq:supp-general-remainder} on $\mathbf R_{T,h}$.
\end{proof}

\begin{proof}[Proof of \Cref{thm:orthogonal-estimation}]
We deduce both parts from \Cref{thm:supp-general}. Recall $n_h = T-h\asymp T$. By \Cref{cond:nuisance-rate}(ii), $r_T = o(T^{-1/4}) = o(n_h^{-1/4})\to0$, and by \Cref{cond:nuisance-rate}(i), $\bar r_T\log n_h\to0$, so in particular $\bar r_T\to0$. For any fixed $c>0$, $\kappa_c$ does not depend on $T$, so the first-stage deviation $D_A = O(\rho_{n,h}+\bar r_T)\to0$ and the smallness condition $D_A\le\underline\sigma_A/2$ holds for all sufficiently large $T$. Since $p$, $s$, $r$, and the number of folds $L$ are fixed, the threshold $N_T$ of \Cref{thm:supp-general} is a constant, so $n_h\ge N_T$ also holds for all sufficiently large $T$. When $\bar r_T>0$, $\bar\ell_T = 1\vee\log(\bar K/\bar r_T)$ satisfies $\bar\ell_T\asymp 1+\log(1/\bar r_T)$ for large $T$, since $\bar K$ is fixed and $\bar r_T\to0$. If $\bar r_T=0$, all population nuisance errors vanish and the corresponding first-order contribution is zero, so the same conclusions follow directly.

\emph{Part (i).} Fix $c>0$ and apply \Cref{thm:supp-general}, whose conclusion holds with probability at least \(\Pr(\mathcal E_{T,h})-n_h^{-c}-3L\,C_\beta e^{-c_\beta b_T}\). Since \(\kappa_c = \{32(c+1)/c_3\}^{1/2} = \{32/c_3\}^{1/2}\sqrt{1+c}\), the oracle term carries the factor $\sqrt{1+c}$. For fixed $c$, the constant $C_c$ does not depend on $T$. Moreover, $\bar r_T\sqrt{\bar\ell_T}\to0$ and $\bar K\log^4 n_h/n_h=o(\rho_{n,h})$, so the entire cross-fitted first-order term in \eqref{eq:supp-general-bound} is $o(\rho_{n,h})$ and is absorbed into the oracle term for all sufficiently large $T$. Collecting the remaining fixed factors into a constant $C$ that does not depend on $c$ or $T$ yields \eqref{eq:consistency-rate}.

\emph{Part (ii).} Using $D_A = O(\rho_{n,h}+\bar r_T)$, the remainder bound \eqref{eq:supp-general-remainder} reads
\[
    \|\mathbf R_{T,h}\|_\Fr
    \lesssim_P
    r_T^2+\bar r_T\sqrt{\bar\ell_T}\,\rho_{n,h}+\bar K\{\log n_h\}^4/n_h+\rho_{n,h}^2+\rho_{n,h}\bar r_T .
\]
The first, third, and fourth terms are $o(n_h^{-1/2})$: here $r_T^2 = o(n_h^{-1/2})$ by the product rate, \(\bar K\{\log n_h\}^4/n_h = O(n_h^{-1}\log^4 n_h) = o(n_h^{-1/2})\), and \(\rho_{n,h}^2 = O(n_h^{-1}\log n_h) = o(n_h^{-1/2})\). The last term is also $o(n_h^{-1/2})$: since $\bar r_T\log n_h\to0$ gives $\bar r_T = o(1/\log n_h)$,
\[
    \rho_{n,h}\bar r_T = \sqrt{\frac{\log n_h}{n_h}}\,\bar r_T
    = o\!\left(\frac{1}{\sqrt{n_h\log n_h}}\right) = o(n_h^{-1/2}).
\]
For the logarithmic product term, write $\bar r_T\sqrt{\bar\ell_T}\,\rho_{n,h} = n_h^{-1/2}\,\bar r_T\sqrt{\bar\ell_T\log n_h}$, so it suffices that $\bar r_T\sqrt{\bar\ell_T\log n_h}\to0$. By the arithmetic--geometric mean inequality,
\[
    \bar r_T\sqrt{\bar\ell_T\log n_h}
    \le \tfrac12\,\bar r_T\bar\ell_T+\tfrac12\,\bar r_T\log n_h .
\]
The second summand vanishes by \Cref{cond:nuisance-rate}(i). For the first, $\bar r_T\bar\ell_T = \bar r_T\{1\vee\log(\bar K/\bar r_T)\}\to0$ because $x\log(\bar K/x)\to0$ as $x\downarrow0$ and $\bar r_T\to0$. Hence $\bar r_T\sqrt{\bar\ell_T}\,\rho_{n,h} = o(n_h^{-1/2})$, and therefore \(\|\mathbf R_{T,h}\|_\Fr = o_p(n_h^{-1/2})\), which proves \eqref{eq:orthogonal-linear-expansion}. For each oracle-score coordinate, \eqref{eq:lrv-epsZ} gives \(\var(\sum_{t = 1}^{n_h}\zeta_t)\le n_hV_{\varepsilon Z}^2\). Chebyshev's inequality and the fixed score dimension therefore give \(\|P_{n,h}(\boldsymbol\varepsilon_{t+h}\widetilde{\mathbf Z}_t^\top)\|_\Fr = O_p(n_h^{-1/2})\). Since \(\|\overline{\mathbf A}_{h,T}^\dagger\|_\op\le 1/\underline\sigma_A\), the leading term is $O_p(n_h^{-1/2})$, and hence \(\|\widehat{\mathbf B}_h-\mathbf B_h\|_\Fr = O_p(n_h^{-1/2})\).
\end{proof}

\subsection{General lemmas for mixing sequences}\label{app:mixing-lemmas}
\begin{lemma}[Uniform truncated variance proxies]\label{lem:supp-truncated-lrv}
Let $(\zeta_t)_{t\in\mathbb Z}$ be a centered, mixing, and sub-exponential process in the sense that
$\beta(k)\le C_\beta e^{-c_\beta k}$, and $\sup_t\|\zeta_t\|_\senorm\le K_\zeta$.
For $M>0$, define the centered truncations
\begin{equation}\label{eq:supp-centered-truncation}
    \zeta_t^{(M)} = \zeta_t\I\{|\zeta_t|\le M\}
    -\E[\zeta_t\I\{|\zeta_t|\le M\}],
\end{equation}
and define
\[
    \begin{aligned}
        v_M^2 & = \sup_{i\in\mathbb Z}\left\{
        \var(\zeta_i^{(M)})
        +2\sum_{j>i}|\cov(\zeta_i^{(M)},\zeta_j^{(M)})|
        \right\},\\
        V_0^2 & = \sup_{i\in\mathbb Z}\left\{
        \var(\zeta_i)+2\sum_{j>i}|\cov(\zeta_i,\zeta_j)|
        \right\}.
    \end{aligned}
\]
Then $V_0 < \infty$ and $v_M^2\to V_0^2$ as $M\to\infty$. 
\end{lemma}
\begin{proof}
The sub-exponential bound is uniform in $t$. Hence, for every fixed $q\ge 1$,
\[
    \sup_t \E|\zeta_t|^q \le C_q K_\zeta^q.
\]
The same bound applies to $\zeta_t^{(M)}$, uniformly over $t$ and $M$. In particular, by \citet[Remark 2.8.8]{vershynin2026high},
\begin{equation} \label{eq:supp-l4-bound}
    \sup_{t,M}\|\zeta_t^{(M)}\|_4+\sup_t\|\zeta_t\|_4 \le C K_\zeta.
\end{equation}
By Davydov's covariance inequality \citep[(1.12b)]{rio2017asymptotic} and $\alpha(k)\le\beta(k)$, we have, uniformly in $i$, $k$, and $M$,
\begin{equation}\label{eq:supp-davydov-envelope}
    |\cov(\zeta_i^{(M)},\zeta_{i+k}^{(M)})|
    +|\cov(\zeta_i,\zeta_{i+k})|
    \le CK_\zeta^2\alpha(k)^{1/2}
    \le CK_\zeta^2 e^{-c_\beta k/2}.
\end{equation}
The right-hand side is summable in $k$, and specifically, after enlarging the constant $C$, 
\begin{align}
    V_0^2 \le C_2 K_\zeta^2 + 2 CK_\zeta^2 / (1-e^{-c_\beta/2}) \le C K_\zeta^2 < \infty,
    \label{eq:supp-variance-bound}
\end{align}
where $C$ only depends on $(C_\beta,c_\beta)$. The covariance sequence is controlled. 

We now turn to the truncated process. It is also bounded by \eqref{eq:supp-variance-bound}, i.e., $\sup_M v_M<\infty$.
We first fix a summand of index $k\ge 0$ and set $j=i+k$. Define
\[
    d_t^{(M)}:=\zeta_t^{(M)}-\zeta_t
    =-\zeta_t\I\{|\zeta_t|>M\}
      +\E[\zeta_t\I\{|\zeta_t|>M\}],
    \qquad
    \Delta_M:=\sup_t\|d_t^{(M)}\|_4.
\]
Since both $\zeta_t$ and $\zeta_t^{(M)}$ are centered,
\begin{align}
    &\left|\cov(\zeta_i^{(M)},\zeta_j^{(M)})
      -\cov(\zeta_i,\zeta_j)\right| \notag \\
    &\qquad\le
      \left|\E[d_i^{(M)}\zeta_j]\right|
      +\left|\E[\zeta_i d_j^{(M)}]\right|
      +\left|\E[d_i^{(M)}d_j^{(M)}]\right| \notag \\
    &\qquad\le (\|\zeta_i\|_4+\|\zeta_j\|_4)\Delta_M+\Delta_M^2,
    \label{eq:supp-cov-diff}
\end{align}
where the last inequality follows from Hölder inequality. Moreover, the uniform sub-exponential tail bound gives
\begin{align}
    \Delta_M
    &\le 2\sup_t
       \left\{\E\big[|\zeta_t|^4\I\{|\zeta_t|>M\}\big]\right\}^{1/4} \notag \\
    &\le 2\left\{
       2M^4e^{-M/K_\zeta}
       +8\int_M^\infty x^3e^{-x/K_\zeta}\,dx
       \right\}^{1/4}
    \longrightarrow0
    \qquad\text{as }M\to\infty.
    \label{eq:supp-trunc-l4}
\end{align}
Substituting \eqref{eq:supp-l4-bound} and \eqref{eq:supp-trunc-l4} into \eqref{eq:supp-cov-diff} and taking the supremum over $i$ gives
\begin{equation}\label{eq:supp-fixed-lag-cov-conv}
    \sup_i
    \left|
    \cov(\zeta_i^{(M)},\zeta_{i+k}^{(M)})
    -\cov(\zeta_i,\zeta_{i+k})
    \right|
    \le C K_\zeta\Delta_M+\Delta_M^2
    \to0 .
\end{equation}
Then, we sum over $k = 0,\ldots,L$ for any fixed $L$ to obtain
\begin{equation}\label{eq:supp-finite-lag-conv}
    \sup_i\sum_{k = 0}^{L}\left|
    \cov(\zeta_i^{(M)},\zeta_{i+k}^{(M)})
    -\cov(\zeta_i,\zeta_{i+k}) \right|
    \to 0 \quad \text{as } M\to\infty .
\end{equation}

The remaining lags are controlled uniformly by the summability of \eqref{eq:supp-davydov-envelope} over $k$. Define the tail lag covariance sum
\[
    R_L = \sup_{i,M}\sum_{k>L}
    \left\{
    |\cov(\zeta_i^{(M)},\zeta_{i+k}^{(M)})|
    +|\cov(\zeta_i,\zeta_{i+k})|
    \right\}.
\]
Then by \eqref{eq:supp-davydov-envelope},
\begin{equation}\label{eq:supp-tail-lag-bound}
    R_L \le CK_\zeta^2\sum_{k\ge L}e^{-c_\beta k/2} = CK_\zeta^2 e^{-c_\beta L/2}/(1-e^{-c_\beta/2})
    \to 0 \quad\text{as } L\to\infty .
\end{equation}

In order to compare the variance proxies $v_M^2$ and $V_0^2$, we define
\[
    S_M(i) = \var(\zeta_i^{(M)})
    +2\sum_{j>i}|\cov(\zeta_i^{(M)},\zeta_j^{(M)})|,
    \qquad
    S(i) = \var(\zeta_i)
    +2\sum_{j>i}|\cov(\zeta_i,\zeta_j)|.
\]
Then, using the reverse triangle inequality, for every fixed $L$,
\begin{align*}
    |v_M^2-V_0^2| & = \left|\sup_i S_M(i)-\sup_i S(i)\right| \\
    &\le \sup_i|S_M(i)-S(i)|\\
    &\le 2\sup_i \sum_{k = 0}^{L}\left|
    \cov(\zeta_i^{(M)},\zeta_{i+k}^{(M)})
    -\cov(\zeta_i,\zeta_{i+k})\right| + 2R_L .
\end{align*}
For any $\epsilon>0$, we fix an $L = L(\epsilon)$ such that $R_L<\epsilon/4$ by \eqref{eq:supp-tail-lag-bound}, and then by \eqref{eq:supp-finite-lag-conv}, there exists $M_\star = M(\epsilon,L(\epsilon))$ such that the first term is less than $\epsilon/2$ whenever $M\ge M_\star$. Then, for all $M\ge M_\star$, $|v_M^2-V_0^2|<\epsilon$. This proves $v_M^2\to V_0^2$ as $M\to\infty$. 

If $V_0 = 0$, then $\var(\zeta_i) = 0$ for every $i$. Since $\E\zeta_i = 0$, $\zeta_i = 0$ almost surely for every $i$.
\end{proof}

\begin{lemma}[Bernstein bound for sub-exponential mixing sequences]\label{lem:supp-mixing-subexp-bernstein}
Let $(\zeta_t)_{t\in\mathbb Z}$ be a centered real-valued process satisfying
\[
    \beta(k)\le C_\beta e^{-c_\beta k},
    \qquad
    \sup_t\|\zeta_t\|_\senorm\le K_\zeta.
\]
Let
\[
    V_0^2
    :=\sup_{i\in\mathbb Z}\left\{
    \var(\zeta_i)+2\sum_{j>i}|\cov(\zeta_i,\zeta_j)|
    \right\},
\]
and let $V_\zeta>0$ be any finite constant such that $V_0\le V_\zeta$. Fix $c>0$, set $\rho_n = \sqrt{\log n/n}$, and let $c_3 = c_3(C_\beta,c_\beta) > 0$ be the constant in \eqref{eq:supp-mpr}. There exists a threshold $N_\star$, depending on $(c,C_\beta,c_\beta,K_\zeta,V_\zeta)$ but not on $n$ and not on any dimension, such that for all $n\ge N_\star$,
\[
    \Pr\left(
    |P_n\zeta_t|>\Big\{\frac{32(c+1)}{c_3}\Big\}^{1/2}V_\zeta\rho_n
    \right)
    \le 3n^{-(c+1)}.
\]
Consequently, for a collection of $D$ coordinate processes of this form, a union bound gives joint failure probability at most $3D n^{-(c+1)}$, which is at most $n^{-c}$ as soon as $n\ge 3D$.
\end{lemma}

\begin{proof}
We apply a Bernstein inequality for bounded mixing sequences after truncating the observations. Set the deviation $\delta_n = \{{32(c+1)}/{c_3}\}^{1/2}V_\zeta\rho_n$. Note that
\[
    \Pr(|P_n\zeta_t|>\delta_n) \le \Pr\left(\max_{1\le t\le n}|\zeta_t|>M_n\right) + \Pr\left(|P_n\zeta_t|>\delta_n, \max_{1\le t\le n}|\zeta_t|\le M_n\right).
\]
Set the truncation level as
\[
    M_n = (c+2)K_\zeta\log n.
\]
The uniform sub-exponential bound gives \(\sup_t\Pr(|\zeta_t|>x)\le 2e^{-x/K_\zeta}\). Hence
\begin{equation}\label{eq:supp-subexp-max-tail}
    \Pr\left(\mathcal E_{M_n}\right) \le 2n^{-(c+1)},
    \quad \text{where } \mathcal E_{M_n}:=\{\max_{1\le t\le n}|\zeta_t|>M_n\}.
\end{equation}
Write \(\mu_{n,t} = \E[\zeta_t\I\{|\zeta_t|\le M_n\}]\). Since $\E\zeta_t = 0$, we also have \(\mu_{n,t} = -\E[\zeta_t\I\{|\zeta_t|> M_n\}].\)
Uniform tail integration gives
\begin{align*}
    \max_{1\le t\le n}|\mu_{n,t}|
    &\le \sup_t\E[|\zeta_t|\I\{|\zeta_t|> M_n\}] \\
    & = \sup_t \left[ M_n \Pr(|\zeta_t|> M_n) + \int_{M_n}^{\infty} \Pr(|\zeta_t|>x) dx \right] \\
    &\le 2(M_n+K_\zeta)n^{-(c+2)}
     \le \delta_n/2
\end{align*}
for all sufficiently large $n$. 

On $\mathcal E_{M_n}^c$, the centered truncated sequence $(\zeta_t^{(M_n)})_{t\in\mathbb Z}$ satisfies $P_n\zeta_t = P_n\zeta_t^{(M_n)}+P_n\mu_{n,t}$,
so $|P_n\zeta_t|>\delta_n$ implies $|P_n\zeta_t^{(M_n)}|>\delta_n/2$ when $(\zeta_t)$ is bounded by $M_n$. Therefore,
\[ \{|P_n\zeta_t|>\delta_n\} \cap \mathcal E_{M_n}^c \subseteq \{|P_n\zeta_t^{(M_n)}|>\delta_n/2\}, \]
and we turn to bound the latter probability.

Note that $(\zeta_t^{(M_n)})_{t\in\mathbb Z}$ is centered, geometrically mixing, and bounded by $2M_n$. By \Cref{lem:supp-truncated-lrv}, $v_{M_n}^2\to V_0^2$. 
As a result of \citet[Theorem 2]{merlevede2009bernstein}, there exists $c_3>0$ depending only on $(C_\beta,c_\beta)$ such that
\begin{align}\label{eq:supp-mpr}
    \Pr\Big(\Big| \sum_{t = 1}^n\zeta_t^{(M_n)} \Big| \ge x \Big)
    \le
    \exp\left\{-\frac{c_3 x^2}{nv_{M_n}^2+M_n^2+xM_n(\log n)^2}\right\}.
\end{align}
Let $x = n\delta_n/2$. Since $V_0\le V_\zeta$, the denominator in \eqref{eq:supp-mpr} satisfies
\[
    n v_{M_n}^2+C M_n^2+C n\delta_n M_n(\log n)^2
    \le 4nV_\zeta^2
\]
whenever $n$ is sufficiently large. Therefore, for the centered truncated sequence, the deviation probability is bounded by
\[
    \Pr\left(|P_n\zeta_t^{(M_n)}|>\delta_n/2\right)
    \le \exp\left\{-\frac{c_3n^2\delta_n^2/4}{4nV_\zeta^2}\right\}
    \le n^{-(c+1)}.
\]
Combining this bound with tail probability \eqref{eq:supp-subexp-max-tail} proves the claim.
\end{proof}

\begin{lemma}[Triangle inequality for mixing coefficients]\label{lem:supp-mixing-triangle}
Let $\mathcal A,\mathcal B,\mathcal C$, and $\mathcal R$ be sub-$\sigma$-fields of a probability space $(\Omega,\mathcal F,\Pr)$. For $\sigma$-fields, define
$\mathcal A\vee\mathcal B=\sigma(\mathcal A\cup\mathcal B)$ 
as the smallest $\sigma$-field containing both $\mathcal A$ and $\mathcal B$. For numerical arguments, $a\vee b$ continues to denote $\max\{a,b\}$. We use the absolute-regularity coefficient
\begin{equation}\label{eq:supp-beta-partitions}
    \beta(\mathcal A,\mathcal B)
    :=\frac12\sup_{(A_i),(B_j)}
      \sum_{i,j}\left|\Pr(A_i\cap B_j)-\Pr(A_i)\Pr(B_j)\right|,
\end{equation}
where the supremum is over finite measurable partitions $(A_i)$ and $(B_j)$ of $\Omega$ with cells in $\mathcal A$ and $\mathcal B$, respectively; see \citet[(1.58)]{rio2017asymptotic}. Then the following properties hold.
\begin{enumerate}
    \item[(i)] If $\mathcal A_0\subseteq\mathcal A$ and $\mathcal B_0\subseteq\mathcal B$, then
    \(
        \beta(\mathcal A_0,\mathcal B_0)
        \le\beta(\mathcal A,\mathcal B).
    \)
    \item[(ii)] The triangle inequality gives
    \begin{equation}\label{eq:supp-beta-triangle}
        \beta(\mathcal B,\mathcal A\vee\mathcal C)
        \le\beta(\mathcal A,\mathcal B)
        +\beta(\mathcal A\vee\mathcal B,\mathcal C)+\beta(\mathcal A,\mathcal C).
    \end{equation}
    \item[(iii)] If $\mathcal R$ is independent of $\mathcal A\vee\mathcal B$, then
    \begin{equation}\label{eq:supp-beta-independent-extension}
        \beta(\mathcal A,\mathcal B\vee\mathcal R)
        =\beta(\mathcal A,\mathcal B).
    \end{equation}
\end{enumerate}
\end{lemma}

\begin{proof}
Part (i) follows directly from \eqref{eq:supp-beta-partitions}, because passing to sub-$\sigma$-fields restricts the partitions over which the supremum is taken.

For part (ii), first take finite partitions $(A_i)$, $(B_j)$, and $(C_k)$ with cells in $\mathcal A$, $\mathcal B$, and $\mathcal C$, respectively. Write
$p_i = \Pr(A_i)$, $p_j = \Pr(B_j)$, and $p_k = \Pr(C_k)$ and
\[
    p_{ijk}=\Pr(A_i\cap B_j\cap C_k),
    \qquad p_{ij}=\Pr(A_i\cap B_j),
    \qquad p_{ik}=\Pr(A_i\cap C_k).
\]
The identity
\(
    p_{ijk}-p_jp_{ik}
    =(p_{ijk}-p_{ij}p_k)
      +(p_{ij}-p_ip_j)p_k
      +p_j(p_ip_k-p_{ik})
\)
and the triangle inequality imply
\begin{align*}
    \frac12\sum_{i,j,k}|p_{ijk}-p_jp_{ik}|
    &\le\frac12\sum_{i,j,k}|p_{ijk}-p_{ij}p_k|+\frac12\sum_{i,j}|p_{ij}-p_ip_j|
      +\frac12\sum_{i,k}|p_{ik}-p_ip_k|\\
    &\le\beta(\mathcal A\vee\mathcal B,\mathcal C)
      +\beta(\mathcal A,\mathcal B)
      +\beta(\mathcal A,\mathcal C).
\end{align*}
Here the sums over $k$ and $j$ in the second and third terms reduce to one. The left-hand side is the partition expression for $(B_j)$ and the refinement $(A_i\cap C_k)$.

It remains to show that restricting the partitions in $\mathcal A\vee\mathcal C$ to those of the form $(A_i\cap C_k)$ leaves the supremum in \eqref{eq:supp-beta-partitions} unchanged.
Let $(D_l)_{l=1}^L$ be an arbitrary finite partition of $\Omega$ measurable with respect to $\mathcal A\vee\mathcal C$. Fix a finite partition $(B_j)$ of $\mathcal B$. We define the variation functional evaluated at the partition $(D_l)$ as
$$V(D) = \frac{1}{2} \sum_{j,l} \left\vert{} \Pr(B_j \cap D_l) - \Pr(B_j)\Pr(D_l) \right\vert{}.$$
By definition, the $\sigma$-algebra $\mathcal A\vee\mathcal C$ is generated by the algebra $\mathcal F$ consisting of all finite disjoint unions of rectangles $A\cap C$, where $A\in\mathcal A$ and $C\in\mathcal C$. By the measure approximation theorem \citep[Section 13, Theorem D]{halmos1976measure}, the algebra $\mathcal F$ is dense in $\mathcal A\vee\mathcal C$ under the symmetric difference pseudometric $d(X,Y)=\Pr(X\Delta Y)$. That is, for any $\epsilon>0$, there exists a partition $(E_l)_{l=1}^L$ with cells in $\mathcal F$ such that
$$\sum_{l=1}^L \Pr(D_l \Delta E_l) < \epsilon.$$
The variation functional $V$ is Lipschitz continuous with respect to this approximation. By the triangle inequality and the union bound $\vert{}\Pr(X)-\Pr(Y)\vert{}\le \Pr(X\Delta Y)$, we have
\begin{align*}
    |V(D) - V(E)|&\le \frac{1}{2} \sum_{j,l} \left| \Pr(B_j \cap D_l) - \Pr(B_j \cap E_l) \right| + \frac{1}{2} \sum_{j,l} \Pr(B_j) \left| \Pr(D_l) - \Pr(E_l) \right|\\
    &\le \frac{1}{2} \sum_{j,l} \Pr(B_j \cap (D_l \Delta E_l)) + \frac{1}{2} \sum_{l} \Pr(D_l \Delta E_l) \sum_j \Pr(B_j)\\
    &= \sum_{l=1}^L \Pr(D_l \Delta E_l) < \epsilon.
\end{align*}
Consequently, $V(D) \le V(E) + \epsilon$.

Because each $E_l \in \mathcal F$, it is formed by finite set operations on a finite collection of base sets $\{A^{(m)}\}_{m=1}^M \subset \mathcal A$ and $\{C^{(n)}\}_{n=1}^N \subset \mathcal C$. Let $(A_i)$ be the finite $\mathcal A$-partition generated by all nonempty intersections of the sets $A^{(m)}$ and their complements, and similarly let $(C_k)$ be the finite $\mathcal C$-partition generated by $C^{(n)}$ and their complements.The resulting grid partition $(A_i \cap C_k)$ strictly refines the algebraic partition $(E_l)$. This means there exists an index mapping such that each $E_l = \bigcup_{(i,k) \in I_l} (A_i \cap C_k)$ for disjoint index sets $I_l$.

By the triangle inequality, refining a partition monotonically increases the variation sum:
\begin{align*}
    V(E) &= \frac{1}{2} \sum_{j,l} \left| \sum_{(i,k) \in I_l} \big( \Pr(B_j \cap A_i \cap C_k) - \Pr(B_j)\Pr(A_i \cap C_k) \big) \right| \\
    &\le \frac{1}{2} \sum_{j,l} \sum_{(i,k) \in I_l} \left| \Pr(B_j \cap A_i \cap C_k) - \Pr(B_j)\Pr(A_i \cap C_k) \right| \\
    &= \frac{1}{2} \sum_{i,j,k} \left| \Pr(B_j \cap A_i \cap C_k) - \Pr(B_j)\Pr(A_i \cap C_k) \right| = V(A \cap C).
\end{align*}
Combining the bounds yields $V(D) \le V(A \cap C) + \epsilon$. Since $\epsilon$ is arbitrary, the variation over any $\mathcal A\vee\mathcal C$-partition $D$ is bounded above by the supremum over all grid partitions $(A_i \cap C_k)$. Because every grid partition is itself a valid $\mathcal A\vee\mathcal C$-partition, the supremum over the two classes of partitions must be strictly equal, completing the proof.

For part (iii), take finite partitions $(A_i)$, $(B_j)$, and $(R_k)$ from $\mathcal A$, $\mathcal B$, and $\mathcal R$. Independence of $\mathcal R$ from the joint field $\mathcal A\vee\mathcal B$ gives
\begin{align*}
    &\frac12\sum_{i,j,k}
      \left|\Pr(A_i\cap B_j\cap R_k)
             -\Pr(A_i)\Pr(B_j\cap R_k)\right|\\
    &\qquad=\frac12\sum_{i,j,k}\Pr(R_k)
      \left|\Pr(A_i\cap B_j)-\Pr(A_i)\Pr(B_j)\right|\\
    &\qquad=\frac12\sum_{i,j}
      \left|\Pr(A_i\cap B_j)-\Pr(A_i)\Pr(B_j)\right|.
\end{align*}
The same partition-approximation argument for $\mathcal B\vee\mathcal R$ therefore yields $\beta(\mathcal A,\mathcal B\vee\mathcal R)\le\beta(\mathcal A,\mathcal B)$. The reverse inequality follows from part (i), proving \eqref{eq:supp-beta-independent-extension}.
\end{proof}

\subsection{Estimation lemmas}\label{app:estimation-lemmas}
\begin{lemma}[Empirical averages]\label{lem:supp-oracle-empirical}
Suppose \Cref{ass:iv,ass:estimation-primitive} hold, and fix any $c>0$. Write $\rho_n = \sqrt{\log n/n}$, let $c_3 = c_3(C_\beta,c_\beta)$, and set $\kappa_c = \{32(c+1)/c_3\}^{1/2}$. There is a dimension-free threshold $N_0$, depending only on $c$ and the primitives of \Cref{ass:estimation-primitive}, such that for every
$n\ge \max\{N_0, 3(pr+sr+s+r)\}$,
the following hold simultaneously with probability at least
$1-3(pr+sr+s+r)n^{-(c+1)}$.
\begin{align}
    \|P_n(\boldsymbol\varepsilon_{t+h}\widetilde{\mathbf Z}_t^\top)\|_\Fr &\le \kappa_c\sqrt{pr}\,V_{\varepsilon Z}\,\rho_n,
    \label{eq:supp-ep-oracle}\\
    \|P_n(\widetilde{\mathbf W}_t\widetilde{\mathbf Z}_t^\top)-\overline{\mathbf A}\|_\Fr &\le \kappa_c\sqrt{sr}\,V_{WZ}\,\rho_n,
    \label{eq:supp-ep-firststage}\\
    P_n\|\widetilde{\mathbf W}_t\|_2^2+P_n\|\widetilde{\mathbf Z}_t\|_2^2 &\le 4\,(sK_W^2+rK_Z^2).
    \label{eq:supp-ep-moment}
\end{align}
\end{lemma}

\begin{proof}
We first consider the oracle score \eqref{eq:supp-ep-oracle}. Fix a coordinate $(l,k)$ and set
\[
    \zeta_t = (\boldsymbol\varepsilon_{t+h}\widetilde{\mathbf Z}_t^\top)_{lk}.
\]
It is centered by the exclusion argument \eqref{eq:supp-exclusion}, and \(\|\zeta_t\|_{\senorm}\le K_\varepsilon K_Z\) by the product property of Orlicz norms. If $V_{\varepsilon Z} = 0$, then this centered coordinate process is identically zero almost surely, and the displayed bound is trivial. If $V_{\varepsilon Z} > 0$, then \eqref{eq:lrv-epsZ} ensures its variance proxy is no larger than $V_{\varepsilon Z}$. Hence \Cref{lem:supp-mixing-subexp-bernstein} applies with $V_\zeta = V_{\varepsilon Z}$, so
\[
    |P_n(\boldsymbol\varepsilon_{t+h}\widetilde{\mathbf Z}_t^\top)_{lk}|
    \le \kappa_c V_{\varepsilon Z}\rho_n
\]
with coordinate failure probability at most $3n^{-(c+1)}$. A union bound over the $pr$ coordinates gives failure probability at most $3pr\,n^{-(c+1)}$, which is at most $n^{-c}$ once $n\ge 3pr$, proving \eqref{eq:supp-ep-oracle}.

The first-stage bound \eqref{eq:supp-ep-firststage} is identical. For a coordinate $(j,k)$, take
\[
    \zeta_t = (\widetilde{\mathbf W}_t\widetilde{\mathbf Z}_t^\top)_{jk}-(\mathbf A_t)_{jk}.
\]
This process is centered by the definition of $\mathbf A_t$, and has $\psi_1$ norm bounded by $2K_W K_Z$. Its sample average is the corresponding coordinate of $P_n(\widetilde{\mathbf W}_t\widetilde{\mathbf Z}_t^\top)-\overline{\mathbf A}$. If $V_{WZ} = 0$, every first-stage coordinate process is identically zero almost surely, and \eqref{eq:supp-ep-firststage} is trivial. If $V_{WZ}>0$, then \eqref{eq:lrv-WZ} gives $V_0\le V_{WZ}$ for this coordinate. Hence \Cref{lem:supp-mixing-subexp-bernstein} applies with $V_\zeta = V_{WZ}$. A union bound over the $sr$ coordinates gives failure probability at most $3sr\,n^{-(c+1)}$, which is at most $n^{-c}$ once $n\ge 3sr$. This proves \eqref{eq:supp-ep-firststage}.

It remains to control the empirical second moments. For each coordinate of $\widetilde{\mathbf W}_t$, the squared residual is sub-exponential, with
\[
    \| (\widetilde{\mathbf W}_t)_j^2\|_{\senorm}\le K_W^2
    \quad \Longrightarrow \quad
    \E(\widetilde{\mathbf W}_t)_j^2\le 2K_W^2 .
\]
Apply \Cref{lem:supp-truncated-lrv} to the centered process
$(\widetilde{\mathbf W}_t)_j^2 - \E(\widetilde{\mathbf W}_t)_j^2$. By \eqref{eq:supp-variance-bound}, its variance proxy is bounded by a constant $V_{W,2}\le C K_W^2$, where $C$ depends only on the mixing constants. Applying \Cref{lem:supp-mixing-subexp-bernstein} with $V_\zeta = V_{W,2}$, and taking $N_0$ large enough, makes the Bernstein deviation for this process at most $2K_W^2$. Thus
\[
    P_n(\widetilde{\mathbf W}_t)_j^2\le \sup_t\E(\widetilde{\mathbf W}_t)_j^2+2K_W^2\le 4K_W^2
\]
with coordinate failure probability at most $3n^{-(c+1)}$. The same holds for each coordinate of $\widetilde{\mathbf Z}_t$. The union over the $s+r$ coordinates is at most $n^{-c}$ once $n\ge 3(s+r)$, which proves \eqref{eq:supp-ep-moment}.
\end{proof}

We next bound the first-order term, whose summands pair one nuisance error with a residual factor. Neyman orthogonality and buffered cross-fitting yield a bound governed by the individual error rate $\bar r_T$, with a logarithmic factor and an additional truncation term, as displayed below.

\begin{lemma}[Cross-fitted first-order empirical average]\label{lem:supp-crossfit-firstorder}
Suppose \Cref{ass:estimation-primitive,ass:cross-fitting} and \Cref{cond:nuisance-rate,cond:sg-nuisance-error} hold for nuisance estimates constructed by \Cref{alg:orthogonal-lpiv}, and fix any $c>0$. Let $L$ be the number of cross-fitting folds. With $\rho_n$, $C_c$, $K_\circ$, and $\bar\ell_T$ as in \Cref{thm:supp-general}, there is a dimension-free threshold $N_1$, depending only on $c$, the primitives of \Cref{ass:estimation-primitive}, and the sub-Gaussian bound $\bar K$ of \Cref{cond:sg-nuisance-error}, such that for every $n\ge\max\{N_1, 3Lpr\}$, with probability at least
\begin{align*}
    \Pr(\mathcal E_{T,h})-9Lpr\,n^{-(c+1)}-3L\,C_\beta\exp(-c_\beta b_T),
\end{align*}
it holds that
\begin{equation}\label{eq:supp-ep-firstorder}
    \left\|P_n\left[-\boldsymbol\Delta_{Y,t}\widetilde{\mathbf Z}_t^\top
    +\mathbf B_h\boldsymbol\Delta_{W,t}\widetilde{\mathbf Z}_t^\top
    -\boldsymbol\varepsilon_{t+h}\boldsymbol\Delta_{Z,t}^\top\right]\right\|_\Fr
    \le LC_c\sqrt{pr}\,K_\circ
    \{\bar r_T\sqrt{\bar\ell_T}\,\rho_n+\bar K\{\log n\}^4/n\} .
\end{equation}
This first-order term scales with the individual error size $\bar r_T$ of \Cref{cond:nuisance-rate}(i), not the product rate $r_T$, because each of its three summands pairs one nuisance error with a conditionally mean-zero residual factor.
\end{lemma}

\begin{proof}
Write $\mathcal I_\ell$ for a validation block and
\[
    P_{n,\ell}f = n^{-1}\sum_{t\in\mathcal I_\ell}f_t .
\]
Since $P_n = \sum_{\ell = 1}^L P_{n,\ell}$, it is enough to prove the coordinate bound for one fold and then sum over the fixed number of folds. Fix such a fold, let $\mathcal G$ be its training $\sigma$-field, and denote the nuisance estimates trained for this fold by
$\widehat{\boldsymbol\mu}_{Y,h}^{(-\ell)}$,
$\widehat{\boldsymbol\mu}_{W}^{(-\ell)}$, and
$\widehat{\boldsymbol\mu}_{Z}^{(-\ell)}$.

For any state value $u$, define the fold-specific error functions
\[
    \delta_Y(u) = \widehat{\boldsymbol\mu}_{Y,h}^{(-\ell)}(u)-\boldsymbol\mu_{Y,h}(u),
    \qquad
    \delta_W(u) = \widehat{\boldsymbol\mu}_{W}^{(-\ell)}(u)-\boldsymbol\mu_W(u),
    \qquad
    \delta_Z(u) = \widehat{\boldsymbol\mu}_{Z}^{(-\ell)}(u)-\boldsymbol\mu_Z(u).
\]
Thus, for $t\in\mathcal I_\ell$,
\[
    \boldsymbol\Delta_{Y,t} = \delta_Y(\mathbf U_t),
    \qquad
    \boldsymbol\Delta_{W,t} = \delta_W(\mathbf U_t),
    \qquad
    \boldsymbol\Delta_{Z,t} = \delta_Z(\mathbf U_t),
\]
and the three functions are $\mathcal G$-measurable as random functions.
Let $\mathcal E_{\ell}^{\mathrm{pop}}$ be the event on which, for this fold, the population versions of the three individual-rate bounds in \Cref{cond:nuisance-rate}(i) hold and, in addition, the $\psi_2$ envelope of \Cref{cond:sg-nuisance-error} holds, namely
\[
    \sup_{t\in\mathcal I_{\ell,h}}\max_{a\in\{Y,W,Z\}}
    \|\boldsymbol\delta_{a,\ell}\|_{\psi_2,P_t}\le\bar K .
\]
This event is $\mathcal G$-measurable because the fitted functions are fixed given $\mathcal G$ and the time-specific validation laws $P_t$ are deterministic population objects, so both the population $L_2$ norms and the $\psi_2$ norms above are functions of the training data. Moreover,
\[
    \mathcal E_{T,h}\subseteq\bigcap_{\ell=1}^L
    \mathcal E_{\ell}^{\mathrm{pop}},
\]
because $\mathcal E_{T,h}$ requires both the population and empirical versions of the nuisance bounds, and \Cref{cond:sg-nuisance-error} states that the $\psi_2$ envelope holds on $\mathcal E_{T,h}$. The envelope requirement is built into the definition of $\mathcal E_{\ell}^{\mathrm{pop}}$ deliberately, since \Cref{cond:sg-nuisance-error} guarantees it only on $\mathcal E_{T,h}$, which involves validation data and is not $\mathcal G$-measurable, whereas the conditional argument below must run on a training-measurable event. On $\mathcal E_{\ell}^{\mathrm{pop}}$, evaluating the fixed functions at an independent draw $\mathbf U_t^\circ$ from the time-$t$ validation law gives
\[
    \sup_t\E\{\|\delta_Y(\mathbf U_t^\circ)\|_2^2\mid\mathcal G\}^{1/2}
    \vee\sup_t\E\{\|\delta_W(\mathbf U_t^\circ)\|_2^2\mid\mathcal G\}^{1/2}
    \vee\sup_t\E\{\|\delta_Z(\mathbf U_t^\circ)\|_2^2\mid\mathcal G\}^{1/2}
    \le\bar r_T.
\]
On the same event, the corresponding coordinatewise conditional $\psi_2$ norms are bounded by $\bar K$ by the envelope requirement in the definition of $\mathcal E_{\ell}^{\mathrm{pop}}$. Thus, on $\mathcal E_{\ell}^{\mathrm{pop}}$, the coordinates of $\mathbf B_h\delta_W(\mathbf U_t^\circ)$ have conditional $L_2$ norm at most $C_B\bar r_T$ and conditional $\psi_2$ norm at most $C_B\bar K$.

For a fixed coordinate $(i,j)$, decompose the fold contribution as
\begin{align*}
    S_{\ell,ij}
    & = P_{n,\ell}\left[
    -\boldsymbol\Delta_{Y,t}\widetilde{\mathbf Z}_t^\top
    +\mathbf B_h\boldsymbol\Delta_{W,t}\widetilde{\mathbf Z}_t^\top
    -\boldsymbol\varepsilon_{t+h}\boldsymbol\Delta_{Z,t}^\top
    \right]_{ij}  \\
    & = P_{n,\ell}\zeta^{Y}_{t,ij}
    +P_{n,\ell}\zeta^{W}_{t,ij}
    +P_{n,\ell}\zeta^{Z}_{t,ij},
\end{align*}
where
\[
    \zeta^{Y}_{t,ij} = -\delta_{Y,i}(\mathbf U_t)\cdot(\widetilde{\mathbf Z}_t)_j,
    \qquad
    \zeta^{W}_{t,ij} = (\mathbf B_h\delta_W(\mathbf U_t))_i(\widetilde{\mathbf Z}_t)_j,
    \qquad
    \zeta^{Z}_{t,ij} = -(\boldsymbol\varepsilon_{t+h})_i\delta_{Z,j}(\mathbf U_t).
\]
All three variables are set equal to zero when $t\notin\mathcal I_\ell$. This zero-padding keeps the average in the form $n^{-1}\sum_{t = 1}^n\zeta_t$. Once centering has been established below, it also preserves centering and cannot increase the mixing coefficients or the variance proxy.

\emph{Coupling the validation block.} Let
\[
    \mathbf V_\ell = \{(\mathbf U_t,\boldsymbol\varepsilon_{t+h},\widetilde{\mathbf Z}_t):
    t\in\mathcal I_\ell\},
    \qquad
    \mathcal V = \sigma(\mathbf V_\ell),
\]
and recall that $\mathcal G$ is generated by the training observations indexed by $\mathcal J_{\ell,h}$ and any independent auxiliary learner randomization. Under the observed-data representation, each training observation $(\mathbf Y_{s+h},\mathbf W_s,\mathbf Z_s,\mathbf U_s)$ is measurable with respect to $\mathcal X_{s,h}$, while $\mathcal V$ is contained in $\sigma(\mathcal X_{t,h}:t\in\mathcal I_\ell)$. Write $\mathcal I_\ell=\{a_\ell,\ldots,d_\ell\}$ and let
\[
    \mathcal H_-=\sigma(\mathcal X_{s,h}:s\le a_\ell-b_T),
    \qquad
    \mathcal H_+=\sigma(\mathcal X_{s,h}:s\ge d_\ell+b_T).
\]
Let $\mathcal G_{\mathrm{data}}$ be the field generated by the training observations and let $\mathcal R_\ell$ be the field generated by the auxiliary learner randomization. Then
\[
    \mathcal G=\mathcal G_{\mathrm{data}}\vee\mathcal R_\ell,
    \qquad
    \mathcal G_{\mathrm{data}}\subseteq\mathcal H_-\vee\mathcal H_+.
\]
By \Cref{lem:supp-mixing-triangle}(ii) and the primitive mixing bound,
\begin{align*}
    \beta(\mathcal V,\mathcal H_-\vee\mathcal H_+)
    &\le\beta(\mathcal H_-,\mathcal V)
      +\beta(\mathcal H_-\vee\mathcal V,\mathcal H_+)\\
    &\quad+\beta(\mathcal H_-,\mathcal H_+)\\
    &\le3C_\beta e^{-c_\beta b_T}.
\end{align*}
Each coefficient on the right concerns fields separated by at least $b_T$ time indices. By \Cref{ass:cross-fitting}, $\mathcal R_\ell$ is independent of the entire data process, and hence of $\mathcal V\vee\mathcal H_-\vee\mathcal H_+$. Parts (i) and (iii) of \Cref{lem:supp-mixing-triangle} therefore give
\begin{align*}
    \beta(\mathcal V,\mathcal G)
    &\le\beta(\mathcal V,\mathcal H_-\vee\mathcal H_+\vee\mathcal R_\ell)\\
    &=\beta(\mathcal V,\mathcal H_-\vee\mathcal H_+)\\
    &\le3C_\beta e^{-c_\beta b_T}.
\end{align*}
By Berbee's coupling lemma as stated in Lemma 5.1 and proved in Section 5.3 of \citet{rio2017asymptotic}, on an extended probability space there is a copy
\[
    \mathbf V_\ell^\ast = \{(\mathbf U_t^\ast,\boldsymbol\varepsilon_{t+h}^\ast,
    \widetilde{\mathbf Z}_t^\ast):t\in\mathcal I_\ell\}
\]
such that \(\mathbf V_\ell^\ast\stackrel{d}{ = }\mathbf V_\ell\), $\mathbf V_\ell^\ast$ is independent of $\mathcal G$, and
\[
    \Pr(\mathbf V_\ell^\ast\ne\mathbf V_\ell)
    \le 3C_\beta e^{-c_\beta b_T}.
\]
Let \(\mathcal N_\ell = \{\mathbf V_\ell^\ast\ne\mathbf V_\ell\}\) and define
\[
    \zeta^{Y,\ast}_{t,ij} = -\delta_{Y,i}(\mathbf U_t^\ast)(\widetilde{\mathbf Z}_t^\ast)_j,
    \quad
    \zeta^{W,\ast}_{t,ij} = (\mathbf B_h\delta_W(\mathbf U_t^\ast))_i(\widetilde{\mathbf Z}_t^\ast)_j,
    \quad
    \zeta^{Z,\ast}_{t,ij} = -(\boldsymbol\varepsilon_{t+h}^\ast)_i\delta_{Z,j}(\mathbf U_t^\ast),
\]
again with zero-padding outside $\mathcal I_\ell$. On $\mathcal N_\ell^c$,
\[
    S_{\ell,ij} = P_{n,\ell}\zeta^{Y,\ast}_{t,ij}
    +P_{n,\ell}\zeta^{W,\ast}_{t,ij}
    +P_{n,\ell}\zeta^{Z,\ast}_{t,ij}.
\]
Let \(\mathcal N = \cup_{\ell = 1}^L\mathcal N_\ell\) denote failure of at least one fold coupling. A union bound over the fixed $L$ folds gives
\[
    \Pr(\mathcal N)\le 3L\,C_\beta e^{-c_\beta b_T}.
\]

\emph{Conditional centering.} Conditional on $\mathcal G$, the three nuisance functions in the coupled summands are fixed. Moreover, $\mathbf V_\ell^\ast$ has the same law as the validation block and is independent of $\mathcal G$. Therefore, for $t\in\mathcal I_\ell$,
\begin{align*}
    \E(\zeta^{Y,\ast}_{t,ij}\mid\mathcal G)
    & = -\E\left[
    \delta_{Y,i}(\mathbf U_t^\ast)
    \E\{(\widetilde{\mathbf Z}_t^\ast)_j
        \mid\mathbf U_t^\ast,\mathcal G\}
    \mid\mathcal G\right] = 0,\\
    \E(\zeta^{W,\ast}_{t,ij}\mid\mathcal G)
    & = \E\left[
    (\mathbf B_h\delta_W(\mathbf U_t^\ast))_i
    \E\{(\widetilde{\mathbf Z}_t^\ast)_j
        \mid\mathbf U_t^\ast,\mathcal G\}
    \mid\mathcal G\right] = 0,\\
    \E(\zeta^{Z,\ast}_{t,ij}\mid\mathcal G)
    & = -\E\left[
    \delta_{Z,j}(\mathbf U_t^\ast)
    \E\{(\boldsymbol\varepsilon_{t+h}^\ast)_i
        \mid\mathbf U_t^\ast,\mathcal G\}
    \mid\mathcal G\right] = 0.
\end{align*}
Thus, after conditioning on $\mathcal G$, each zero-padded coordinate process is centered and geometrically mixing with the original mixing envelope.

\emph{Uniform variance proxy.} It remains to bound the conditional long-run variance for any one of the three scalar processes. Write
\[
    \zeta_t^\ast = d(\mathbf U_t^\ast)\xi_t^\ast\I\{t\in\mathcal I_\ell\},
\]
where the pair $(d,\xi_t^\ast)$ is, respectively,
\[
\begin{array}{lll}
    Y\text{-piece:} & d(u) = -\delta_{Y,i}(u),
        & \xi_t^\ast = (\widetilde{\mathbf Z}_t^\ast)_j,\\[2mm]
    W\text{-piece:} & d(u) = (\mathbf B_h\delta_W(u))_i,
        & \xi_t^\ast = (\widetilde{\mathbf Z}_t^\ast)_j,\\[2mm]
    Z\text{-piece:} & d(u) = -\delta_{Z,j}(u),
        & \xi_t^\ast = (\boldsymbol\varepsilon_{t+h}^\ast)_i .
\end{array}
\]
Let $K_\xi = K_Z$ for the $Y$- and $W$-pieces and $K_\xi = K_\varepsilon$ for the $Z$-piece. Let $R_\Delta = 1$ for the $Y$- and $Z$-pieces and $R_\Delta = C_B$ for the $W$-piece. On $\mathcal E_{\ell}^{\mathrm{pop}}$, by its definition,
\[
    \sup_t\E\{d(\mathbf U_t^\ast)^2\mid\mathcal G\}^{1/2}
    \le R_\Delta \bar r_T,
    \qquad
    \sup_t\|d(\mathbf U_t^\ast)\mid\mathcal G\|_{\psi_2}
    \le R_\Delta\bar K .
\]
Because the coupled validation block is independent of $\mathcal G$ and has
the same joint law as the original validation block, the conditional residual
tail bound is preserved after coupling. More precisely, for the relevant scalar
coordinate $\xi_t^\ast$,
\[
    \E\left[
      \exp\{(\xi_t^\ast)^2/K_\xi^2\}
      \,\middle|\,\mathbf U_t^\ast,\mathcal G
    \right]
    =
    \left.
    \E\left[
      \exp\{\xi_t^2/K_\xi^2\}
      \,\middle|\,\mathbf U_t
    \right]\right|_{\mathbf U_t=\mathbf U_t^\ast}
    \le2
    \quad\text{a.s.}
\]
Thus
$\|\xi_t^\ast\mid\mathbf U_t^\ast,\mathcal G\|_{\psi_2}\le K_\xi$.
In particular, conditional second-moment control and the tower property give
\[
    \sup_t\E\{(\zeta_t^\ast)^2\mid\mathcal G\}
    \le C K_\xi^2R_\Delta^2\bar r_T^2 .
\]
Hence Cauchy--Schwarz gives the short-lag bound
\[
    |\cov(\zeta_a^\ast,\zeta_{a+m}^\ast\mid\mathcal G)|
    \le C K_\xi^2R_\Delta^2\bar r_T^2 .
\]
Davydov's covariance inequality gives the mixing-informed bound
\[
    |\cov(\zeta_a^\ast,\zeta_{a+m}^\ast\mid\mathcal G)|
    \le C\alpha(m)^{1/2}\sup_t\|\zeta_t^\ast\mid\mathcal G\|_4^2
    \le C K_\xi^2R_\Delta^2\bar K^2 e^{-c_\beta m/2}.
\]
If $\bar r_T=0$, the population $L_2$ bound makes
$d(\mathbf U_t^\ast)=0$ almost surely for every validation time, so the desired
bound is immediate. Suppose henceforth that $\bar r_T>0$. We choose different bounds when $m$ takes different values. Let \(m_T = \left\lceil 2c_\beta^{-1}\log\{(\bar K/\bar r_T)^2\vee e\}\right\rceil\).
We use the Cauchy-Schwarz bound for $m\le m_T$ and the Davydov bound for $m>m_T$. Then, uniformly in $a$,
\begin{align*}
    \sum_{m\ge 0}|\cov(\zeta_a^\ast,\zeta_{a+m}^\ast\mid\mathcal G)|
    &\le C K_\xi^2R_\Delta^2
    \left\{
        \bar r_T^2(m_T+1)+\bar K^2\sum_{m>m_T}e^{-c_\beta m/2}
    \right\} \\
    &\le C_{\mathrm{mix}}K_\xi^2R_\Delta^2\bar r_T^2
    \{1\vee\log(\bar K/\bar r_T)\}.
\end{align*}
The last inequality uses $\bar K^2e^{-c_\beta m_T/2}\le \bar r_T^2$ when $\bar r_T<\bar K$, while if $\bar r_T\ge\bar K$, the Davydov series is already bounded by a constant multiple of $K_\xi^2R_\Delta^2\bar r_T^2$. Hence
\begin{equation}\label{eq:supp-lrv-firstorder}
    v_\ast^2
    :=\sup_a\left\{
      \var(\zeta_a^\ast\mid\mathcal G)
      +2\sum_{b>a}|\cov(\zeta_a^\ast,\zeta_b^\ast\mid\mathcal G)|
      \right\}
    \le C_{\mathrm{mix}}K_\xi^2R_\Delta^2\bar r_T^2\bar\ell_T,
    \quad
    \bar\ell_T = 1\vee\log(\bar K/\bar r_T),
\end{equation}
where $C_{\mathrm{mix}}$ depends only on $(C_\beta,c_\beta)$.

The same variance bound holds uniformly for the conditionally centered truncations
\[
    \zeta_t^{\ast,(M)}
    =
    \zeta_t^\ast\I\{|\zeta_t^\ast|\le M\}
    -\E[\zeta_t^\ast\I\{|\zeta_t^\ast|\le M\}\mid\mathcal G].
\]
Indeed, conditionally on $\mathcal G$,
\[
    \|\zeta_t^{\ast,(M)}\|_2^2
    \le \E\{(\zeta_t^\ast)^2\mid\mathcal G\},
    \qquad
    \|\zeta_t^{\ast,(M)}\|_4
    \le 2\|\zeta_t^\ast\|_4 .
\]
Consequently, the Cauchy--Schwarz and Davydov bounds above apply to
$\zeta_t^{\ast,(M)}$ with the same orders. Using the same split at $m_T$ gives, uniformly in $M$,
\begin{equation}\label{eq:supp-lrv-firstorder-truncated}
    v_{\ast,M}^2
    \le C_{\mathrm{mix}}K_\xi^2R_\Delta^2
    \bar r_T^2\bar\ell_T ,
\end{equation}
where $v_{\ast,M}^2$ is the conditional long-run variance proxy of the centered truncated process.

\emph{Bernstein bound and union.} Set
\[
    V_\ast = K_\xi R_\Delta \bar r_T\sqrt{\bar\ell_T},
    \qquad
    B_\ast = K_\xi R_\Delta\bar K,
    \qquad
    \delta_n = C_c\{V_\ast\rho_n+B_\ast\log^4 n/n\}.
\]
The conditional sub-Gaussian controls and the product property give
$\sup_t\|\zeta_t^\ast\mid\mathcal G\|_{\psi_1}\le C_\psi B_\ast$
for a universal constant $C_\psi$. We may suppose $B_\ast>0$, since otherwise
the process is identically zero. Truncate at
\[
    M_n=(c+2)C_\psi B_\ast\log n .
\]
The conditional sub-exponential tail bound and a union bound yield
\begin{equation}\label{eq:supp-firstorder-max-tail}
    \Pr\left(
    \max_{1\le t\le n}|\zeta_t^\ast|>M_n
    \,\middle|\,\mathcal G
    \right)
    \le 2n^{-(c+1)} .
\end{equation}
Let
\[
    \mu_{n,t}
    =
    \E[\zeta_t^\ast\I\{|\zeta_t^\ast|\le M_n\}\mid\mathcal G].
\]
Conditional centering and tail integration give
\[
    \max_t|\mu_{n,t}|
    \le
    2(M_n+C_\psi B_\ast)
    \exp\{-M_n/(C_\psi B_\ast)\}
    \le \delta_n/2
\]
for all sufficiently large $n$, with a threshold depending on $c$ but not on
$\bar r_T$. On the complement of the event in
\eqref{eq:supp-firstorder-max-tail},
$P_{n,\ell}\zeta_t^\ast
=P_{n,\ell}\zeta_t^{\ast,(M_n)}+P_{n,\ell}\mu_{n,t}$.

Apply \eqref{eq:supp-mpr} to the centered truncated process using
\eqref{eq:supp-lrv-firstorder-truncated} and set $x=n\delta_n/2$.
The denominator in \eqref{eq:supp-mpr} is bounded, up to constants depending
only on $c$ and the mixing constants, by
\[
    nV_\ast^2+B_\ast^2\log^2 n
    +B_\ast V_\ast\sqrt n\log^{7/2} n
    +B_\ast^2\log^7 n
    \le C\{nV_\ast^2+B_\ast^2\log^7 n\},
\]
where the cross term is absorbed by $2ab\le a^2+b^2$. On the other hand,
\[
    x^2
    \ge c_0 C_c^2\{nV_\ast^2\log n+B_\ast^2\log^8 n\}
\]
for a numerical constant $c_0>0$. Thus, by choosing $C_c$ sufficiently large
as a function only of $c$ and the mixing constants, the truncated-process
probability is at most $n^{-(c+1)}$, uniformly in $\bar r_T$. Combining this
bound with \eqref{eq:supp-firstorder-max-tail} and the centering-bias bound
gives, for each coordinate,
\[
    |P_{n,\ell}\zeta_t^\ast|
    \le C_c K_\xi
    R_\Delta(\bar r_T\sqrt{\bar\ell_T}\rho_n+\bar K\log^4 n/n)
\]
with conditional probability at least $1-3n^{-(c+1)}$, once $n\ge N_1$. Summing the three pieces gives the coordinate bound for this fold contribution with coefficient
$K_\varepsilon+(1+C_B)K_Z = K_\circ$. Summing these bounds over the fixed $L$ fold contributions gives the displayed $L$-multiple in \eqref{eq:supp-ep-firstorder}. The union over the three pieces, the $pr$ coordinates, and the $L$ folds has probability at most $9Lpr\,n^{-(c+1)}$. For every training realization in
$\mathcal E_{\ell}^{\mathrm{pop}}$, integrating the conditional bound over $\mathcal G$ removes the conditioning. Taking the union over folds and using
$\mathcal E_{T,h}\subseteq\cap_{\ell=1}^L
\mathcal E_{\ell}^{\mathrm{pop}}$, together with the coupling failure probability, gives \eqref{eq:supp-ep-firstorder} with the probability stated in the lemma.
\end{proof}

\begin{lemma}[Perturbation of the sample first stage moment]\label{lem:supp-first-stage}
Suppose $\mathcal E_{T,h}$ occurs and the bounds \eqref{eq:supp-ep-firststage} and \eqref{eq:supp-ep-moment} hold. Then \(\|\widehat{\mathbf A}_h-\overline{\mathbf A}\|_\Fr\le D_A\), where $D_A$ is the first-stage deviation in \eqref{eq:supp-first-stage-smallness}. If $D_A\le\underline\sigma_A/2$, then $\widehat{\mathbf A}_h$ has full row rank and
\[
    \|\widehat{\mathbf A}_h^\dagger\|_{\op}\le \frac{2}{\underline\sigma_A},
    \qquad
    \|\widehat{\mathbf A}_h^\dagger-\overline{\mathbf A}^\dagger\|_{\op}
    \le \frac{3}{\underline\sigma_A^2}\,\|\widehat{\mathbf A}_h-\overline{\mathbf A}\|_{\op}
    \le \frac{3}{\underline\sigma_A^2}\,D_A .
\]
\end{lemma}
This lemma bounds how far the estimated first-stage matrix $\widehat{\mathbf A}_h$ can be from the average population moment $\overline{\mathbf A}$, and converts this into a bound on the pseudoinverses.
\begin{proof}
Expanding the sample first-stage matrix gives
\begin{align*}
    \widehat{\mathbf A}_h
    & = P_n\{(\widetilde{\mathbf W}_t-\boldsymbol\Delta_{W,t})(\widetilde{\mathbf Z}_t-\boldsymbol\Delta_{Z,t})^\top\} \\
    & = P_n(\widetilde{\mathbf W}_t\widetilde{\mathbf Z}_t^\top)
      -P_n(\boldsymbol\Delta_{W,t}\widetilde{\mathbf Z}_t^\top)
      -P_n(\widetilde{\mathbf W}_t\boldsymbol\Delta_{Z,t}^\top)
      +P_n(\boldsymbol\Delta_{W,t}\boldsymbol\Delta_{Z,t}^\top).
\end{align*}
By \eqref{eq:supp-ep-firststage} the oracle term is within $\kappa_c\sqrt{sr}\,V_{WZ}\rho_n$ of $\overline{\mathbf A}$. By Cauchy--Schwarz and \eqref{eq:supp-ep-moment}, the two linear terms pair a single nuisance error with a residual factor, so they scale with the individual rate,
\[
    \|P_n(\boldsymbol\Delta_{W,t}\widetilde{\mathbf Z}_t^\top)\|_\Fr
    \le \|\boldsymbol\Delta_W\|_{n,h,2}\,(P_n\|\widetilde{\mathbf Z}_t\|_2^2)^{1/2}
    \le \bar r_T\cdot 2\sqrt{r}\,K_Z,
\]
and likewise \(\|P_n(\widetilde{\mathbf W}_t\boldsymbol\Delta_{Z,t}^\top)\|_\Fr\le 2\sqrt{s}\,K_W\,\bar r_T\), using \(\|\boldsymbol\Delta_W\|_{n,h,2}\vee\|\boldsymbol\Delta_Z\|_{n,h,2}\le \bar r_T\) from \Cref{cond:nuisance-rate}(i) on $\mathcal E_{T,h}$. The quadratic term pairs two nuisance errors, so the product rate applies, \(\|P_n(\boldsymbol\Delta_{W,t}\boldsymbol\Delta_{Z,t}^\top)\|_\Fr\le\|\boldsymbol\Delta_W\|_{n,h,2}\|\boldsymbol\Delta_Z\|_{n,h,2}\le r_T^2\) by \Cref{cond:nuisance-rate}(ii). Summing the three nuisance contributions gives \(\|\widehat{\mathbf A}_h-\overline{\mathbf A}\|_\Fr\le D_A\) with $D_A$ as in \eqref{eq:supp-first-stage-smallness}.

If $D_A\le \underline\sigma_A/2$, then \(\|\widehat{\mathbf A}_h-\overline{\mathbf A}\|_{\op}\le\underline\sigma_A/2\), and by Weyl's inequality for singular values,
\[\sigma_{\min}(\widehat{\mathbf A}_h)\ge \sigma_{\min}(\overline{\mathbf A})-\|\widehat{\mathbf A}_h-\overline{\mathbf A}\|_{\op}\ge \underline\sigma_A/2 .\]
Hence $\widehat{\mathbf A}_h$ has full row rank, with \(\|\widehat{\mathbf A}_h^\dagger\|_{\op} = \sigma_{\min}(\widehat{\mathbf A}_h)^{-1}\le 2/\underline\sigma_A\). Because $\overline{\mathbf A}$ and $\widehat{\mathbf A}_h$ both have full row rank $s$, the pseudo-inverse perturbation bound of \citet[Theorem 4.1]{wedin1973perturbation} gives
\[
    \|\widehat{\mathbf A}_h^\dagger-\overline{\mathbf A}^\dagger\|_{\op}
    \le \sqrt{2}\|\widehat{\mathbf A}_h^\dagger\|_{\op}\|\overline{\mathbf A}^\dagger\|_{\op}\|\widehat{\mathbf A}_h-\overline{\mathbf A}\|_{\op}
    \le \frac{3}{\underline\sigma_A^2}\,\|\widehat{\mathbf A}_h-\overline{\mathbf A}\|_{\op},
\]
which is at most $(3/\underline\sigma_A^2)D_A$.
\end{proof}

\section{Proofs for uncertainty quantification}\label{sec:inference-proof}

This section proves the central limit theorem and feasible inference results in \Cref{sec:inference}. We keep the horizon $h$ fixed, write $n=n_h$, and use $P_n=P_{n,h}$. Let
\[
    \mathbf q_{t,h}
    :=\vectorize(\boldsymbol\varepsilon_{t+h}\widetilde{\mathbf Z}_t^\top),
    \qquad
    \mathbf G_{h,T}
    :=\overline{\mathbf A}_{h,T}^{\dagger\top}\otimes\mathbf I_p,
    \qquad
    \mathbf G_{h,\infty}
    :=\mathbf A_\infty^{\dagger\top}\otimes\mathbf I_p.
\]
The exact finite-sample influence vector is
\[
    \boldsymbol\psi_{t,h}^{(T)}=\mathbf G_{h,T}\mathbf q_{t,h}.
\]
Under \Cref{ass:iv,ass:estimation-primitive,ass:uq-stationarity}, $\mathbf q_{t,h}$ is centered, strictly stationary, geometrically $\beta$-mixing, and sub-exponential. Hence every fixed linear transformation of $\mathbf q_{t,h}$ has absolutely summable autocovariances. Moreover, continuity of the pseudoinverse at the full row-rank matrix $\mathbf A_\infty$ gives
\[
    \mathbf G_{h,T}\longrightarrow\mathbf G_{h,\infty}.
\]
Consequently, $\|\mathbf G_{h,T}\|_{\op}$ is bounded for all sufficiently large $T$.

\subsection{Proof of the joint central limit theorem}\label{app:vector-clt-proof}

\begin{proof}[Proof of \Cref{thm:vector-clt}]
Let
\[
    \mathbf S_n:=\frac1{\sqrt n}\sum_{t=1}^n\mathbf q_{t,h}.
\]
Absolute summability of the coordinatewise autocovariances implies $\E\|\mathbf S_n\|^2=O(1)$, so $\mathbf S_n=O_p(1)$. Vectorizing the oracle expansion \eqref{eq:orthogonal-linear-expansion} gives
\[
    \sqrt n\,\vectorize(\widehat{\mathbf B}_h-\mathbf B_h)
    =\mathbf G_{h,T}\mathbf S_n
     +\sqrt n\,\vectorize(\mathbf R_{T,h}).
\]
Fix $\lambda\in\mathbb R^{ps}$. Since $\mathbf G_{h,T}\to\mathbf G_{h,\infty}$ and $\mathbf S_n=O_p(1)$,
\[
    \lambda^\top(\mathbf G_{h,T}-\mathbf G_{h,\infty})\mathbf S_n=o_p(1).
\]
Also, by \Cref{thm:orthogonal-estimation}(ii),
\[
    \left|\sqrt n\,\lambda^\top\vectorize(\mathbf R_{T,h})\right|
    \le\|\lambda\|\sqrt n\,\|\mathbf R_{T,h}\|_\Fr=o_p(1).
\]
It remains to study
\[
    \lambda^\top\mathbf G_{h,\infty}\mathbf S_n
    =\frac1{\sqrt n}\sum_{t=1}^n\zeta_t(\lambda),
    \qquad
    \zeta_t(\lambda):=\lambda^\top\mathbf G_{h,\infty}\mathbf q_{t,h}.
\]
This is a centered, strictly stationary, geometrically mixing, sub-exponential scalar sequence. Its long-run variance is
$\omega^2=\lambda^\top\boldsymbol\Omega_h\lambda$.
If $\omega>0$, then $\E|\zeta_t(\lambda)|^3<\infty$ and the geometric mixing rate gives $\sum_{m\ge1}\alpha(m)^{1/3}<\infty$. The strong-mixing central limit theorem of \citet{ibragimov1962some} therefore gives
\[
    \frac1{\sqrt n}\sum_{t=1}^n\zeta_t(\lambda)
    \dto N(0,\omega^2).
\]
If $\omega=0$, let $\gamma_\lambda(m)=\cov\{\zeta_t(\lambda),\zeta_{t+m}(\lambda)\}$. The Davydov bound gives $\sum_m|\gamma_\lambda(m)|<\infty$, and therefore
\[
    \var\!\left\{\frac1{\sqrt n}\sum_{t=1}^n\zeta_t(\lambda)\right\}
    =\sum_{|m|<n}\left(1-\frac{|m|}{n}\right)\gamma_\lambda(m)
    \longrightarrow\sum_{m\in\mathbb Z}\gamma_\lambda(m)=0.
\]
The sum then converges to zero in mean square. Thus, for every $\lambda$,
\[
    \sqrt n\,\lambda^\top\vectorize(\widehat{\mathbf B}_h-\mathbf B_h)
    \dto N(0,\lambda^\top\boldsymbol\Omega_h\lambda).
\]
The Cramer--Wold theorem proves the stated joint convergence, including the degenerate case.

For a scalar contrast, the identity
\[
    \sqrt n\{\widehat\theta_h(a,b)-\theta_h(a,b)\}
    =(b\otimes a)^\top
      \sqrt n\,\vectorize(\widehat{\mathbf B}_h-\mathbf B_h)
\]
and \Cref{thm:vector-clt} give convergence to a centered normal variable with variance $\sigma_h^2(a,b)$. Division by the positive constant $\sigma_h(a,b)$ proves the result.
\end{proof}

\begin{remark}[Possible quantitative refinement]
For a fixed scalar contrast, the stationary strong-mixing oracle sum can also admit a quantitative normal approximation under the present moment and dependence conditions; see \citet{tikhomirov1981rate}. Transferring such a bound to a feasible confidence interval would additionally require explicit high-probability rates for the DML remainder, HAC estimation, and studentization. We do not pursue that refinement here.
\end{remark}

\subsection{Feasible long-run variance and Wald inference} \label{app:wald-inference}

Before proving \Cref{prop:hac-consistent}, we verify the claims of \Cref{rem:vector-plugin-sufficient}, namely that the estimation theory of \Cref{thm:orthogonal-estimation} supplies the plug-in rate required by \Cref{cond:vector-plugin} under an explicit bandwidth restriction.

\begin{lemma}[Plug-in stability rate; Proof of Remark \ref{rem:vector-plugin-sufficient}]\label{lem:supp-plugin-rate}
Suppose the conditions of \Cref{thm:orthogonal-estimation} hold. Then the plug-in influence vectors in \eqref{eq:vector-hac} satisfy
\begin{equation}\label{eq:supp-plugin-rate}
    P_n\big\|\widehat{\boldsymbol\psi}_{t,h}-\boldsymbol\psi^{(T)}_{t,h}\big\|_2^2
    =O_p\!\left(\bar r_T^2\log n+\frac{\log n}{n}\right).
\end{equation}
Consequently, \Cref{cond:vector-plugin} holds for every bandwidth sequence $\ell_n\to\infty$ with
\(
    \ell_n\bar r_T^2\log n\to0
\)
and
\(
    \ell_n\log n/n\to0 .
\)
\end{lemma}

\begin{proof}
All constants below may depend on the fixed dimensions $(p,s,r,L)$ and the model primitives $(C_\beta,c_\beta,K_W,K_Z,K_\varepsilon,C_B,\bar K,\underline\sigma_A)$, but not on $T$. Write $\rho_n=\sqrt{\log n/n}$ and $\overline{\mathbf A}=\overline{\mathbf A}_{h,T}$. By \Cref{lem:supp-first-stage} and the proof of \Cref{thm:orthogonal-estimation},
\[
    \|\widehat{\mathbf A}_h^\dagger\|_{\op}=O_p(1),
    \qquad
    \|\widehat{\mathbf A}_h^\dagger-\overline{\mathbf A}^\dagger\|_{\op}
    =O_p(\rho_n+\bar r_T).
\]
Indeed, the first-stage bounds hold on the event $\mathcal G_n$ used there, whose probability tends to one, and $D_A=O(\rho_n+\bar r_T)\to0$.

We first record a maximum bound for the fitted instrument residuals. Use the foldwise coupling and the training-measurable events $\mathcal E_\ell^{\mathrm{pop}}$ from the proof of \Cref{lem:supp-crossfit-firstorder}. On each such event, the coordinates of $\boldsymbol\delta_{Z,\ell}(\mathbf U_t^\ast)$ have conditional $\psi_2$ norm at most $\bar K$. A conditional union bound over validation times and coordinates, followed by the coupling bound, gives, for every $u>0$,
\[
    \Pr\!\left(\max_{t\le n}\|\boldsymbol\Delta_{Z,t}\|_2>u\right)
    \le \Pr(\mathcal E_{T,h}^{\,c})
       +3LC_\beta e^{-c_\beta b_T}
       +2nr\exp\!\left(-\frac{u^2}{r\bar K^2}\right).
\]
Here $\mathcal E_{T,h}\subseteq\cap_\ell\mathcal E_\ell^{\mathrm{pop}}$, and no independence within a validation block is needed for the union bound. Taking $u$ to be a sufficiently large constant times $\sqrt{\log n}$ proves $\max_t\|\boldsymbol\Delta_{Z,t}\|_2^2=O_p(\log n)$. The marginal sub-Gaussian bounds in \Cref{ass:estimation-primitive}(iii) give the same order for the maxima of $\|\widetilde{\mathbf Z}_t\|_2^2$ and $\|\boldsymbol\varepsilon_{t+h}\|_2^2$. Consequently,
\[
    M_n:=\max_{t\le n}\left\{
       \|\widehat{\mathbf e}_{Z,t}\|_2^2
       +\|\boldsymbol\varepsilon_{t+h}\|_2^2\right\}
    =O_p(\log n).
\]

Set
\[
    \mathbf d_t
    :=-(\widehat{\mathbf B}_h-\mathbf B_h)\widetilde{\mathbf W}_t
      -\boldsymbol\Delta_{Y,t}
      +\widehat{\mathbf B}_h\boldsymbol\Delta_{W,t},
\]
so that $\widehat{\mathbf e}_{Y,t,h}-\widehat{\mathbf B}_h\widehat{\mathbf e}_{W,t}=\boldsymbol\varepsilon_{t+h}+\mathbf d_t$. By \Cref{thm:orthogonal-estimation}(ii), $\|\widehat{\mathbf B}_h-\mathbf B_h\|_\Fr=O_p(n^{-1/2})$ and $\|\widehat{\mathbf B}_h\|_{\op}=O_p(1)$. Also, the sub-Gaussian moment bounds and Markov's inequality give
\[
    P_n\|\widetilde{\mathbf W}_t\|_2^2=O_p(1),
    \qquad
    P_n\|\boldsymbol\varepsilon_{t+h}\widetilde{\mathbf Z}_t^\top\|_\Fr^2=O_p(1),
\]
where the second bound follows from Cauchy--Schwarz and the uniformly bounded fourth moments. Since $P_n\|\boldsymbol\Delta_{a,t}\|_2^2\le\bar r_T^2$ for $a\in\{Y,W,Z\}$ on $\mathcal E_{T,h}$, expanding $\mathbf d_t$ gives
\[
\begin{split}
    P_n\|\mathbf d_t\|_2^2
    \le{}&3\|\widehat{\mathbf B}_h-\mathbf B_h\|_{\op}^2P_n\|\widetilde{\mathbf W}_t\|_2^2
       +3P_n\|\boldsymbol\Delta_{Y,t}\|_2^2\\
       &+3\|\widehat{\mathbf B}_h\|_{\op}^2P_n\|\boldsymbol\Delta_{W,t}\|_2^2
    =O_p(n^{-1}+\bar r_T^2).
\end{split}
\]
The difference of the influence matrices before vectorization is exactly
\[
    \boldsymbol\varepsilon_{t+h}\widetilde{\mathbf Z}_t^\top
       (\widehat{\mathbf A}_h^\dagger-\overline{\mathbf A}^\dagger)
    +\left(\mathbf d_t\widehat{\mathbf e}_{Z,t}^\top
       -\boldsymbol\varepsilon_{t+h}\boldsymbol\Delta_{Z,t}^\top\right)
       \widehat{\mathbf A}_h^\dagger.
\]
Using $\|\vectorize(\mathbf H)\|_2=\|\mathbf H\|_\Fr$ and $\|\mathbf a\mathbf b^\top\|_\Fr=\|\mathbf a\|_2\|\mathbf b\|_2$, we obtain
\[
\begin{split}
    P_n\|\widehat{\boldsymbol\psi}_{t,h}-\boldsymbol\psi^{(T)}_{t,h}\|_2^2
    \le{}&3\|\widehat{\mathbf A}_h^\dagger-\overline{\mathbf A}^\dagger\|_{\op}^2
       P_n\|\boldsymbol\varepsilon_{t+h}\widetilde{\mathbf Z}_t^\top\|_\Fr^2\\
    &+3\|\widehat{\mathbf A}_h^\dagger\|_{\op}^2M_n
       \left(P_n\|\mathbf d_t\|_2^2+P_n\|\boldsymbol\Delta_{Z,t}\|_2^2\right)\\
    ={}&O_p\!\left(\rho_n^2+\bar r_T^2
                   +\log n\{n^{-1}+\bar r_T^2\}\right)
     =O_p\!\left(\bar r_T^2\log n+\frac{\log n}{n}\right).
\end{split}
\]
This proves \eqref{eq:supp-plugin-rate}. Multiplication by $\ell_n$ gives the required $o_p(1)$ bound because $\ell_n\bar r_T^2\log n\to0$ and $\ell_n\log n/n\to0$. The latter also implies $\ell_n/n\to0$, completing the verification of \Cref{cond:vector-plugin}.
\end{proof}

\begin{proof}[Proof of \Cref{prop:hac-consistent}]
Let $\boldsymbol\Psi_n$ and $\widehat{\boldsymbol\Psi}_n$ be the $n\times ps$ matrices with rows $\boldsymbol\psi_{t,h}^{(T)\top}$ and $\widehat{\boldsymbol\psi}_{t,h}^\top$, respectively, and set
\[
    \mathbf K_n=\left[K\!\left(\frac{t-s}{\ell_n}\right)\right]_{t,s=1}^n,
    \qquad
    \boldsymbol\Omega_n^{\mathrm o}
    :=n^{-1}\boldsymbol\Psi_n^\top\mathbf K_n\boldsymbol\Psi_n.
\]
The Bartlett matrix $\mathbf K_n$ is positive semidefinite by the construction underlying \citet[Theorem 1]{newey1987simple}, and its absolute row sums give $\|\mathbf K_n\|_{\op}\le1+2\ell_n$.

We apply \citet[Theorem 2.1]{dejongdavidson2000consistency} to the array $\mathbf X_{nt}=n^{-1/2}\boldsymbol\psi_{t,h}^{(T)}=n^{-1/2}\mathbf G_{h,T}\mathbf q_{t,h}$. To verify their assumptions, take the underlying mixing process to be $\mathbf q_{t,h}$ itself, so that the near-epoch approximation error is zero. The geometric mixing bound, uniformly bounded fourth moments, and boundedness of $\mathbf G_{h,T}$ verify their Assumption 2 with $c_{nt}=n^{-1/2}$. The Bartlett kernel belongs to their kernel class, and their bandwidth requirement reduces to $\ell_n\to\infty$ and $\ell_n/n\to0$, both implied by \Cref{cond:vector-plugin}. Thus $\boldsymbol\Omega_n^{\mathrm o}-\var(n^{-1/2}\sum_t\boldsymbol\psi_{t,h}^{(T)})\pto\mathbf0$. Absolute summability of the autocovariances of $\mathbf q_{t,h}$ and $\mathbf G_{h,T}\to\mathbf G_{h,\infty}$ give
\[
    \var\!\left(n^{-1/2}\sum_{t=1}^n\boldsymbol\psi_{t,h}^{(T)}\right)
    =\mathbf G_{h,T}\left\{\sum_{|m|<n}
       \left(1-\frac{|m|}{n}\right)
       \cov(\mathbf q_{0,h},\mathbf q_{m,h})\right\}\mathbf G_{h,T}^\top
    \longrightarrow\boldsymbol\Omega_h.
\]
Consequently, $\boldsymbol\Omega_n^{\mathrm o}\pto\boldsymbol\Omega_h$ and $\trace(\boldsymbol\Omega_n^{\mathrm o})=O_p(1)$.

The fitted influence vectors are already centered. Indeed, the closed form of $\widehat{\mathbf B}_h$ and the Moore--Penrose identity $\widehat{\mathbf A}_h^\dagger\widehat{\mathbf A}_h\widehat{\mathbf A}_h^\dagger=\widehat{\mathbf A}_h^\dagger$ give
\[
    P_n\widehat{\boldsymbol\psi}_{t,h}
    =\vectorize\!\left\{
       (\widehat{\mathbf M}_h-\widehat{\mathbf B}_h\widehat{\mathbf A}_h)
       \widehat{\mathbf A}_h^\dagger\right\}
    =\mathbf0.
\]
Hence the $\widehat{\boldsymbol\Omega}_h$ in \eqref{eq:vector-hac} equals $n^{-1}\widehat{\boldsymbol\Psi}_n^\top\mathbf K_n\widehat{\boldsymbol\Psi}_n$. Write $\mathbf D_n=\widehat{\boldsymbol\Psi}_n-\boldsymbol\Psi_n$ and
\[
    \delta_n^2:=n^{-1}\|\mathbf D_n\|_\Fr^2
       =P_n\|\widehat{\boldsymbol\psi}_{t,h}-\boldsymbol\psi_{t,h}^{(T)}\|_2^2,
    \qquad
    a_n^2:=n^{-1}\|\mathbf K_n^{1/2}\mathbf D_n\|_\Fr^2
       \le(1+2\ell_n)\delta_n^2=o_p(1),
\]
where the last equality follows from \Cref{cond:vector-plugin}. Expanding the two quadratic forms and applying Cauchy--Schwarz to $\mathbf K_n^{1/2}\boldsymbol\Psi_n$ and $\mathbf K_n^{1/2}\mathbf D_n$ yields
\[
    \|\widehat{\boldsymbol\Omega}_h-\boldsymbol\Omega_n^{\mathrm o}\|_\Fr
    \le 2\{\trace(\boldsymbol\Omega_n^{\mathrm o})\}^{1/2}a_n+a_n^2
    =o_p(1).
\]
This proves $\widehat{\boldsymbol\Omega}_h\pto\boldsymbol\Omega_h$. The contrast result follows by continuity of the quadratic form $(b\otimes a)^\top\boldsymbol\Omega(b\otimes a)$.
\end{proof}

\begin{proof}[Proof of \Cref{cor:vector-wald}]
Let $\mathbf S_n=\sqrt n\,\vectorize(\widehat{\mathbf B}_h-\mathbf B_h)$. By \Cref{thm:vector-clt}, $\mathbf S_n\dto\mathbf G\sim N(\mathbf 0,\boldsymbol\Omega_h)$. By \Cref{prop:hac-consistent}, $\widehat{\boldsymbol\Omega}_h\to_p\boldsymbol\Omega_h$. Since $\boldsymbol\Omega_h$ is invertible, matrix inversion is continuous at $\boldsymbol\Omega_h$, so $\widehat{\boldsymbol\Omega}_h^{-1}\to_p\boldsymbol\Omega_h^{-1}$. The continuous mapping and Slutsky theorems give
\[
    \mathbf S_n^\top\widehat{\boldsymbol\Omega}_h^{-1}\mathbf S_n
    \dto
    \mathbf G^\top\boldsymbol\Omega_h^{-1}\mathbf G
    \sim\chi^2_{ps}.
\]
The confidence ellipsoid follows by inverting the quadratic form. For a fixed full-row-rank $\mathbf M\in\mathbb R^{k\times ps}$, the same argument applied to $\mathbf M\mathbf S_n$ gives the $\chi_k^2$ limit.
\end{proof}

\begin{proof}[Proof of \Cref{cor:studentized-inference}]
By the scalar CLT in \Cref{thm:vector-clt}, the infeasible standardized statistic with denominator $\sigma_h(a,b)$ converges to $N(0,1)$. By \Cref{prop:hac-consistent} and $\sigma_h(a,b)>0$, $s(a,b)/\sigma_h(a,b)\to_p1$, where $s(a,b)$ denotes the positive square root of $s^2(a,b)$. Slutsky theorem gives the studentized convergence. The displayed interval follows by inverting the limiting normal distribution.
\end{proof}

\section{Buffered transfer of temporal learner rates}\label{app:nuisance-rate-transfer}

This section proves the general transfer theorem in \Cref{thm:regression-rate-transfer} and its application to the IV nuisance regressions in \Cref{cor:iv-nuisance-rates}. The argument converts population prediction error into error on the observed validation blocks. The concrete learners and their proofs in \Cref{sec:nuisance-examples} supply the population rates through prediction bounds for stationary principal-component regression (PCR) and spline least squares.

\subsection{Proof of the general transfer theorem}\label{app:regression-rate-transfer-proof}
\begin{proof}[Proof of \Cref{thm:regression-rate-transfer}]
Let $\mathcal D_t$ denote the full observed data vector, including $\mathbf U_t$ and all variables used for fitting. For each fold $\ell$, let $\mathcal G_\ell$ be the field generated by its training observations and independent learner randomization. Fix the regression target $\boldsymbol\mu$ and a fold $\ell$, and write $\mathcal I_{\ell,h}=\{a_\ell,\ldots,d_\ell\}$. Let
\[
    \mathbf V_\ell=\{\mathcal D_t:t\in\mathcal I_{\ell,h}\},
    \qquad \mathcal V_\ell=\sigma(\mathbf V_\ell),
\]
and define the fields
\[
    \mathcal H_-=\sigma(\mathcal D_s:s\le a_\ell-b_T),
    \qquad
    \mathcal H_+=\sigma(\mathcal D_s:s\ge d_\ell+b_T).
\]
Write $\mathcal G_{\ell,\mathrm{data}}$ for the training-data field and $\mathcal R_\ell$ for the field generated by the learner randomization. By \Cref{ass:cross-fitting},
\[
    \mathcal G_\ell=\mathcal G_{\ell,\mathrm{data}}\vee\mathcal R_\ell,
    \qquad
    \mathcal G_{\ell,\mathrm{data}}\subseteq\mathcal H_-\vee\mathcal H_+,
\]
and $\mathcal R_\ell$ is independent of the entire data process. Applying \Cref{lem:supp-mixing-triangle}(ii) gives
\begin{align*}
    \beta(\mathcal V_\ell,\mathcal H_-\vee\mathcal H_+)
    &\le\beta(\mathcal H_-,\mathcal V_\ell)
      +\beta(\mathcal H_-\vee\mathcal V_\ell,\mathcal H_+)\\
    &\quad+\beta(\mathcal H_-,\mathcal H_+)\\
    &\le3\beta(b_T),
\end{align*}
since each pair of fields on the right is separated by at least $b_T$ time indices. Parts (i) and (iii) of \Cref{lem:supp-mixing-triangle} then give
\begin{align*}
    \beta(\mathcal V_\ell,\mathcal G_\ell)
    &\le\beta(\mathcal V_\ell,\mathcal H_-\vee\mathcal H_+\vee\mathcal R_\ell)\\
    &=\beta(\mathcal V_\ell,\mathcal H_-\vee\mathcal H_+)\\
    &\le3\beta(b_T).
\end{align*}
By Berbee's coupling lemma, as stated in Lemma 5.1 and proved in Section 5.3 of \citet{rio2017asymptotic}, an extended probability space supports a copy $\mathbf V_\ell^\ast$ with the same law as $\mathbf V_\ell$, independent of $\mathcal G_\ell$, such that
\[
    \Pr(\mathbf V_\ell^\ast\ne\mathbf V_\ell)
    \le 3\beta(b_T).
\]
Construct these copies for the fixed number of folds on a common extension. They need not be mutually independent. A union bound gives
\[
    \Pr\!\left(\bigcup_{\ell=1}^L
       \{\mathbf V_\ell^\ast\ne\mathbf V_\ell\}\right)
    \le 3L\beta(b_T).
\]

Write $\mathbf U_t^\ast$ for the state coordinate in the copied block. The assumed measurable prediction rule means that $(\omega,u)\mapsto\widehat{\boldsymbol\mu}_\ell(\omega,u)$ is $\mathcal G_\ell\otimes\mathcal B(\mathbb R^{q_T})$-measurable, where $\omega$ denotes the underlying sample outcome. Define
\begin{align*}
    R_\ell^2
    &=\sup_{t\in\mathcal I_{\ell,h}}
      \int\|\widehat{\boldsymbol\mu}_\ell(u)-\boldsymbol\mu(u)\|_2^2\,P_t(du),\\
    (S_\ell^\ast)^2
    &=\frac{1}{|\mathcal I_{\ell,h}|}
      \sum_{t\in\mathcal I_{\ell,h}}
      \|\widehat{\boldsymbol\mu}_\ell(\mathbf U_t^\ast)
        -\boldsymbol\mu(\mathbf U_t^\ast)\|_2^2.
\end{align*}
Then $E_{P,T}=\max_{\ell\le L}R_\ell^2$, and each $R_\ell$ is $\mathcal G_\ell$-measurable. The joint measurability of the fitted maps and independence of $\mathbf V_\ell^\ast$ from $\mathcal G_\ell$ give
\begin{align*}
    \E\!\left[(S_\ell^\ast)^2
          \,\middle|\,\mathcal G_\ell\right]
    &=\frac{1}{|\mathcal I_{\ell,h}|}
      \sum_{t\in\mathcal I_{\ell,h}}
      \int\|\widehat{\boldsymbol\mu}_\ell(u)-\boldsymbol\mu(u)\|_2^2\,P_t(du)\\
    &\le R_\ell^2.
\end{align*}
This calculation does not require independence among observations within a validation block.

Fix $K>0$ and $M\ge1$, and let
\[
    \mathcal E_{\ell,K}
    =\{R_\ell^2\le K\rho_T^2\}.
\]
This event is $\mathcal G_\ell$-measurable. If $\rho_T>0$, conditional Markov inequality yields
\[
    \Pr\!\left[
      \mathcal E_{\ell,K}\cap
      \{(S_\ell^\ast)^2>MK\rho_T^2\}
    \right]\le M^{-1}.
\]
If $\rho_T=0$, the conditional second moment vanishes on $\mathcal E_{\ell,K}$, and the probability on the left is zero. On $\{E_{P,T}\le K\rho_T^2\}$ all these events occur. Taking a union bound over folds, and then accounting for coupling failure, gives
\[
    \Pr(E_{V,T}>MK\rho_T^2)
    \le \Pr(E_{P,T}>K\rho_T^2)
       +\frac{L}{M}
       +3L\beta(b_T).
\]
This proves \eqref{eq:buffered-rate-transfer-prob}. If $E_{P,T}=O_p(\rho_T^2)$, first choose a fixed $K$ to control the first term and then a fixed $M$ to control the second. Since $b_T\to\infty$ and $\beta(b_T)\to0$, this also proves \eqref{eq:buffered-rate-transfer-op}.
\end{proof}

\subsection{Proof of the IV nuisance-rate corollary}\label{app:iv-nuisance-rates-proof}
\begin{proof}[Proof of \Cref{cor:iv-nuisance-rates}]
Fix $a\in\{Y,W,Z\}$ and fold $\ell$. Holding the fitted map fixed, the conditional-mean identity \eqref{eq:excess-risk-L2} gives, at every validation time $t$,
\[
    \int\|\widehat{\boldsymbol\mu}^{(-\ell)}_a(u)-\boldsymbol\mu_a(u)\|_2^2\,P_t(du)
    =R_{a,t}(\widehat{\boldsymbol\mu}^{(-\ell)}_a)
     -R_{a,t}(\boldsymbol\mu_a).
\]
Taking the supremum over validation times and a union bound over folds in \eqref{eq:generic-excess-risk} yields
\[
    \Pr\!\left(\|\boldsymbol\Delta_a\|_{L_2}
                   >C_a\rho_{a,T}\right)\le\delta_{a,T}.
\]
Apply \Cref{thm:regression-rate-transfer} with target $\boldsymbol\mu=\boldsymbol\mu_a$ and fold estimates $\widehat{\boldsymbol\mu}_\ell=\widehat{\boldsymbol\mu}^{(-\ell)}_a$. It gives the $O_p(\rho_{a,T})$ rate for the largest foldwise empirical norm. The squared pooled norm $\|\boldsymbol\Delta_a\|_{n,h,2}^2$ is a weighted average of the squared foldwise norms and is therefore bounded by their maximum. Applying this argument to each of the three nuisances proves the first line of \eqref{eq:generic-rates}. The second follows by multiplying the corresponding stochastic bounds; it requires no independence among nuisance estimators.

To construct the event in \Cref{cond:nuisance-rate}, use \eqref{eq:generic-product-condition} to choose a deterministic $M_T\to\infty$ slowly enough that
\[
    M_T\rho_{\max,T}\log T\to0,
    \qquad
    M_T^2\rho_{Z,T}\rho_{Q,T}\sqrt T\to0.
\]
Such a sequence exists because both expressions without $M_T$ converge to zero. Let $C_0=\max\{1,C_Y,C_W,C_Z\}$, which is model-dependent and independent of $T$, and set
\[
    \bar r_T=C_0M_T\rho_{\max,T},
    \qquad
    r_T^2=C_0^2M_T^2\rho_{Z,T}\rho_{Q,T}.
\]
Apply \eqref{eq:buffered-rate-transfer-prob} separately to each nuisance regression with $K=C_0^2$ and $M=M_T^2$, and intersect the resulting events with the three population-error events above. Their joint probability tends to one because $\sum_a\delta_{a,T}=o(1)$, $M_T^{-2}\to0$, and $\beta(b_T)\to0$. On this event, the population and empirical individual errors are bounded by $\bar r_T$, and their required products are bounded by $r_T^2$. The displayed limits imply $\bar r_T\log n_h\to0$ and $r_T=o(T^{-1/4})$, as required by \Cref{cond:nuisance-rate}.
\end{proof}

\begin{remark} \label{supp:rem:const-instrument}
    Specifically, if \(\boldsymbol\mu_Z\) is known to be zero and is not estimated, then \(\boldsymbol\Delta_Z=\mathbf0\) identically. If this conditional mean is an unknown constant, its learner is the foldwise training mean. Under the geometric mixing and moment bounds in \Cref{ass:estimation-primitive}, the instrument coordinates have uniformly summable autocovariances. Since \(|\mathcal J_{\ell,h}|\asymp T\), each coordinate of the foldwise training mean has variance \(O(T^{-1})\). Chebyshev's inequality therefore gives \(\|\boldsymbol\Delta_Z\|=O_p(T^{-1/2})\) uniformly over the fixed folds and instrument coordinates. Hence
\[
    \|\boldsymbol\Delta_Z\|\max\{\|\boldsymbol\Delta_{Y,h}\|,\|\boldsymbol\Delta_W\|\}=O_p(T^{-1/2}\rho_{Q,T}),
\]
which is \(o_p(T^{-1/2})\) whenever \(\rho_{Q,T}=o(1)\). The same product-rate argument applies to other parametric or nonparametric learners for the outcome and treatment nuisances.
\end{remark}

\section{Nuisance learner examples and their proofs}
\label{sec:nuisance-examples}
This section verifies the conditions in \Cref{sec:nuisance-learners} for two nuisance learners using the transfer results proved in \Cref{app:nuisance-rate-transfer}. Each specialization is followed by its proof. Principal-component regression allows the conditional means to depend linearly on a fixed number of leading population principal components of a stationary high-dimensional Gaussian state. Sparse spline regression allows smooth nonlinear dependence on a small, unknown subset of the state coordinates.

\subsection{Stationary principal-component regression}
\label{subsec:stationary-pcr}

Principal-component regression (PCR) reduces a high-dimensional state to a small number of estimated indices before fitting the nuisance regressions \citep{stockwatson2002forecasting}. We consider a stationary Gaussian specialization in which the conditional means lie in the leading population eigenspace. A short population prediction bound then allows direct application of \Cref{thm:regression-rate-transfer}.

Fix the horizon \(h\), write \(\mathbf D_{Y,t}=\mathbf Y_{t+h}\), \(\mathbf D_{W,t}=\mathbf W_t\), and \(\mathbf D_{Z,t}=\mathbf Z_t\), and let \(d_Y=p,d_W=s,d_Z=r\). These response dimensions remain fixed. We subtract population means for notation, although the fitted regressions below include intercepts. Let \(q=q_T\), and let \(\mathbf P_K\) be the orthogonal projector onto the leading \(K\) eigenvectors of \(\boldsymbol\Sigma_U=\E(\mathbf U_t\mathbf U_t^\top)\), where \(1\le K<q\) is fixed and known.

\begin{example}[Stationary Gaussian principal-component nuisance means]
\label{ex:stationary-pcr}
Suppose the following conditions hold. All positive constants below are model-dependent and independent of \(T,q\), and the fold.
\begin{enumerate}
    \item[(i)] The full observed process \(\{(\mathbf U_t,\mathbf D_{Y,t},\mathbf D_{W,t},\mathbf D_{Z,t})\}_{t\in\mathbb Z}\) is centered, jointly Gaussian, and strictly stationary. For every \(a\in\{Y,W,Z\}\) and \(j\le d_a\),
    \[
        \mu_{a,j}(u)=\boldsymbol\theta_{a,j}^\top u,
        \qquad \mathbf P_K\boldsymbol\theta_{a,j}=\boldsymbol\theta_{a,j},
        \qquad \E D_{a,j,t}^2\le C_D.
    \]
    \item[(ii)] The ordered eigenvalues of \(\boldsymbol\Sigma_U\) satisfy
    \[
        c_\lambda q\le\lambda_K\le\lambda_1\le C_\lambda q,
        \qquad \lambda_K-\lambda_{K+1}\ge c_g q.
    \]
    \item[(iii)] Write \(\mathbf e_{a,t}=\mathbf D_{a,t}-\boldsymbol\mu_a(\mathbf U_t)\) and \(\mathbf V_t=(\mathbf U_t^\top,\mathbf e_{Y,t}^\top,\mathbf e_{W,t}^\top,\mathbf e_{Z,t}^\top)^\top\). For every pair of coordinates \(i,j\) and every \(k\in\mathbb Z\),
    \[
        |\cov(V_{i,0},V_{j,k})|\le C_0 e^{-c_0|k|}.
    \]
\end{enumerate}
\end{example}

The coefficients can be dense in the observed coordinates. For instance, the linear simulation state has covariance \((\boldsymbol\Lambda\boldsymbol\Lambda^\top+\sigma_E^2\mathbf I_q)/(K+\sigma_E^2)\), where \(\boldsymbol\Lambda^\top\boldsymbol\Lambda=q\mathbf I_K\). Its leading eigenspace is \(\operatorname{col}(\boldsymbol\Lambda)\), its eigengap is \(q/(K+\sigma_E^2)\), and its nuisance means depend on observed indices in this space. Thus repeated leading eigenvalues are permitted, and the target remains the conditional mean given \(\mathbf U_t\).

Use deterministic fold boundaries satisfying \Cref{ass:cross-fitting}. In fold \(\ell\), let \(m_{\ell,T}=|\mathcal J_{\ell,h}|\asymp T\), and denote training means by \(\overline{\mathbf U}_\ell\) and \(\overline D_{a,j,\ell}\). Let \(\widehat{\mathbf V}_\ell\in\mathbb R^{q\times K}\) contain the leading eigenvectors of \(m_{\ell,T}^{-1}\sum_{t\in\mathcal J_{\ell,h}}(\mathbf U_t-\overline{\mathbf U}_\ell)(\mathbf U_t-\overline{\mathbf U}_\ell)^\top\), with their eigenvalues in the diagonal matrix \(\widehat{\mathbf L}_\ell\). The PCR fit is
\begin{equation}\label{eq:pcr-fit}
\begin{split}
    \widehat{\boldsymbol\theta}_{a,j,\ell}
    &=\widehat{\mathbf V}_\ell\widehat{\mathbf L}_\ell^{-1}\widehat{\mathbf V}_\ell^\top
      \frac1{m_{\ell,T}}\sum_{t\in\mathcal J_{\ell,h}}
      (\mathbf U_t-\overline{\mathbf U}_\ell)(D_{a,j,t}-\overline D_{a,j,\ell}),\\
    \widehat\mu^{(-\ell)}_{a,j}(u)
    &=\overline D_{a,j,\ell}
      +\widehat{\boldsymbol\theta}_{a,j,\ell}^\top(u-\overline{\mathbf U}_\ell).
\end{split}
\end{equation}
If \(\widehat{\mathbf L}_\ell\) is singular, set the fitted slope to zero. This event has probability tending to zero under the conditions below. For use in the next subsection, also denote by \(\mathcal K_{\ell,h}\subseteq\mathcal J_{\ell,h}\) the longer of its at most two contiguous training stretches, which still has size of order \(T\).

\begin{corollary}[Stationary PCR specialization]\label{cor:stationary-pcr-rate}
Suppose \Cref{ex:stationary-pcr} holds. Let \(\beta_T(k)\) denote the beta-mixing coefficients of the full observed process, allowing its law to depend on \(T\), and suppose \(\beta_T(b_T)\to0\). Then the fits in \eqref{eq:pcr-fit} satisfy
\begin{equation}\label{eq:pcr-nuisance-rate}
    \|\boldsymbol\Delta_a\|_{L_2}\vee\|\boldsymbol\Delta_a\|_{n,h,2}
    =O_p(T^{-1/2}),\qquad a\in\{Y,W,Z\},
\end{equation}
and verify \Cref{cond:nuisance-rate,cond:sg-nuisance-error}. No restriction on \(q/T\) is needed for the population prediction bound. The full-process mixing requirement is separate from the coordinate covariance bounds in \Cref{ex:stationary-pcr}(iii).
\end{corollary}

\begin{proof}
Fix a fold and a scalar response, suppress their indices, and write \(J=\mathcal J_{\ell,h}\), \(m=|J|\), \(\boldsymbol\Sigma=\boldsymbol\Sigma_U\), and \(\varepsilon_t=D_t-\boldsymbol\theta^\top\mathbf U_t\). Throughout, \(C\) denotes a finite model-dependent constant, independent of \(m,q\), which may change between displays. For any square-integrable scalar stationary process \(A_t\) with summable autocovariances, counting at most \(m\) ordered pairs at each lag gives
\[
    \var\left(\frac1m\sum_{t\in J}A_t\right)
    \le\frac1m\sum_{k\in\mathbb Z}|\cov(A_0,A_k)|.
\]
This bound also applies to the two separated training stretches. The Gaussian fourth-moment identity of \citet{isserlis1918formula} and \Cref{ex:stationary-pcr}(iii) give \(|\cov(V_{i,0}V_{j,0},V_{i,k}V_{j,k})|\le2C_0^2e^{-2c_0|k|}\). Since \(\E(\mathbf U_t\varepsilon_t)=\mathbf0\), defining \(\mathbf G=m^{-1}\sum_{t\in J}\mathbf U_t\mathbf U_t^\top\) and \(\mathbf g=m^{-1}\sum_{t\in J}\mathbf U_t\varepsilon_t\) gives
\begin{equation}\label{eq:pcr-moment-bounds}
    \E\|\mathbf G-\boldsymbol\Sigma\|_F^2\le\frac{Cq^2}{m},\quad
    \E\|\mathbf g\|_2^2\le\frac{Cq}{m},\quad
    \E\|\overline{\mathbf U}\|_2^2\le\frac{Cq}{m},\quad
    \E\overline\varepsilon^2\le\frac C m.
\end{equation}
Consequently, Markov's inequality and training centering yield
\[
    \|\mathbf G_c-\boldsymbol\Sigma\|_{\op}=O_p(qm^{-1/2}),\qquad
    \|\mathbf g_c\|_2=O_p(\sqrt{q/m}),
\]
where \(\mathbf G_c=\mathbf G-\overline{\mathbf U}\overline{\mathbf U}^\top\) and \(\mathbf g_c=\mathbf g-\overline{\mathbf U}\,\overline\varepsilon\).

Consider the event \(\mathcal A_m=\{\|\mathbf G_c-\boldsymbol\Sigma\|_{\op}\le q\min(c_\lambda,c_g)/2\}\), whose probability tends to one. Weyl's inequality gives \(\widehat\lambda_K\ge c_\lambda q/2\). Put \(\mathbf Q=\mathbf I_q-\mathbf P_K\). Each leading sample eigenvector \(\widehat v_j\) satisfies
\[
    (\widehat\lambda_j\mathbf I_q-\boldsymbol\Sigma)\mathbf Q\widehat v_j
    =\mathbf Q(\mathbf G_c-\boldsymbol\Sigma)\widehat v_j.
\]
On \(\operatorname{col}(\mathbf Q)\), the matrix on the left has eigenvalues at least \(c_gq/2\). Hence, with \(\widehat{\mathbf P}=\widehat{\mathbf V}\widehat{\mathbf V}^\top\),
\[
    \|\widehat{\mathbf P}-\mathbf P_K\|_{\op}
    \le\|\widehat{\mathbf P}-\mathbf P_K\|_F
    =\sqrt2\|\mathbf Q\widehat{\mathbf V}\|_F
    \le\frac{2\sqrt{2K}}{c_gq}\|\mathbf G_c-\boldsymbol\Sigma\|_{\op}
    =O_p(m^{-1/2}).
\]
This is the leading-eigenspace case of the Davis--Kahan bound in \citet[Theorem~2]{yu2015useful}. The signal bound gives \(\|\boldsymbol\theta\|_2^2\le C_D/(c_\lambda q)\). On \(\mathcal A_m\), the centered normal equations give the exact identity
\[
    \widehat{\boldsymbol\theta}-\boldsymbol\theta
    =(\widehat{\mathbf P}-\mathbf P_K)\boldsymbol\theta
      +\widehat{\mathbf V}\widehat{\mathbf L}^{-1}\widehat{\mathbf V}^\top\mathbf g_c,
    \qquad
    \|\widehat{\boldsymbol\theta}-\boldsymbol\theta\|_2=O_p((mq)^{-1/2}).
\]
The fitted intercept is \(\widehat b=\overline\varepsilon-(\widehat{\boldsymbol\theta}-\boldsymbol\theta)^\top\overline{\mathbf U}=O_p(m^{-1/2})\). Holding the fit fixed and integrating an independent state therefore gives
\begin{equation}\label{eq:pcr-population-risk}
    \E_U\{\widehat\mu(U)-\mu(U)\}^2
    =\widehat b^2+(\widehat{\boldsymbol\theta}-\boldsymbol\theta)^\top
      \boldsymbol\Sigma(\widehat{\boldsymbol\theta}-\boldsymbol\theta)
    =O_p(m^{-1}).
\end{equation}
Taking maxima over the fixed folds and summing over the fixed response coordinates proves the population part of \eqref{eq:pcr-nuisance-rate}. Applying \eqref{eq:buffered-rate-transfer-prob} with \(\beta_T(b_T)\to0\) gives its validation part.

For a fresh Gaussian state, every scalar fitted error is affine Gaussian. Its \(\psi_2\) norm is at most a universal constant times
\[
    |\widehat b|+
    \{(\widehat{\boldsymbol\theta}-\boldsymbol\theta)^\top
      \boldsymbol\Sigma(\widehat{\boldsymbol\theta}-\boldsymbol\theta)\}^{1/2}.
\]
Applying this bound in every unit response direction bounds the vector error's \(\psi_2\) norm by a fixed constant times its population \(L_2\) norm. By \eqref{eq:pcr-nuisance-rate}, the event that both nuisance norms are at most \(T^{-1/2}\log T\) for all three regressions has probability tending to one. On this event, taking \(\bar r_T=r_T=T^{-1/2}\log T\) verifies \Cref{cond:nuisance-rate}, since \(\bar r_T\log T\to0\) and \(r_T=o(T^{-1/4})\). The Gaussian bound supplies a fixed model-dependent envelope \(\bar K\) on the same event for all sufficiently large \(T\), verifying \Cref{cond:sg-nuisance-error}.
\end{proof}

This specialization uses a fixed known rank and linear conditional means. Consistently selecting the rank would require a separate verification for the training samples, while quadratic PCR would require a separate prediction and tail analysis.

\subsection{Sparse spline nuisance model}
\label{subsec:sparse-spline-nuisance}
\label{supp:subsec:sparse-spline-proof}

We next allow the nuisance means to be nonlinear functions of a small, unknown subset of the state coordinates. The learner selects the coordinates and fits a spline regression on the same training stretch. We use the population prediction bound of \citet{barrera2021generalization} for least squares with dependent observations. Its application here requires stationarity, bounded responses, and smoothness of the nuisance means on their active coordinates.

For \(\nu>0\), set \(b_\nu=\lceil\nu\rceil-1\) and \(\gamma_\nu=\nu-b_\nu\in(0,1]\). Let \(\mathcal H_d^\nu(A_0)\) be the H\"older ball of functions on \([0,1]^d\) whose partial derivatives up to order \(b_\nu\) are continuous, extend to the boundary, and satisfy
\[
    \max_{|\alpha|_1\le b_\nu}\|\partial^\alpha f\|_\infty
    +\max_{|\alpha|_1=b_\nu}\sup_{x\ne y}
      \frac{|\partial^\alpha f(x)-\partial^\alpha f(y)|}
           {\|x-y\|_\infty^{\gamma_\nu}}
    \le A_0.
\]
This convention uses Lipschitz derivatives of order \(\nu-1\) when \(\nu\) is an integer.

\begin{example}[Nuisance means with sparse coordinate dependence]
\label{ex:sparse-spline-nuisance}
Fix the horizon \(h\). Suppose that the following conditions hold, with fixed \(d\ge1\), \(\nu>0\), and model-dependent constants \(A_0,B_0,C_B,C_\beta,c_\beta>0\) independent of \(T\) and \(q=q_T\). The dimensions \(p,s,r\) remain fixed.
\begin{enumerate}
    \item[(i)] The joint process \(\{(\mathbf U_t,\mathbf Y_{t+h},\mathbf W_t,\mathbf Z_t)\}_{t\in\mathbb Z}\) is strictly stationary and satisfies \(\beta(k)\le C_\beta e^{-c_\beta k}\). Almost surely,
    \[
        \mathbf U_t\in[0,1]^q,
        \qquad
        \max\{\|\mathbf Y_{t+h}\|_\infty,
                  \|\mathbf W_t\|_\infty,
                  \|\mathbf Z_t\|_\infty\}\le B_0.
    \]
    \item[(ii)] For all sufficiently large \(T\), \(q_T\ge d\) and there is an unknown deterministic set \(S_T\subseteq\{1,\ldots,q\}\) of cardinality \(d\) and maps \(\mathbf f_g,\mathbf f_W,\mathbf f_Z\) on \([0,1]^d\) such that, almost surely,
    \[
        \mathbf g_h(\mathbf U_t)=\mathbf f_g(\mathbf U_{t,S_T}),
        \qquad
        \boldsymbol\mu_W(\mathbf U_t)=\mathbf f_W(\mathbf U_{t,S_T}),
        \qquad
        \boldsymbol\mu_Z(\mathbf U_t)=\mathbf f_Z(\mathbf U_{t,S_T}).
    \]
    Every scalar coordinate of these maps belongs to \(\mathcal H_d^\nu(A_0)\), and \(\|\mathbf B_h\|_{\op}\le C_B\). The support and the maps may vary with \(T\), subject to the same bounds.
\end{enumerate}
\end{example}

The common support may contain inactive coordinates, so \(d\) is an upper bound on the number of coordinates needed jointly by the primitive nuisances. Their outcome counterpart satisfies
\begin{equation}
    \boldsymbol\mu_{Y,h}(\mathbf U_t)
    =\mathbf f_Y(\mathbf U_{t,S_T}),
    \qquad
    \mathbf f_Y=\mathbf B_h\mathbf f_W+\mathbf f_g.
    \label{eq:spline-outcome-closure}
\end{equation}
Every coordinate of \(\mathbf f_Y\) belongs to \(\mathcal H_d^\nu\{(1+\sqrt{s}C_B)A_0\}\). Thus the outcome regression has the same support bound and smoothness. These restrictions allow interactions among the active coordinates and \(q\gg T\).

The learner uses the bounds \(d,\nu,B_0\), while \(S_T\) is unknown. For a fixed integer degree \(\kappa_{\mathrm{sp}}\ge\lceil\nu\rceil\), let \(\mathcal S_{J,\kappa_{\mathrm{sp}}}\) be the univariate spline space on \([0,1]\) with \(J\) equal knot intervals, degree \(\kappa_{\mathrm{sp}}\), and \(\kappa_{\mathrm{sp}}-1\) continuous derivatives at the interior knots. Its dimension is \(J+\kappa_{\mathrm{sp}}\). Define
\begin{equation}
    \mathcal F_{q,J}
    =\bigcup_{\substack{S\subseteq\{1,\ldots,q\}\\|S|=d}}
       \left\{u\mapsto f(u_S):
           f\in\mathcal S_{J,\kappa_{\mathrm{sp}}}^{\otimes d}\right\}.
    \label{eq:spline-sieve}
\end{equation}
Each candidate support therefore has \(K_J=(J+\kappa_{\mathrm{sp}})^d\) spline coefficients.

Use deterministic fold boundaries satisfying \Cref{ass:cross-fitting}, and take the longer contiguous training stretch \(\mathcal K_{\ell,h}\) from \Cref{subsec:stationary-pcr}, with \(m_{\ell,T}=|\mathcal K_{\ell,h}|\asymp T\). For each response coordinate \(D_{a,j,t}\), \(a\in\{Y,W,Z\}\), fit
\begin{equation}
    \widehat f_{a,j,\ell}
    \in\argmin_{f\in\mathcal F_{q,J_{\ell,T}}}
       \frac{1}{m_{\ell,T}}
       \sum_{t\in\mathcal K_{\ell,h}}
           \{D_{a,j,t}-f(\mathbf U_t)\}^2,
    \qquad
    \widehat\mu^{(-\ell)}_{a,j}(u)
    =\operatorname{clip}_{B_0}\{\widehat f_{a,j,\ell}(u)\},
    \label{eq:spline-ls}
\end{equation}
where \(\operatorname{clip}_{B_0}(x)=\max\{-B_0,\min(B_0,x)\}\). Within each support, choose the minimum Euclidean norm coefficient solution in a fixed B-spline basis, and break ties between supports by a fixed ordering. This defines a measurable least-squares minimizer even when a design matrix is singular. Stack the clipped coordinates to obtain \(\widehat{\boldsymbol\mu}^{(-\ell)}_a\).

\begin{corollary}[Sparse spline prediction rate]
\label{cor:sparse-spline-rate}
Suppose \Cref{ex:sparse-spline-nuisance} holds. In \eqref{eq:spline-ls}, take
\[
    J_{\ell,T}
    =\left\lceil
        \left\{\frac{m_{\ell,T}}{\log^2(2m_{\ell,T})}\right\}^{1/(2\nu+d)}
      \right\rceil,
    \qquad
    \rho_{\mathrm{SP},T}
    =\left\{\frac{\log^2(2T)}{T}\right\}^{\nu/(2\nu+d)}
      +\sqrt{\frac{d\log(2q_T)\log^2(2T)}{T}}.
\]
Then the population nuisance errors satisfy
\begin{equation}
    \max_{a\in\{Y,W,Z\}}\|\boldsymbol\Delta_a\|_{L_2}
    =O_p(\rho_{\mathrm{SP},T}).
    \label{eq:spline-population-rate}
\end{equation}
The constants in this rate may depend on the fixed model bounds, \(d,\nu,\kappa_{\mathrm{sp}}\), the response dimensions, and the fixed fold geometry, but not on \(T\) or \(q_T\).
\end{corollary}

Clipping also verifies the tail requirement. Part (i) of \Cref{ex:sparse-spline-nuisance} implies \(|\mu_{a,j}(\mathbf U_t)|\le B_0\) almost surely. Hence every fitted coordinate error is at most \(2B_0\) in absolute value, for every training realization. With \(P_U\) denoting the stationary state law,
\begin{equation}
    \max_{\ell\le L}\max_{a\in\{Y,W,Z\}}
    \|\boldsymbol\delta_{a,\ell}\|_{\psi_2,P_U}
    \le 2B_0\sqrt{\frac{\max\{p,s,r\}}{\log2}}
    =:\bar K.
    \label{eq:spline-tail-bound}
\end{equation}
This bound gives \Cref{cond:sg-nuisance-error} with a constant independent of the ambient dimension.

We now apply \Cref{thm:regression-rate-transfer} to \eqref{eq:spline-population-rate}. The resulting validation errors satisfy
\begin{equation}
    \max_{a\in\{Y,W,Z\}}
    \left\{\|\boldsymbol\Delta_a\|_{L_2}
           \vee\|\boldsymbol\Delta_a\|_{n,h,2}\right\}
    =O_p(\rho_{\mathrm{SP},T}).
    \label{eq:spline-validation-rate}
\end{equation}
In particular, if
\begin{equation}
    \nu>d/2,
    \qquad
    d\log(2q_T)\log^2(2T)=o(\sqrt T),
    \label{eq:spline-inference-growth}
\end{equation}
then \(\rho_{\mathrm{SP},T}=o(T^{-1/4})\), so the individual and product requirements of \Cref{cond:nuisance-rate} hold in both norms. Together with \eqref{eq:spline-tail-bound}, this verifies the nuisance conditions used by the inference results, subject to their remaining assumptions. The proof below constructs the common event whose probability tends to one.

The ambient dimension enters the rate through \(\log(2q_T)\), because the active coordinates are selected from the full state. For fixed \(d\) and \(\nu>d/2\), \eqref{eq:spline-inference-growth} permits every polynomial sequence \(q_T=T^A\) with fixed \(A>0\). The additional boundedness assumption concerns the observed responses themselves; it is used in the imported prediction theorem. Dependence between successive regression errors is allowed under the joint mixing assumption.

\begin{proof}[Proof of \Cref{cor:sparse-spline-rate} and its application]
Fix a fold \(\ell\), a response coordinate \(D_{a,j,t}\), and the horizon \(h\). Write \(m=m_{\ell,T}\), \(J=J_{\ell,T}\), and \(\mu=\mu_{a,j}\). The training indices form one deterministic contiguous stretch. Its observations have the stationary marginal law and geometric mixing bound in \Cref{ex:sparse-spline-nuisance}, uniformly over folds. All constants denoted by \(C\) below are finite, may change between displays, and depend only on the fixed bounds listed in \Cref{cor:sparse-spline-rate}.

We first bound the complexity of \(\mathcal F_{q,J}\). For one support, this class is a linear space of dimension \(K_J=(J+\kappa_{\mathrm{sp}})^d\), so its VC-subgraph dimension is at most \(v_J=K_J+1\); see \citet[Example 3.7]{barrera2021generalization}. There are \(N_q=\binom qd\) candidate supports. By the Sauer bound, their union induces at most
\[
    N_q\left(\frac{eM}{v_J}\right)^{v_J}
\]
distinct subgraph traces on \(M\ge v_J\) points. If \(M\ge16v_J\) and \(M\ge4\log_2N_q\), the logarithm of this expression is at most
\[
    \frac{M\log2}{4}+\frac{M\log(16e)}{16}
    <M\log2.
\]
Here we used that \(x\mapsto\log(ex)/x\) decreases for \(x\ge1\). Such a set of points cannot be shattered. Consequently, for a universal constant \(C_{\mathrm{vc}}\), the VC-subgraph dimension \(V_{q,J}\) satisfies
\begin{equation}
    V_{q,J}
    \le C_{\mathrm{vc}}\{K_J+1+\log N_q\}
    \le C_{\mathrm{vc}}\{K_J+1+d\log(2q)\}.
    \label{eq:spline-vc-bound}
\end{equation}
For each support, rational coefficient vectors give a countable pointwise dense subclass. Taking their finite union proves pointwise measurability. The minimum norm least-squares coefficients are Borel functions of the data through the Moore--Penrose inverse, and the fixed tie rule preserves measurability. Thus the learner meets the measurability requirements of the imported theorem and of \Cref{thm:regression-rate-transfer}. The spline functions can be extended by zero outside \([0,1]^q\) when applying the imported theorem on \(\mathbb R^q\).

Theorem 3.10 and Remark 3.11 of \citet{barrera2021generalization} apply to ordinary least squares followed by clipping, as in \eqref{eq:spline-ls}. To specify their parameters, take their \(c=4\), \(\lambda=2\), and dependence lag
\[
    k_m=\left\lceil\frac{3\log m}{c_\beta}\right\rceil.
\]
Their admissibility condition \((3.8)\) holds for all sufficiently large \(m\), since \(c^2-71<0\) and \(\lfloor m/k_m\rfloor\ge1\). The statistical term in their inequality is at most
\[
    C B_0^2(V_{q,J}+1)\frac{k_m\log m}{m},
\]
whereas its mixing remainder is at most
\[
    C B_0^2 m\beta(k_m)
    \le C B_0^2 m^{-2}.
\]
Their mixing coefficients are bounded by the joint process coefficients assumed here. Stationarity makes their average of marginal prediction integrals equal to integration against \(P_U\). Absorbing the mixing remainder gives, for all sufficiently large \(m\),
\begin{equation}
\begin{split}
    \E\left[
       \|\widehat\mu^{(-\ell)}_{a,j}-\mu\|_{L_2(P_U)}^2
       \right]
    \le{}&2\inf_{f\in\mathcal F_{q,J}}
               \|f-\mu\|_{L_2(P_U)}^2\\
       &+C B_0^2(V_{q,J}+1)\frac{\log^2(2m)}{m}.
\end{split}
    \label{eq:spline-population-oracle}
\end{equation}
The expectation is over the training sample; the inner norm integrates a fresh state with the fitted function held fixed. The constant in this inequality is uniform in the sieve dimension and the ambient dimension.

It remains to bound the approximation error. By \eqref{eq:spline-outcome-closure}, for every \(a\in\{Y,W,Z\}\) the target has a representation \(\mu(\mathbf U_t)=f_{a,j}(\mathbf U_{t,S_T})\), where \(f_{a,j}\) belongs to a H\"older ball with radius at most \((1+\sqrt{s}C_B)A_0\). Tensor-product B-spline approximation on the regular knot grid gives
\[
    \inf_{v\in\mathcal S_{J,\kappa_{\mathrm{sp}}}^{\otimes d}}
       \|v-f_{a,j}\|_\infty
    \le C J^{-\nu};
\]
see \citet[Lemma 2(a) and Remark 7]{shen2015adaptive}. The constant depends on the fixed H\"older radius, dimension, smoothness, and spline degree. Since the true support is among the candidates in \eqref{eq:spline-sieve}, this uniform approximation also bounds the population error under any state law in the model. Equations \eqref{eq:spline-vc-bound} and \eqref{eq:spline-population-oracle} therefore yield
\begin{equation}
    \E\left[
       \|\widehat\mu^{(-\ell)}_{a,j}-\mu_{a,j}\|_{L_2(P_U)}^2
       \right]
    \le C\left[
        J^{-2\nu}
        +\{(J+\kappa_{\mathrm{sp}})^d+d\log(2q)\}
           \frac{\log^2(2m)}{m}
        \right].
    \label{eq:spline-bias-variance}
\end{equation}
The choice of \(J\) in the corollary balances the approximation term and the term from the spline coefficients. Using \(m\asymp T\) gives a bound \(C\rho_{\mathrm{SP},T}^2\) on the right. Summing over the fixed response coordinates and folds and applying Markov's inequality proves \eqref{eq:spline-population-rate}.

For the tail bound, let \(d_a\) be the dimension of response \(a\). Bounded responses imply bounded conditional means, and clipping bounds the fitted coordinates for every training sample. Thus, for every unit vector \(v\in\mathbb R^{d_a}\),
\[
    |v^\top\boldsymbol\delta_{a,\ell}(\mathbf U_t)|
    \le\|\boldsymbol\delta_{a,\ell}(\mathbf U_t)\|_2
    \le2B_0\sqrt{d_a}
\]
almost surely under \(P_U\). The definition of the sub-Gaussian norm gives \eqref{eq:spline-tail-bound}, including conditional on the training data with a fresh state drawn from \(P_U\).

Each fitted map uses only observations in \(\mathcal K_{\ell,h}\subseteq\mathcal J_{\ell,h}\). Hence \Cref{thm:regression-rate-transfer} applies to its population \(O_p\) rate and proves \eqref{eq:spline-validation-rate}. Under \eqref{eq:spline-inference-growth}, both terms in \(\rho_{\mathrm{SP},T}\) are \(o(T^{-1/4})\). Choose a deterministic sequence \(M_T\to\infty\) sufficiently slowly that
\[
    M_T\rho_{\mathrm{SP},T}=o(T^{-1/4}).
\]
By \eqref{eq:spline-validation-rate}, the event
\[
    \mathcal E_{T,h}
    =\left\{
       \max_{a\in\{Y,W,Z\}}
       \left(\|\boldsymbol\Delta_a\|_{L_2}
             \vee\|\boldsymbol\Delta_a\|_{n,h,2}\right)
       \le M_T\rho_{\mathrm{SP},T}
      \right\}
\]
has probability tending to one. Set \(\bar r_T=r_T=M_T\rho_{\mathrm{SP},T}\). Then \(\bar r_T\log n_h\to0\), \(r_T=o(T^{-1/4})\), and on this event the required product of nuisance errors is at most \(r_T^2\) in both norms. This verifies \Cref{cond:nuisance-rate}; the deterministic bound \eqref{eq:spline-tail-bound} verifies \Cref{cond:sg-nuisance-error} on the same event.
\end{proof}

\section{Simulation designs and additional results}\label{app:simulation-details}

\subsection{Data-generating process and estimation} \label{app:subsec:dgp}

\subsubsection{Dense nuisance functions and predetermined feedback}\label{app:sim-dense-feedback}

The observed state is an approximate factor panel. For column vectors \(U_t\in\mathbb R^q\), \(F_t\in\mathbb R^r\), and \(E_t\in\mathbb R^q\), let
\begin{align}
 U_t&=\frac{\Lambda F_t+\sigma_E E_t}{\sqrt{r+\sigma_E^2}},
 &F_t&=\rho_F F_{t-1}+\sqrt{1-\rho_F^2}\,\eta_t^F,\label{eq:supp-sim-state}\\
 E_t&=\rho_E E_{t-1}+\sqrt{1-\rho_E^2}\,\eta_t^E,
 &C_t&\stackrel{\mathrm{iid}}{\sim}N(0,1),\label{eq:supp-sim-confounder}\\
 \zeta_t&=\rho_Z\zeta_{t-1}+\sqrt{1-\rho_Z^2}
 \{aC_{t-1}+\sqrt{1-a^2}\,\xi_t\}.\label{eq:supp-sim-residual-instrument}
\end{align}
All primitive innovation streams are mutually independent standard Gaussian sequences. Factor and idiosyncratic coordinates start in their stationary distributions. The loading matrix contains the first \(r\) columns of the real Fourier dictionary
\[
 1,\quad \sqrt2\cos(2\pi j/q),\quad\sqrt2\sin(2\pi j/q),\quad
 \sqrt2\cos(4\pi j/q),\quad\sqrt2\sin(4\pi j/q),\quad\ldots,
\]
where \(j=0,\ldots,q-1\). Thus \(\Lambda^{\mathsf T}\Lambda/q=I_r\), and its entries are bounded by the universal constant \(\sqrt2\). Define the observed indices
\[
 S_t=B^{\mathsf T}U_t,\qquad
 B=\frac{\sqrt{r+\sigma_E^2}}{q\sqrt{1+\sigma_E^2/q}}\Lambda.
\]
They satisfy \(\operatorname{Var}(S_t)=I_r\). The feasible estimators observe \(U_t\); the indices \(S_t\) and matrix \(B\) are used only to generate the data.

Let \(A_r=(d,g,\gamma)\in\mathbb R^{r\times3}\). To hold the joint nuisance signal covariance constant when changing \(r\) or \(q\), set \(A_r=H_rR\), where
\[
 R^{\mathsf T}R=
 \begin{pmatrix}
 .54&.1025&.06\\ .1025&.4975&.35\\ .06&.35&.6775
 \end{pmatrix}.
\]
The rows of the embedding matrix are
\[
 (H_r)_{j,\cdot}=\frac1{\sqrt r}
 \left(1,\sqrt2\cos\frac{\pi(j-1/2)}r,
 \sqrt2\cos\frac{2\pi(j-1/2)}r\right),
\]
for \(j=1,\ldots,r\), and \(R\) is the upper triangular Cholesky factor with positive diagonal. The coefficients \(Bd,Bg,B\gamma\) are dense vectors in the observed coordinates. Outcomes, treatment, and instrument satisfy
\begin{align}
 Z_t&=g^{\mathsf T}S_t+\zeta_t,\label{eq:supp-sim-instrument}\\
 W_t&=d^{\mathsf T}S_t+\pi\zeta_t+0.8C_t+0.7\nu_t,\label{eq:supp-sim-treatment}\\
 Y_{t+h}^{(h)}&=\beta_hW_t+\gamma^{\mathsf T}S_t+0.8C_t
 +\frac{0.8}{\sqrt{h+1}}\sum_{j=0}^h e_{t+j}+\tau\zeta_t,
 \qquad \beta_h=0.8e^{-h/8}.\label{eq:supp-sim-outcome}
\end{align}
The main design sets \(T=500\), \(q=1000\), \(r=6\), \(h=4\), \(\rho_F=.8\), \(\rho_E=.2\), \(\sigma_E=1\), \(\rho_Z=.8\), \(a=.9\), \(\pi=.6\), and \(\tau=0\). These model parameters are fixed within each setting. We discard 200 initial observations and retain \(T\) horizon-aligned observations, generating the additional innovations needed for the moving average. The moving-average normalization keeps its marginal variance fixed across horizons.

The white-noise restriction on \(C_t\) is essential to this feedback design. It makes \(\zeta_t\) independent of \(C_t\) and of the outcome innovation sequence, so conditional exclusion holds exactly when \(\tau=0\). In particular,
\[
 \mathbb E(Y_{t+h}^{(h)}\mid U_t)=\beta_h d^{\mathsf T}S_t+\gamma^{\mathsf T}S_t,
 \qquad \mathbb E(W_t\mid U_t)=d^{\mathsf T}S_t,
 \qquad \mathbb E(Z_t\mid U_t)=g^{\mathsf T}S_t.
\]
For \(k\ge1\), however,
\[
 \operatorname{Cov}(\zeta_{t+k},C_t)
 =a\sqrt{1-\rho_Z^2}\,\rho_Z^{k-1},\qquad
 \operatorname{Cov}(\zeta_t,C_{t+k})=0.
\]
Consequently a valid current instrument can depend on past confounding shocks. The settings \(a=0\) and \(a=.5\) vary this feedback while preserving the marginal instrument variance and persistence. Serial dependence remains in the score through the persistent instrument and overlapping outcome innovations even at \(a=0\).

\subsubsection{Nuisance fitting and interval construction}\label{app:sim-nuisance-fitting}

The proposed estimator divides the observations into ten contiguous validation folds, using all remaining observations except those within
\[
 b_T=\max\{h+1,\lceil\log T\rceil\}
\]
indices of a validation observation. The main design gives \(b_T=7\). Every observation appears in exactly one validation fold. On each training sample of size \(m\), we center \(U_t\), compute its principal components, and select the rank from \(\{0,\ldots,10\}\) using
\[
 \operatorname{IC}_{p2}(k)=\log V(k)
 +k\frac{m+q}{mq}\log\{\min(m,q)\},
\]
where \(V(k)\) is the mean squared rank-\(k\) reconstruction error over the \(mq\) entries \citep{baing2002determining}. We then regress each nuisance response on the selected scores and an intercept by OLS. All steps use training observations only. This is principal-component regression as in \citet{stockwatson2002forecasting}.

Write the held-out residuals as \(\widehat y_t,\widehat w_t,\widehat z_t\), and let
\[
 \widehat A=T^{-1}\sum_t\widehat z_t\widehat w_t,
 \qquad \widehat\beta=\frac{\sum_t\widehat z_t\widehat y_t}{\sum_t\widehat z_t\widehat w_t},
 \qquad \widehat d_t=\widehat y_t-\widehat\beta\widehat w_t,
 \qquad \widehat\psi_t=\widehat z_t\widehat d_t.
\]
We adjust the score for the sensitivity of the foldwise OLS nuisance fits before applying Bartlett HAC. For fold \(j\), denote its training and validation sets by \(\mathcal T_j\) and \(\mathcal V_j\), and let \(x_{jt}\) contain its intercept and selected factor scores. Set
\[
 Q_j=\sum_{s\in\mathcal T_j}x_{js}x_{js}^{\mathsf T},\qquad
 v_{jz}=Q_j^{-1}\sum_{s\in\mathcal V_j}x_{js}\widehat z_s,
 \qquad
 v_{jd}=Q_j^{-1}\sum_{s\in\mathcal V_j}x_{js}\widehat d_s.
\]
Let \(\widehat e_{a,jt}=a_t-x_{jt}^{\mathsf T}\widehat\theta_{a,j}\) be the training residual for response \(a\in\{Y,W,Z\}\). We apply Bartlett HAC to
\begin{equation}
 \widehat\phi_t=\widehat\psi_t-
 \sum_{j:t\in\mathcal T_j}
 \left\{x_{jt}^{\mathsf T}v_{jz}
 (\widehat e_{Y,jt}-\widehat\beta\widehat e_{W,jt})
 +x_{jt}^{\mathsf T}v_{jd}\widehat e_{Z,jt}\right\}.
 \label{eq:supp-sim-corrected-score}
\end{equation}
All inverses above are evaluated only for nonsingular training designs. The dimensions of \(x_{jt},v_{jz},v_{jd}\) equal the selected rank plus one. Holding the PCA spaces, ranks, and folds fixed, \(\widehat\phi_t/(T\widehat A)\) is exactly the derivative of the IV estimate with respect to observation \(t\)'s case weight when that weight enters both the training OLS fits and the final IV moment. This is a stacked estimating-equation calculation, analogous to the two-step sensitivity calculation in \citet[Section~6]{neweymcfadden1994large}.

For a scalar series \(v_t\), define
\[
 \widehat\Omega_L(v)=\widehat\gamma_0(v)
 +2\sum_{k=1}^{L}\left(1-\frac{k}{L+1}\right)\widehat\gamma_k(v),
 \qquad
 \widehat\gamma_k(v)=T^{-1}\sum_{t=k+1}^T(v_t-\bar v)(v_{t-k}-\bar v).
\]
The proposed standard error and interval are
\[
 \widehat{\mathrm{SE}}_{\mathrm{corr}}
 =\frac{\{\widehat\Omega_{L_T}(\widehat\phi)/T\}^{1/2}}{|\widehat A|},
 \qquad
 \widehat\beta\ \pm\ 1.959964\,\widehat{\mathrm{SE}}_{\mathrm{corr}},
 \qquad
 L_T=\max\left\{h+1,\left\lfloor4(T/100)^{2/9}\right\rfloor\right\}.
\]
The main design gives \(L_T=5\). The OLS adjustment changes the standard error while leaving the point estimate unchanged.


\subsubsection{Monte Carlo design}\label{app:sim-monte-carlo}

Each setting has 1000 replications, indexed by \(s=0,\ldots,999\), with data-generation seed \(203609170+s\). Common random numbers pair the settings, and all estimators within a setting use the same observations. Six independent streams generate \(F_t\), \(E_t\), \(C_t\), \(\xi_t\), \(\nu_t\), and \(e_t\). Random fold assignment for PCR uses the data seed plus 1101. Changes in \(q\), \(r\), or \(\rho_F\) leave the oracle residual streams unchanged.

We report bias and RMSE relative to the structural effect \(\beta_h\), coverage of nominal 95\% intervals, and mean interval length. Under nuisance misspecification or exclusion failure, we also report coverage of the population residualization target. All replications enter the summaries, including extreme estimates under weak instruments; numerical failures terminate the computation. With 1000 independent replications, the binomial Monte Carlo standard error near 95\% coverage is about 0.007. Paired differences have corresponding Monte Carlo standard errors.

\subsection{Estimator comparisons} \label{app:subsec:sim-comparison}
\subsubsection{Competing estimators}\label{app:simulation-baselines}

The main shuffled iid DML--IV comparator uses the same training-only PCA rank selector, OLS learner, and ten folds as the proposed estimator, but assigns observations randomly to folds and applies no temporal buffer. Its standard error is the usual iid orthogonal-score expression
\[
 \widehat{\mathrm{SE}}_{\mathrm{iid}}
 =\frac{\{T^{-1}\widehat\Omega_0(\widehat\psi)\}^{1/2}}{|\widehat A|}.
\]
This applies the iid partially linear IV procedure \citep{chernozhukov2018double} on dependent observations. The nuisance learner and number of folds match the proposed estimator, so differences in coverage reflect fold construction and variance estimation.

We also compare Lasso and post-Lasso nuisance learners with ten contiguous folds and the same temporal buffer as the proposed estimator. Both use penalties \(0.75\widehat\sigma_a\sqrt{\log(2q)/m}\), where the response standard deviation \(\widehat\sigma_a\) and predictor standardization are computed on the training sample. Post-Lasso refits the selected coordinates by OLS and uses a training-mean prediction when none are selected. Both estimators use ordinary Bartlett HAC. Because the nuisance coefficients are dense, sparse-regression guarantees need not apply to these comparisons.

The conventional LP--IV baseline prespecifies the first four state coordinates, partials all three responses on those controls and an intercept by full-sample OLS, and construct intervals following \citep{stock2018identification}. Its exact population target is
\[
 \beta_h^{\mathrm{conv}}=\beta_h+
 \frac{g^{\mathsf T}V_4\gamma+\tau}{\pi+g^{\mathsf T}V_4d},
 \qquad
 V_4=I_r-\operatorname{Cov}(S_t,U_{Jt})
 \operatorname{Var}(U_{Jt})^{-1}\operatorname{Cov}(U_{Jt},S_t),
\]
where \(J=\{1,2,3,4\}\). In the main design this gives a target of 1.003048, compared with the structural value 0.485225. The target accounts for the covariance among the observed state coordinates.

The desparsified HDLP--IV ratio applies the modified desparsified Lasso of \citet{adamek2024local} separately to the outcome reduced form and treatment first stage. The external-IV ratio is our adaptation, since their paper considers an exogenous target regressor. \citet{breunig2020illposed} develop a related high-dimensional IV desparsification. In each regression, the instrument coefficient is unpenalized and the state coefficients are penalized. A shared nodewise Lasso of \(Z\) on \(U\) yields residual \(\widehat v_t\) and scale \(\widehat\tau_v^2=T^{-1}\sum_t\widehat v_t^2+\lambda_Z\|\widehat\delta\|_1\). With initial coefficients \((\widetilde\phi_a,\widetilde\gamma_a)\), the correction is
\[
 \widehat\phi_a^{\mathrm{db}}
 =\widetilde\phi_a+\frac{T^{-1}\sum_t\widehat v_t
 (a_t-\widetilde\phi_a Z_t-U_t^{\mathsf T}\widetilde\gamma_a)}{\widehat\tau_v^2},
 \qquad
 \widehat\beta^{\mathrm{dbIV}}=\frac{\widehat\phi_Y^{\mathrm{db}}}{\widehat\phi_W^{\mathrm{db}}}.
\]
All variables are demeaned and state predictors are standardized. The initial nuisance penalty is \(0.75\widehat\sigma_{M_Za}\sqrt{\log(2q)/T}\), and the nodewise penalty is \(0.75\widehat\sigma_Z\sqrt{\log(2q)/T}\). The multiplier is fixed at 0.75; the original procedure uses an iterative plug-in selector. Inference applies ordinary Bartlett HAC to the delta-method contribution \((\widehat\psi_{Y,t}-\widehat\beta^{\mathrm{dbIV}}\widehat\psi_{W,t})/\widehat\phi_W^{\mathrm{db}}\), where \(\widehat\psi_{a,t}=\widehat v_t\widehat e_{a,t}/\widehat\tau_v^2\). The oracle estimator uses the true conditional means and estimates the long-run variance by ordinary Bartlett HAC.

\begin{table}[!htbp]
\centering
\caption{Estimation and coverage under alternative nuisance learners and variance estimators.}
\label{tab:supp-literature-baselines}
\small
\setlength{\tabcolsep}{4pt}
\begin{tabular}{lrrrrr}
\toprule
Estimator & Bias & RMSE & SE/SD & Coverage & Length \\
\midrule
Proposed PCR--IV & 0.014 & 0.135 & 1.003 & 0.935 & 0.527 \\
Shuffled iid PCR--IV & -0.042 & 0.145 & 0.668 & 0.810 & 0.363 \\
Desparsified HDLP--IV ratio & 0.016 & 0.137 & 0.799 & 0.879 & 0.425 \\
Conventional LP--IV, 4 controls & 0.520 & 0.568 & 0.835 & 0.227 & 0.748 \\
Oracle-nuisance IV & -0.011 & 0.130 & 0.906 & 0.916 & 0.462 \\
Blocked Lasso IV & 0.057 & 0.145 & 0.895 & 0.871 & 0.469 \\
Blocked post-Lasso IV & 0.017 & 0.139 & 0.834 & 0.873 & 0.452 \\
Blocked PCR, ordinary HAC & 0.014 & 0.135 & 0.841 & 0.877 & 0.442 \\
Blocked PCR, iid SE & 0.014 & 0.135 & 0.611 & 0.757 & 0.321 \\
\bottomrule
\end{tabular}
\par\medskip\begin{minipage}{.97\linewidth}\footnotesize The last two rows share the proposed point estimate and change only its standard error. Both rows use blocked folds; the shuffled iid estimator appears separately. The blocked Lasso and post-Lasso estimators use penalty multiplier 0.75 and ordinary Bartlett HAC.\end{minipage}
\end{table}

\subsubsection{State adjustment and instrument use}\label{app:sim-adjustment}

The unadjusted IV estimator uses an intercept but no state controls. Its population shift is \((g^{\mathsf T}\gamma+\tau)/(\pi+g^{\mathsf T}d)\). The adjusted OLS estimator replaces the instrument with treatment in the blocked PCR procedure, using the corresponding OLS sensitivity adjustment with Bartlett HAC. Its baseline shift is \(0.8^2/(\pi^2+0.8^2+0.7^2)\). Unadjusted OLS has shift \((d^{\mathsf T}\gamma+0.8^2)/(d^{\mathsf T}d+\pi^2+0.8^2+0.7^2)\). These estimators isolate the contributions of state adjustment and instrumental-variable identification.

\begin{table}[!htbp]
\centering
\caption{The roles of state adjustment and instrumental-variable identification.}
\label{tab:supp-adjustment-ablations}
\small
\setlength{\tabcolsep}{4pt}
\begin{tabular}{lrrrrr}
\toprule
Estimator & Target shift & Bias & RMSE & Cov. $\beta$ & Cov. ref. \\
\midrule
Proposed PCR--IV & 0.000 & 0.014 & 0.135 & 0.935 & 0.935 \\
Conventional LP--IV, 4 controls & 0.518 & 0.520 & 0.568 & 0.227 & 0.900 \\
IV without state adjustment & 0.498 & 0.500 & 0.538 & 0.170 & 0.919 \\
PCR--OLS without instrument & 0.430 & 0.423 & 0.427 & 0.000 & 0.942 \\
OLS without state adjustment & 0.345 & 0.347 & 0.352 & 0.001 & 0.899 \\
\bottomrule
\end{tabular}
\par\medskip\begin{minipage}{.97\linewidth}\footnotesize Reference coverage uses the exact population residualization target. Dense nuisance coefficients may violate the conditions for sparse-regression consistency.\end{minipage}
\end{table}

\subsubsection{Sample size and instrument feedback}\label{app:sim-sample-size-feedback}

We vary \(T\) over 250, 500, 1000, and 2000 while retaining \(q=1000\) and six factors. The feedback comparison uses \(a=0,.5,.9\) at the main sample size. These are paired comparisons, and both the proposed and shuffled estimates use the same PCA learner throughout.

\begin{table}[!htbp]
\centering
\caption{Sample-size scaling with $q=1000$ and six factors.}
\label{tab:supp-scaling-results}
\small
\setlength{\tabcolsep}{4pt}
\begin{tabular}{llrrrr}
\toprule
$T$ & Method & Bias & $\sqrt T\,\mathrm{SD}$ & $\sqrt T\,\overline{\mathrm{SE}}$ & Coverage \\
\midrule
250 & Proposed & 0.020 & 3.294 & 3.344 & 0.940 \\
250 & Shuffled iid & -0.091 & 3.400 & 2.323 & 0.842 \\
250 & Oracle & -0.025 & 3.020 & 2.643 & 0.912 \\
500 & Proposed & 0.014 & 2.994 & 3.004 & 0.935 \\
500 & Shuffled iid & -0.042 & 3.104 & 2.073 & 0.810 \\
500 & Oracle & -0.011 & 2.907 & 2.633 & 0.916 \\
1000 & Proposed & 0.006 & 3.035 & 2.852 & 0.933 \\
1000 & Shuffled iid & -0.022 & 3.081 & 1.972 & 0.795 \\
1000 & Oracle & -0.008 & 3.006 & 2.649 & 0.916 \\
2000 & Proposed & 0.003 & 2.943 & 2.783 & 0.936 \\
2000 & Shuffled iid & -0.011 & 2.967 & 1.927 & 0.794 \\
2000 & Oracle & -0.004 & 2.902 & 2.674 & 0.929 \\
\bottomrule
\end{tabular}
\end{table}
\begin{table}[!htbp]
\centering
\caption{Estimation and coverage across instrument-feedback strengths.}
\label{tab:supp-feedback-results}
\small
\setlength{\tabcolsep}{4pt}
\begin{tabular}{llrrrr}
\toprule
$a$ & Method & Bias & RMSE & Coverage & Length \\
\midrule
0 & Proposed & -0.006 & 0.142 & 0.940 & 0.536 \\
0 & Shuffled iid & -0.007 & 0.135 & 0.798 & 0.342 \\
0 & Oracle & -0.005 & 0.132 & 0.915 & 0.455 \\
0.5 & Proposed & 0.005 & 0.140 & 0.942 & 0.530 \\
0.5 & Shuffled iid & -0.026 & 0.142 & 0.791 & 0.353 \\
0.5 & Oracle & -0.008 & 0.134 & 0.916 & 0.457 \\
0.9 & Proposed & 0.014 & 0.135 & 0.935 & 0.527 \\
0.9 & Shuffled iid & -0.042 & 0.145 & 0.810 & 0.363 \\
0.9 & Oracle & -0.011 & 0.130 & 0.916 & 0.462 \\
\bottomrule
\end{tabular}
\end{table}

\subsection{Sensitivity analyses} \label{app:subsec:sim-sensitivity}
\subsubsection{Nonlinear nuisance functions}\label{app:sim-nonlinear}

For the nonlinear specifications, define
\[
 f_W(S)=\frac{S_1^2-1}{\sqrt2},\qquad
 f_Z(S)=\frac{S_2^2-1}{\sqrt2},\qquad
 f_Y(S)=\frac{f_Z(S)+S_3S_4}{\sqrt2}.
\]
We add \(\delta_W f_W(S_t)\) to the treatment mean, \(\delta_Z f_Z(S_t)\) to the instrument mean, and \(\delta_Y f_Y(S_t)\) to the structural outcome nuisance. The total outcome mean also incorporates \(\beta_h\delta_W f_W(S_t)\). The residual disturbances and instrument validity are unchanged. At each time the indices are independent standard normals, so the three functions have mean zero and variance one and are orthogonal to all linear functions of \(U_t\). Moreover, \(\mathbb E(f_Zf_Y)=1/\sqrt2\) and \(\mathbb E(f_Zf_W)=0\). Population linear residualization therefore targets
\[
 \beta_h^{\mathrm{lin}}=\beta_h+\frac{\delta_Z\delta_Y}{\sqrt2\pi}.
\]
The grid includes \((\delta_W,\delta_Z,\delta_Y)=(.5,0,.5),(0,.5,0)\), and common magnitudes \(.25,.5,1\). The first two isolate nuisance errors for which the displayed population shift vanishes.

The nonlinear learner uses the selected principal-component scores, every pairwise product of those scores including squares, and an intercept. It estimates all coefficients by OLS on the training fold. With six selected factors this gives 28 coefficients per nuisance including the intercept. The complete quadratic space is invariant under rotation of the selected factor coordinates. Both the factor space and the regression coefficients are estimated, so factor-estimation error remains part of the comparison. Blocked quadratic PCR uses the same OLS sensitivity adjustment and Bartlett rule, while shuffled quadratic PCR uses ordinary iid score variance. At the moderate common magnitude, the blocked quadratic estimator has bias 0.037, RMSE 0.153, and coverage 0.948, compared with -0.091, 0.180, and 0.764 for its shuffled iid counterpart. The quadratic results vary little across nonlinear magnitudes. This is consistent with accurate factor recovery and close approximation of the added terms by the full quadratic dictionary.

\begin{table}[!htbp]
\centering
\caption{Nonlinear nuisance sensitivity.}
\label{tab:supp-misspec-results}
\small
\setlength{\tabcolsep}{4pt}
\begin{tabular}{llrrrr}
\toprule
$(\delta_W,\delta_Z,\delta_Y)$ & Method & Bias & RMSE & Coverage & Bias ref. \\
\midrule
(0.5,0,0.5) & Linear PCR & 0.016 & 0.157 & 0.934 & 0.016 \\
(0.5,0,0.5) & Quadratic PCR & 0.037 & 0.153 & 0.948 & 0.037 \\
(0.5,0,0.5) & iid quadratic & -0.091 & 0.180 & 0.765 & -0.091 \\
(0.5,0,0.5) & Oracle & -0.011 & 0.130 & 0.916 & -0.011 \\
(0,0.5,0) & Linear PCR & 0.013 & 0.152 & 0.945 & 0.013 \\
(0,0.5,0) & Quadratic PCR & 0.037 & 0.153 & 0.948 & 0.037 \\
(0,0.5,0) & iid quadratic & -0.091 & 0.180 & 0.764 & -0.091 \\
(0,0.5,0) & Oracle & -0.011 & 0.130 & 0.916 & -0.011 \\
(0.25,0.25,0.25) & Linear PCR & 0.091 & 0.167 & 0.894 & 0.017 \\
(0.25,0.25,0.25) & Quadratic PCR & 0.037 & 0.153 & 0.948 & 0.037 \\
(0.25,0.25,0.25) & iid quadratic & -0.091 & 0.180 & 0.764 & -0.091 \\
(0.25,0.25,0.25) & Oracle & -0.011 & 0.130 & 0.916 & -0.011 \\
(0.5,0.5,0.5) & Linear PCR & 0.320 & 0.381 & 0.598 & 0.026 \\
(0.5,0.5,0.5) & Quadratic PCR & 0.037 & 0.153 & 0.948 & 0.037 \\
(0.5,0.5,0.5) & iid quadratic & -0.091 & 0.180 & 0.764 & -0.091 \\
(0.5,0.5,0.5) & Oracle & -0.011 & 0.130 & 0.916 & -0.011 \\
(1,1,1) & Linear PCR & 1.295 & 1.839 & 0.268 & 0.116 \\
(1,1,1) & Quadratic PCR & 0.037 & 0.153 & 0.948 & 0.037 \\
(1,1,1) & iid quadratic & -0.091 & 0.180 & 0.765 & -0.091 \\
(1,1,1) & Oracle & -0.011 & 0.130 & 0.916 & -0.011 \\
\bottomrule
\end{tabular}
\par\medskip\begin{minipage}{.97\linewidth}\footnotesize Linear-PCR reference shift is $\delta_Z\delta_Y/(\sqrt2\pi)$. Quadratic PCR uses all pairwise products of estimated factor scores, including squares, learned within each training fold. Its reference target and the oracle target equal $\beta_4$.\end{minipage}
\end{table}

\subsubsection{Instrument quality}\label{app:sim-instrument-quality}

We vary relevance over \(\pi=.15,.30,.60\) under exclusion, and vary \(\tau=0,.03,.06,.12\) at \(\pi=.60\). Under exclusion failure, even oracle residualization targets \(\beta_h+\tau/\pi\). We report coverage of both the structural effect and the displaced IV target. The Wald procedures are not weak-IV robust. At the weakest first stage, RMSE is 1.026 for the proposed estimator and 20.757 for shuffled iid PCR. Their largest absolute errors are 17.272 and 603.309, respectively. The extreme ratios make RMSE and interval-length summaries unstable, so high coverage at this setting should be interpreted with caution.

\begin{table}[!htbp]
\centering
\caption{Instrument relevance and exclusion sensitivity.}
\label{tab:supp-instrument-quality}
\small
\setlength{\tabcolsep}{4pt}
\begin{tabular}{llrrrr}
\toprule
$(\pi,\tau)$ & Method & Bias & RMSE & Cov. $\beta$ & Cov. ref. \\
\midrule
(0.15,0) & Proposed & -0.015 & 1.025 & 0.945 & 0.945 \\
(0.15,0) & Shuffled iid & -0.922 & 20.757 & 0.913 & 0.913 \\
(0.15,0) & Oracle & -0.100 & 2.155 & 0.945 & 0.945 \\
(0.3,0) & Proposed & 0.018 & 0.275 & 0.937 & 0.937 \\
(0.3,0) & Shuffled iid & -0.103 & 0.331 & 0.861 & 0.861 \\
(0.3,0) & Oracle & -0.033 & 0.278 & 0.927 & 0.927 \\
(0.6,0) & Proposed & 0.014 & 0.135 & 0.935 & 0.935 \\
(0.6,0) & Shuffled iid & -0.042 & 0.145 & 0.810 & 0.810 \\
(0.6,0) & Oracle & -0.011 & 0.130 & 0.916 & 0.916 \\
(0.6,0.03) & Proposed & 0.063 & 0.147 & 0.901 & 0.935 \\
(0.6,0.03) & Shuffled iid & 0.010 & 0.137 & 0.802 & 0.806 \\
(0.6,0.03) & Oracle & 0.040 & 0.134 & 0.898 & 0.917 \\
(0.6,0.06) & Proposed & 0.113 & 0.173 & 0.831 & 0.933 \\
(0.6,0.06) & Shuffled iid & 0.062 & 0.149 & 0.724 & 0.799 \\
(0.6,0.06) & Oracle & 0.090 & 0.155 & 0.806 & 0.918 \\
(0.6,0.12) & Proposed & 0.211 & 0.247 & 0.599 & 0.934 \\
(0.6,0.12) & Shuffled iid & 0.167 & 0.212 & 0.472 & 0.784 \\
(0.6,0.12) & Oracle & 0.192 & 0.228 & 0.570 & 0.917 \\
\bottomrule
\end{tabular}
\par\medskip\begin{minipage}{.97\linewidth}\footnotesize The reference target is $\beta_4+\tau/\pi$. All draws are retained, including extreme weak-instrument ratios. These Wald intervals are not weak-IV robust.\end{minipage}
\end{table}
\begin{figure}[!htbp]
 \centering
 \includegraphics[width=\textwidth]{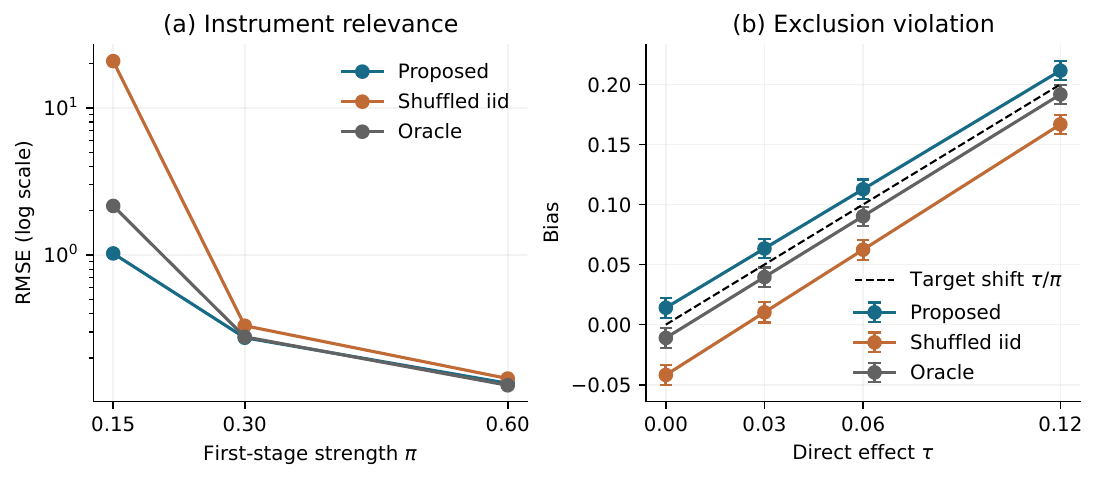}
 \caption{Instrument relevance and exclusion under the factor model. Panel (a) uses a logarithmic RMSE scale and retains all weak-instrument draws. Bars in panel (b) are 95\% Monte Carlo intervals for bias; the dashed line is the exact exclusion-induced target shift.}
 \label{fig:supp-instrument-quality}
\end{figure}

\subsubsection{Dimension, persistence, horizon, and buffer}\label{app:sim-dimension-horizon-buffer}

The dimension grid uses \(q=120,400,1000\), holding the factor count and joint nuisance signal covariance fixed. Increasing \(q\) adds noisy measurements of the factors. A larger panel can therefore improve factor estimation. Factor-count sensitivity uses \(r=3,6,10\), factor persistence uses \(\rho_F=.2,.5,.8\), and instrument persistence uses \(\rho_Z=.1,.5,.8\). The rank search is 0 through 10 in every setting, with ten factors at the upper boundary. The proposed estimator's coverage is 0.923 with three factors, 0.936 with ten factors, and 0.926--0.927 at the two lower factor persistences. Coverage remains below nominal even though every training fold selects the correct rank. At instrument persistence 0.1, the proposed and iid RMSE values are also close, 0.096 and 0.095.

\begin{table}[!htbp]
\centering
\caption{Dimension, factor count, and persistence sensitivity.}
\label{tab:supp-dimension-results}
\small
\setlength{\tabcolsep}{4pt}
\begin{tabular}{llrrrr}
\toprule
Setting & Method & Bias & RMSE & Coverage & Rank correct \\
\midrule
$q=120$ & Proposed & 0.014 & 0.135 & 0.936 & 1.000 \\
$q=120$ & Shuffled iid & -0.042 & 0.145 & 0.806 & 1.000 \\
$q=400$ & Proposed & 0.014 & 0.135 & 0.934 & 1.000 \\
$q=400$ & Shuffled iid & -0.042 & 0.145 & 0.808 & 1.000 \\
$q=1000$ & Proposed & 0.014 & 0.135 & 0.935 & 1.000 \\
$q=1000$ & Shuffled iid & -0.042 & 0.145 & 0.810 & 1.000 \\
$r=3$ & Proposed & 0.004 & 0.136 & 0.923 & 1.000 \\
$r=3$ & Shuffled iid & -0.031 & 0.141 & 0.816 & 1.000 \\
$r=10$ & Proposed & 0.026 & 0.142 & 0.936 & 1.000 \\
$r=10$ & Shuffled iid & -0.057 & 0.156 & 0.815 & 1.000 \\
$\rho_F=0.2$ & Proposed & -0.002 & 0.132 & 0.927 & 1.000 \\
$\rho_F=0.2$ & Shuffled iid & -0.022 & 0.136 & 0.823 & 1.000 \\
$\rho_F=0.5$ & Proposed & 0.003 & 0.132 & 0.926 & 1.000 \\
$\rho_F=0.5$ & Shuffled iid & -0.027 & 0.138 & 0.823 & 1.000 \\
$\rho_Z=0.1$ & Proposed & 0.015 & 0.096 & 0.946 & 1.000 \\
$\rho_Z=0.1$ & Shuffled iid & -0.017 & 0.095 & 0.931 & 1.000 \\
$\rho_Z=0.5$ & Proposed & 0.016 & 0.115 & 0.942 & 1.000 \\
$\rho_Z=0.5$ & Shuffled iid & -0.029 & 0.119 & 0.871 & 1.000 \\
\bottomrule
\end{tabular}
\par\medskip\begin{minipage}{.97\linewidth}\footnotesize Each row changes the indicated parameter from the main design. Rank correct is the fraction of training folds selecting the true factor count, averaged over replications.\end{minipage}
\end{table}

Horizon sensitivity uses \(h=0,4,8,12\), with \(\beta_h=0.8e^{-h/8}\) and the prescribed horizon-specific buffer and bandwidth. At horizon zero the shuffled iid interval has coverage 0.948, compared with 0.934 for the proposed interval; the coverage ordering is therefore also not uniform across horizons. At horizons 8 and 12, proposed coverage is 0.924 and 0.929, compared with 0.732 and 0.699 for shuffled iid PCR. At horizon 12 the RMSE values are close, 0.180 and 0.179. We also vary the buffer over 0, 5, 7, 28, and 56 under the main DGP, using paired observations. The zero-buffer comparison falls outside the asymptotic separation requirement.

\begin{table}[!htbp]
\centering
\caption{Horizon sensitivity with unit-variance normalized moving-average innovations.}
\label{tab:supp-horizon-results}
\small
\setlength{\tabcolsep}{4pt}
\begin{tabular}{llrrrr}
\toprule
$h$ & Method & Bias & RMSE & Coverage & Length \\
\midrule
0 & Proposed & 0.015 & 0.093 & 0.934 & 0.370 \\
0 & Shuffled iid & -0.042 & 0.109 & 0.948 & 0.366 \\
0 & Oracle & -0.011 & 0.092 & 0.942 & 0.337 \\
4 & Proposed & 0.014 & 0.135 & 0.935 & 0.527 \\
4 & Shuffled iid & -0.042 & 0.145 & 0.810 & 0.363 \\
4 & Oracle & -0.011 & 0.130 & 0.916 & 0.462 \\
8 & Proposed & 0.014 & 0.161 & 0.924 & 0.616 \\
8 & Shuffled iid & -0.043 & 0.166 & 0.732 & 0.361 \\
8 & Oracle & -0.012 & 0.150 & 0.914 & 0.522 \\
12 & Proposed & 0.015 & 0.180 & 0.929 & 0.674 \\
12 & Shuffled iid & -0.043 & 0.179 & 0.699 & 0.359 \\
12 & Oracle & -0.012 & 0.164 & 0.908 & 0.558 \\
\bottomrule
\end{tabular}
\end{table}
\begin{table}[!htbp]
\centering
\caption{Buffer length in blocked PCR with OLS-adjusted Bartlett HAC.}
\label{tab:supp-buffer-results}
\small
\setlength{\tabcolsep}{4pt}
\begin{tabular}{lrrrr}
\toprule
Buffer & Bias & RMSE & Coverage & Length \\
\midrule
0 & 0.011 & 0.134 & 0.929 & 0.524 \\
5 & 0.014 & 0.134 & 0.933 & 0.526 \\
7 & 0.014 & 0.135 & 0.935 & 0.527 \\
28 & 0.016 & 0.137 & 0.937 & 0.537 \\
56 & 0.018 & 0.141 & 0.937 & 0.557 \\
\bottomrule
\end{tabular}
\par\medskip\begin{minipage}{.97\linewidth}\footnotesize The main rule gives buffer 7. A zero buffer falls outside the asymptotic separation requirement.\end{minipage}
\end{table}

\subsubsection{Two-outcome inference}\label{app:sim-vector}

For joint inference, we use the main treatment and instrument with a second outcome
\begin{align*}
 Y_{2,t+4}&=-.35W_t+\gamma_2^{\mathsf T}S_t-.5C_t
 +\frac{.8}{\sqrt5}\sum_{j=0}^4e^{(2)}_{t+j},\\
 (\gamma_2)_j&=.6\cos\{\pi(j-1/2)/6\},\qquad j=1,\ldots,6.
\end{align*}
The second outcome uses an independent innovation stream initialized from the integer pair \((\mathrm{seed},90901)\) through \texttt{SeedSequence}. Both outcomes share the fold-specific PCA spaces. We form one OLS-adjusted score per outcome for the proposed estimator and estimate their full Bartlett covariance matrix in chronological order. The iid comparator uses only the contemporaneous covariance of its two ordinary scores. Component intervals use normal critical values, and the joint Wald region uses the 0.95 quantile of \(\chi^2_2\). The joint Wald statistic requires a positive-definite estimated covariance matrix.

\begin{table}[!htbp]
\centering
\caption{Two-outcome inference with one treatment and one instrument.}
\label{tab:supp-vector-results}
\small
\setlength{\tabcolsep}{4pt}
\begin{tabular}{llrrrr}
\toprule
Component & Method & Bias & RMSE & Marginal cov. & Joint cov. \\
\midrule
1 & Proposed & 0.014 & 0.135 & 0.935 & 0.917 \\
2 & Proposed & -0.005 & 0.126 & 0.931 & 0.917 \\
1 & Shuffled iid & -0.042 & 0.145 & 0.810 & 0.661 \\
2 & Shuffled iid & 0.029 & 0.130 & 0.748 & 0.661 \\
1 & Oracle & -0.011 & 0.130 & 0.916 & 0.894 \\
2 & Oracle & 0.008 & 0.119 & 0.916 & 0.894 \\
\bottomrule
\end{tabular}
\par\medskip\begin{minipage}{.97\linewidth}\footnotesize The target is $(0.485225,-0.35)^\mathsf T$. Joint coverage uses the 0.95 quantile of $\chi^2_2$ and is repeated across the two component rows.\end{minipage}
\end{table}

\section{Empirical design and additional diagnostics}\label{app:empirical-details}

\subsection{Variable construction and estimation} \label{app:subsec:estimation}

The baseline application combines the April 2026 FRED-MD vintage \citep{mccracken2016fred} with the updated Romer--Romer shock \citep{romer2004new,wieland2021updated}. The nonmissing instrument window contains 468 monthly policy observations from January 1969 through December 2007. For each FRED-MD series, we apply the transformation code (\texttt{tcode}) supplied with this vintage, following the transformation convention of \citet{mccracken2016fred} and the accompanying FRED-MD data appendix. The data appendix and the code implementing these transformations are available at \url{https://www.stlouisfed.org/research/economists/mccracken/fred-databases}. \Cref{tab:application-variables} records the codes and transformations for the outcomes and treatment.
In particular, the CPI outcome is inflation acceleration, industrial production is monthly growth, and housing starts are in log levels.

\begin{table}[H]
    \centering
    \caption{Variables in the monetary-policy application.}
    \label{tab:application-variables}
    \begin{tabularx}{\textwidth}{ll>{\raggedright\arraybackslash}X}
        \toprule
        Object & Series & Construction \\
        \midrule
        CPI outcome & \texttt{CPIAUCSL} & Second difference of log CPI (code 6) \\
        Unemployment outcome & \texttt{UNRATE} & Monthly first difference (code 2) \\
        Industrial-production outcome & \texttt{INDPRO} & Monthly first difference of log production (code 5) \\
        Housing-starts outcome & \texttt{HOUST} & Log level (code 4) \\
        Treatment & \texttt{GS1} & Monthly first difference in percentage points (code 2) \\
        Instrument & \texttt{resid\_full} & Updated Romer--Romer residual \\
        \bottomrule
    \end{tabularx}
\end{table}

Let \(X_t\in\mathbb R^{122}\) contain the retained transformed FRED-MD series. The pre-treatment state is
\[
    U_t=\big(X_{t-1}^\top,X_{t-2}^\top,\ldots,X_{t-12}^\top\big)^\top
    \in\mathbb R^{1464}.
\]
We exclude \texttt{ACOGNO}, \texttt{ANDENOx}, \texttt{TWEXAFEGSMTHx}, and \texttt{UMCSENTx} because their lag histories are incomplete over the instrument window. All controls precede treatment by at least one month, so no contemporaneous outcomes or post-treatment variables are leaked into the state $\mathbf U_t$.

For each outcome and horizon, \(\beta_h\) is the corresponding scalar entry of the causal response matrix in \Cref{eq:model}. Its interpretation as a response to a one-percentage-point instrumented increase in \(\Delta\texttt{GS1}_t\) requires the linear response model and the conditional IV restrictions in \Cref{ass:iv}. The application cannot test the exclusion restriction or establish the causal validity of the updated Romer--Romer shock by itself, so its validity is endorsed by the standard practice in macroeconometics. The normalization is a one-percentage-point yield change.

We separately estimate the conditional means of \(Y_{t+h}\), \(W_t\), and \(Z_t\) by Lasso with five contiguous validation folds. At horizon \(h\), each training sample excludes its validation block and \(b_h=\max\{12,h+1\}\) observations on either side. With the resulting cross-fitted residuals, the estimate is
\[
    \widehat\beta_h
    =\frac{\sum_t\widehat e_{Z,t}\widehat e_{Y,t,h}}
           {\sum_t\widehat e_{Z,t}\widehat e_{W,t}}.
\]
This is the scalar specialization of \Cref{alg:orthogonal-lpiv}. We use a Bartlett HAC estimator with truncation lag \(b_h\) for both the score variance
and the residualized first-stage diagnostic. The displayed normal 95\% intervals are pointwise. 

\subsection{Sensitivity analyses and calibration} \label{app:subsec:real-sensitivity}

\subsubsection{Nuisance specifications and response sensitivity}\label{app:empirical-nuisance-sensitivity}

We compare the baseline Lasso learner with ridge regression and regression on eight principal components of the 1,464-dimensional state. All three learners use the same 468 observations, five contiguous folds, temporal buffer, and HAC lag. We standardize controls within each training fold. Lasso selects among 30 penalties by three-fold cross-validation; ridge uses three-fold cross-validation over 25 logarithmically spaced penalties from \(10^{-4}\) to \(10^4\). The factor learner extracts eight principal components from the standardized training controls and regresses each nuisance response on those components. 

The fourth specification uses Lasso on 60 controls, comprising 12 lags of the treatment and the four outcomes. It retains the same five-fold estimation and inference protocol. These comparisons cover every horizon from 0 through 24 for every outcome.

\Cref{tab:application-nuisance-sensitivity} compares the housing responses at six and twelve months. Every specification gives negative point estimates and intervals below zero at both horizons. At six months, the estimates are close, ranging from \(-0.155\) to \(-0.138\). At twelve months, the range widens to \([-0.230,-0.168]\). Factor adjustment gives a deeper and more persistent housing contraction than the baseline Lasso fit. The largest difference between these two housing paths is 0.078 log points at month 11.

\begin{table}[htbp]
\centering
\caption{Housing responses under alternative nuisance specifications. The intervals are pointwise 95\% Bartlett HAC intervals.}
\label{tab:application-nuisance-sensitivity}
\small
\setlength{\tabcolsep}{5pt}
\begin{tabular}{lrrrr}
\toprule
& \multicolumn{2}{c}{Six months} & \multicolumn{2}{c}{Twelve months} \\
Nuisance learner & Estimate & 95\% interval & Estimate & 95\% interval \\
\midrule
Lasso, 1464 controls & -0.155 & $[-0.260,-0.050]$ & -0.168 & $[-0.309,-0.027]$ \\
Eight factors & -0.145 & $[-0.255,-0.035]$ & -0.208 & $[-0.364,-0.052]$ \\
Ridge, 1464 controls & -0.139 & $[-0.235,-0.044]$ & -0.198 & $[-0.378,-0.017]$ \\
Lasso, 60 controls & -0.138 & $[-0.202,-0.074]$ & -0.230 & $[-0.400,-0.060]$ \\
\bottomrule
\end{tabular}
\end{table}

The Lasso and factor intervals overlap at all 100 outcome-horizon pairs, although conclusions about individual horizons can differ. For CPI inflation acceleration, 22 of the 25 intervals include zero under each learner. For industrial-production growth, the counts are 21 with Lasso and 20 with factors; for unemployment, they are 23 and 20. At twelve months, the unemployment response is \(0.051\) percentage points with Lasso, with interval \([-0.013,0.114]\), compared with \(0.063\) and interval \([0.018,0.108]\) under factor adjustment. Thus the nuisance specification affects the estimated magnitude and duration, as well as pointwise significance, while preserving the main housing pattern. These interval comparisons are descriptive and do not test equality of the response paths.

\subsubsection{HAC lag and policy-rate sensitivity}\label{app:empirical-hac-policy-rate}

With the baseline Lasso nuisance fits, varying the HAC truncation lag over 6, 12, 18, and 24 months, together with the horizon-dependent baseline choice, preserves the negative six- and twelve-month housing intervals. At six months, the HAC standard error ranges from 0.0498 to 0.0598; at twelve months, it ranges from 0.0711 to 0.0773. The point estimates are unchanged because only variance estimation is affected. 

Replacing \(\Delta\texttt{GS1}_t\) with the monthly change in the Wu--Xia shadow rate \citep{wu-xia} leaves 215 overlapping months, from February 1990 through December 2007, with the same 1,464 controls. Using the same five-fold Lasso procedure over all 25 horizons gives a median residualized first-stage HAC \(F\)-statistic of 1.03, with a range from 0.12 to 2.67. This alternative policy-rate measure has a weak first stage and may require weak-IV-robust inference for a reliable comparison of response paths.

\subsubsection{Plasmode design and calibration}\label{app:empirical-plasmode}

We use a plasmode experiment to assess estimation accuracy and interval coverage at the sample size and dimension of the empirical application. The experiment keeps the observed treatment \(W_t\), instrument \(Z_t\), and 1464-dimensional state \(U_t\) fixed for the 468 policy months. We replace the observed housing outcomes with synthetic outcomes whose treatment coefficient is specified in advance,
\[
    Y^{\ast}_{t+h}=\beta_h^{\ast}W_t+m_h(U_t)+e^{\ast}_{t+h}.
\]
We evaluate horizons \(h\in\{0,6,12,18,24\}\) under two response paths. The zero-effect path sets \(\beta_h^{\ast}=0\). The housing-like path uses the baseline Lasso estimates of the housing response, rounded to three decimal places. Specifically, it sets \(\beta_h^{\ast}\) to \(-0.010\), \(-0.155\), \(-0.168\), \(-0.108\), and \(-0.108\) at these horizons, respectively.

The function \(m_h(U_t)\) supplies the part of the synthetic outcome driven by the observed state. At each horizon, we standardize the state coordinates and select the ten with the largest absolute sample correlations with the observed housing outcome \(Y_{t+h}\). We sum these coordinates, weighting each by the sign of its correlation, and rescale the sum to have half the sample standard deviation of \(Y_{t+h}\). This gives a sparse linear nuisance signal, which we hold fixed across replications.

To construct the synthetic errors, we first form an empirical residual pool,
\[
    \widetilde\beta_h=\frac{\sum_t Z_tY_{t+h}}{\sum_t Z_tW_t},
    \qquad
    r_{t,h}=Y_{t+h}-\widetilde\beta_hW_t-m_h(U_t).
\]
The unadjusted IV ratio \(\widetilde\beta_h\) is a pilot treatment coefficient used to form this pool. We center the residuals and sample blocks of 12 consecutive values with replacement. We then concatenate the blocks and retain the first 468 entries as \(e^{\ast}_{t+h}\). Sampling whole blocks retains dependence within each block. The bootstrap therefore supplies the errors for a new synthetic outcome sample in each replication.

For each response path and horizon, we run 200 replications. We fit Lasso nuisance regressions with three contiguous validation folds and temporal buffer \(b_h=\max\{12,h+1\}\). Treatment and instrument residuals are computed once per horizon, while the synthetic outcome is re-residualized in every replication. Each replication produces a pointwise 95\% interval using Bartlett HAC standard errors with truncation lag \(b_h\) and normal critical values. We assess coverage by checking whether these intervals contain the imposed coefficient \(\beta_h^{\ast}\).

Under the housing-like path, coverage ranges from 0.955 to 0.990. Absolute bias is below 0.04 log points at every horizon, compared with median HAC standard errors between 0.034 and 0.065. Under the zero-effect path, the nominal 5\% test falsely rejects with frequency between 0.010 and 0.060. These results support the procedure's ability to recover imposed dynamic effects with small bias and generally conservative pointwise intervals at the dimensions of the application. This evidence is conditional on the chosen sparse nuisance signal and resampled errors; the causal interpretation of the empirical housing response still depends on the IV assumptions.

\putbib
\clearpage
\end{bibunit}
\end{document}